\documentclass[11pt]{article}

\usepackage{hyperref}
\usepackage{amsmath} 
\usepackage{amsthm} 
\usepackage{amssymb}	
\usepackage{graphicx} 
\usepackage{multicol} 
\usepackage{multirow}
\usepackage{color}
\usepackage[dvips,letterpaper,margin=1in,bottom=1in]{geometry}
\usepackage[capitalize,noabbrev]{cleveref}
\usepackage{orcidlink} 

\usepackage[utf8]{inputenc}
\usepackage[english]{babel}

\usepackage{diagbox}
\usepackage{mathtools}
\usepackage{ytableau}

\newtheorem{theorem}{Theorem}[section]

\newtheorem{lemma}[theorem]{Lemma}
\newtheorem{corollary}[theorem]{Corollary}
\newtheorem{proposition}[theorem]{Proposition}
\newtheorem{fact}[theorem]{Fact}

\newtheorem{definition}{Definition}[section]

\newtheorem{remark}{Remark}[section]
\newtheorem{problem}{Problem}
\newtheorem{question}{Question}
\crefname{question}{Question}{Questions}

\newcommand{\braket}[2]{\left< #1 \vphantom{#2} \middle| #2 \vphantom{#1} \right>} 

\DeclarePairedDelimiter\rbra{\lparen}{\rparen}
\DeclarePairedDelimiter\sbra{\lbrack}{\rbrack}
\DeclarePairedDelimiter\cbra{\{}{\}}
\DeclarePairedDelimiter\abs{\lvert}{\rvert}
\DeclarePairedDelimiter\Abs{\lVert}{\rVert}

\DeclarePairedDelimiter\ket{\lvert}{\rangle}
\DeclarePairedDelimiter\bra{\langle}{\rvert}

\newcommand{\set}[2] {\left\{\, #1 \colon #2 \,\right\}}

\newcommand{\tr} {\operatorname{tr}}

\newcommand{\polylog} {\operatorname{polylog}}
\newcommand{\rank} {\operatorname{rank}}

\newcommand{\spanspace} {\operatorname{span}}

\newcommand{\HA}[0]{\mathcal{H}_{\textup{A}}}
\newcommand{\HB}[0]{\mathcal{H}_{\textup{B}}}
\newcommand{\HAB}[0]{\mathcal{H}_{\textup{AB}}}

\newcommand{\ketbra}[2]{\ensuremath{\ket{#1}\!\bra{#2}}}
\renewcommand{\braket}[2]{\ensuremath{\langle {#1} \vert {#2} \rangle}}
\newcommand{\kett}[1]{|#1\rangle\!\rangle}
\newcommand{\bbra}[1]{\langle\!\langle#1|}
\newcommand{\kettbbra}[2]{\ensuremath{\kett{#1}\!\bbra{#2}}}
\newcommand{\bbrakett}[2]{\ensuremath{\langle\!\langle{#1}\vert{#2}\rangle\!\rangle}}

\newcommand{\TA}[0]{\textup{A}}
\newcommand{\TB}[0]{\textup{B}}
\newcommand{\TAB}[0]{\textup{AB}}
\allowdisplaybreaks

\usepackage{latexsym}
\usepackage{CJK}

\usepackage{enumerate}

\usepackage{algorithm}
\usepackage{algpseudocode}

\usepackage{stmaryrd}
\usepackage{tabularx}
\usepackage{booktabs}

\usepackage{tikz}
\usetikzlibrary{quantikz2}

\begin{document}

\title{Local Test for Unitarily Invariant Properties \\ of Bipartite Quantum States}

\author{
    Kean Chen\thanks{Kean Chen was with the Institute of Software, Chinese Academy of Sciences, Beijing 100190, China. He is now with the University of Pennsylvania, Philadelphia, PA 19104, USA (e-mail: \url{keanchen.gan@gmail.com}).}\and
    Qisheng Wang\thanks{Qisheng Wang is with the School of Computer Science, Shanghai Jiao Tong University, Shanghai 200240, China (e-mail: \url{QishengWang1994@gmail.com}).}\and
    Zhicheng Zhang\thanks{Zhicheng Zhang is with the School of Information and Communication Technology, Griffith University, Brisbane, QLD 4111, Australia (e-mail: \url{iszczhang@gmail.com}).}
}
\date{}
\maketitle

\begin{abstract}
    We study the power of local test for bipartite quantum states.  
    Our central result is that,
    for properties of bipartite pure states,
    unitary invariance on one part implies an \textit{optimal} (over all global testers) local tester acting only on the other part.
    As an application, we demonstrate
    \begin{itemize}
        \item Purified samples offer no advantage in property testing of mixed states.
        \item A matching lower bound $\Omega(r^2/\varepsilon^2)$ for testing the Schmidt rank of bipartite states with perfect completeness, settling an open question raised in the survey of \hyperlink{cite.MdW16}{Montanaro and de Wolf (ToC 2016)}.
        \item A lower bound $\Omega((\sqrt{n}+\sqrt{r})\cdot\sqrt{r}/\varepsilon^2)$ for testing whether an $n$-partite state is a matrix product state of bond dimension $r$ or $\varepsilon$-far, improving the prior lower bounds $\Omega(\sqrt{n}/\varepsilon^2)$ by \hyperlink{cite.soleimanifar2022testing}{Soleimanifar and Wright (SODA 2022)} and $\Omega(\sqrt{r})$ by \hyperlink{cite.aaronson2022quantum}{Aaronson et al.\ (ITCS 2024)}. 
        \item A matching lower bound $\Omega(d/\varepsilon^2)$ for testing whether a $d$-dimensional bipartite state is maximally entangled or $\varepsilon$-far, showing that the algorithm of \hyperlink{cite.OW15}{O'Donnell and Wright (STOC 2015)} is optimal for this task.
        \item A query lower bound $\widetilde\Omega(\sqrt{d/\Delta})$ for the $d$-dimensional entanglement entropy problem with gap $\Delta$, improving the prior lower bounds $\Omega(\sqrt[4]{d})$ by \hyperlink{cite.SY23}{She and Yuen (ITCS 2023)} and $\widetilde{\Omega}(1/\sqrt{\Delta})$ by \hyperlink{cite.WZ23}{Wang and Zhang (SICOMP 2025)} and \hyperlink{cite.Weg24}{Weggemans (Quantum 2025)}.
    \end{itemize}
    Moreover, we extend our central result to a robust version where the tested states are subject to noise and are not guaranteed to be pure: in this case,
    one-way LOCC is sufficient to realize the optimal tester.
\end{abstract}

\newpage
\tableofcontents
\newpage

\section{Introduction}

Bipartite quantum states are the most basic objects
for quantum entanglement (cf.\ \cite{HHHH09}) to manifest (e.g., EPR pairs~\cite{EPR35})
and to be exploited (e.g., superdense coding~\cite{BJ92} and quantum teleportation~\cite{BBC+93}).
In this paper, we consider the task of testing a property $\mathcal{P}=(\mathcal{P}^{\textup{yes}}, \mathcal{P}^{\textup{no}})$ of bipartite pure states;
that is, given samples of a bipartite pure state,
determine which one of the disjoint sets $\mathcal{P}^{\textup{yes}}$ or $\mathcal{P}^{\textup{no}}$ it falls in.
There is a folklore implication:
\begin{equation}
    \label{eq:dual}
    \begin{split}
    &\mathcal{P}\text{ is unitarily invariant on one part} 
    \Longrightarrow {} 
    \mathcal{P}\text{ is locally testable on the other part}.
    \end{split}
\end{equation}
This implication, however, does not concern the complexity of the local tester.
Is it efficient?
To this question we give a very positive answer:
\begin{itemize}
    \item Local tester is able to achieve the \textit{optimal} sample complexity over all global testers.
\end{itemize}
Furthermore, one may wonder if this optimality still applies when the provided states are mixed rather than pure (while the property to be verified, i.e., \(\mathcal{P}^{\textup{yes}}\), remains pure). In this case, we show that one-way LOCC (local operations and classical communication) suffices for an optimal tester.

\subsection{Motivations}
\label{sub:moti}

The concern of the complexity of local testers for bipartite quantum states is further motivated by four explicit questions of individual interest.
The first question is about purified samples in testing properties of mixed quantum states and the last three questions are about testing properties of the entanglement spectra of bipartite pure states.
They will be finally addressed by our central result.

\paragraph{Purified samples in mixed state testing}

Property testing of bipartite quantum states is closely related to that of mixed quantum states via the following observation.
Testing the property of a mixed state given its purified samples (which are bipartite states) is always no harder than given mixed samples directly, as one can just ignore the ancilla qubits of the purification and use the ordinary tester for mixed states.
Although mixed state testing has been extensively studied in the literature, e.g., \cite{MdW16,OW15,BOW19}, 
to the best of our knowledge, we are not aware of any mixed state testing problem that could benefit from using purified samples of mixed states.
We therefore ask the following question.
\begin{question} \label{q:purification}
    Can we test mixed states more efficiently given purified samples?
\end{question}

\paragraph{Entanglement spectrum}

A bipartite pure quantum state can be expressed via the Schmidt decomposition as:
\begin{equation*} 
   \ket{\psi}_{\textup{AB}} = \sum_{j} \lambda_j \ket{\phi_j}_{\textup{A}} \ket{\gamma_j}_{\textup{B}}
\end{equation*}
for some orthonormal bases $\cbra{\ket{\phi_j}_{\textup{A}}}$ and $\cbra{\ket{\gamma_j}_{\textup{B}}}$.
The set of Schmidt coefficients $\cbra{\lambda_j}$ is known as the entanglement spectrum,
which generalizes the entanglement entropy~\cite{BBPS96} 
and encodes richer information of the bipartite entanglement~\cite{LH08,CL08,YQ10}.

Testing properties of entanglement spectra is of great interest in quantum physics, e.g., as a probe to test topological properties of many-body quantum states \cite{KP06,LW06,PTBO10}.
From the perspective of quantum computational complexity, a fundamental problem is how many samples of a bipartite pure state are required to test a property of its entanglement spectrum.
This is generally known as the sample complexity for quantum property testing \cite{MdW16}.
The sample complexity of many property testing problems related to the entanglement spectrum are not yet well understood.
In the following, we introduce three representative questions in this area. 

\begin{itemize}
\item \textbf{Schmidt rank}. The Schmidt rank of a bipartite pure state is the number of its non-zero Schmidt coefficients, which can be understood as a measure of entanglement --- the entanglement entropy induced by the Tsallis entropy \cite{Tsa88}. 
Hence, testing the Schmidt rank is a way to quantify the quantum entanglement.
An approach to this problem is by reducing to testing the rank of the reduced density operator of the bipartite state, and the latter can be done by the rank tester proposed in \cite{OW15}.
However, there is no known general lower bound for this problem other than the trivial $\Omega\rbra{1}$.
Regarding the importance of this problem, it was raised as an open question in \cite{MdW16}.
\begin{question} [{\cite[Question 8]{MdW16}}] \label{q:schmidt-rank}
    Can we show a general lower bound for testing the Schmidt rank?
\end{question}

\item \textbf{Bond dimension of matrix product states}. 
A generalization of the Schmidt rank to the multipartite case is the bond dimension of matrix product states (MPS, cf.\ \cite{PGVWC07}).
Any $n$-partite pure state can be represented in the form of MPS:
as a sequential contraction of $n$ local tensors (see \cref{eq:mat-prd-state} for the formal definition).
The bond dimension of an MPS is the maximum dimension of the shared indices between neighboring tensors,
which is also a measure of the strength of entanglement.
When $n=2$, the bond dimension is exactly the Schmidt rank.
Testing whether a state is an MPS of bond dimension $r$ or $\varepsilon$-far (in trace distance) is of great interest. 
In \cite{soleimanifar2022testing}, they provided a tester for this task with sample complexity $O\rbra{nr^2/\varepsilon^2}$ (with perfect completeness);
they also showed an $\Omega\rbra{\sqrt{n}/\varepsilon^2}$ lower bound.
Later in \cite{aaronson2022quantum}, they proved
an $\Omega\rbra{\sqrt{r}}$ lower bound. 
\begin{question} \label{q:mps}
    Can we show a better lower bound for testing MPS?
\end{question}

\item \textbf{Maximal entanglement}. A bipartite pure state is maximally entangled if and only if its reduced density operator (on either part) is maximally mixed.
Hence, the tester proposed in \cite{OW15} for the maximal mixedness of $d$-dimensional mixed states with sample complexity $O\rbra{d}$ directly implies a sample upper bound $O\rbra{d}$ for testing the maximal entanglement of bipartite pure states. 
\begin{question} \label{q:max-entangled}
    Can we test the maximal entanglement more efficiently by a global tester on bipartite states instead of on reduced states?
\end{question}
\end{itemize}

\subsection{Our Results}

In this paper, we show that, for properties of bipartite pure states,
unitary invariance on one part implies an \textit{optimal} (over all global testers) local tester acting only on the other part.
This suggests a \textit{canonical} local tester for entanglement spectra and reveals the limitations of purifications, allowing us to give answers to the aforementioned questions as well as to prove new lower bounds for a series of quantum property testing problems. 

To illustrate our results, we first introduce the notions of \textit{local} testers and \textit{unitarily invariant} properties (see \cref{sec:def-testers} for general definitions). 

\begin{definition} [Local testers for bipartite states]
    Let $\HAB = \HA \otimes \HB$ be a bipartite Hilbert space. 
    A (global) tester $\mathcal{T}$ with sample complexity $N$ for bipartite states in $\HAB$ acts on $\HAB^{\otimes N}$. The tester $\mathcal{T}$ is \emph{local} on $\HA$, if it acts
    non-trivially only on $\HA^{\otimes N}$. 
\end{definition}

\begin{definition} [Unitarily invariant properties] \label{def:local-unitary-invariance}
    A property $\mathcal{P} = \rbra{\mathcal{P}^{\textup{yes}}, \mathcal{P}^{\textup{no}}}$ of bipartite pure states in $\HAB$ is \emph{unitarily invariant} on $\HA$, if $\rbra{U_\textup{A} \otimes I_{\textup{B}}} \ket{\psi}_{\textup{AB}} \in \mathcal{P}^X$ for every $\ket{\psi}_{\textup{AB}} \in \mathcal{P}^X$, unitary operator $U_{\textup{A}}$ on $\HA$, and $X \in \cbra{\textup{yes}, \textup{no}}$, where $I_{\textup{B}}$ is the identity operator on $\HB$.
\end{definition}

Now we formally state the central result as follows. 

\begin{theorem} [Unitary invariance implies optimal local testability, \cref{lemma-3191738} restated] \label{thm:key}
    For any property $\mathcal{P}$ of bipartite pure states in $\HAB$, if $\mathcal{P}$ is unitarily invariant on $\HB$, then there is a local tester on $\HA$ for $\mathcal{P}$ with \emph{optimal} sample complexity (over all global testers) with any soundness and completeness.
\end{theorem}

Our central result in \cref{thm:key} is non-trivial,
and significantly relies on 
(i) the unitary invariance of $\mathcal{P}$ on one part ($\HB$); and
(ii) $\mathcal{P}$ is a property of bipartite pure states.
To see the necessity of condition (i),
note that not all properties of bipartite pure states have a local tester. 
For example, consider property $\mathcal{P}$ with
$\mathcal{P}^{\textup{yes}}=\cbra{\ket{\psi}_{\textup{A}}\otimes \ket{\phi}_{\textup{B}}: \ket{\psi}_{\textup{A}}\in \HA}$ and
$\mathcal{P}^{\textup{no}}=\cbra{\ket{\psi}_{\textup{A}}\otimes \ket{\gamma}_{\textup{B}}: \ket{\psi}_{\textup{A}}\in \HA}$ for some fixed states $\ket{\phi}_{\textup{B}}\neq \ket{\gamma}_{\textup{B}}$.
Then $\mathcal{P}$ does not satisfy condition (i),
and it is easy to see there is no local tester on $\HA$ for $\mathcal{P}$,
because the tester has to obtain knowledge about Bob's system. 
To see the necessity of condition (ii),
consider, for example, property $\mathcal{P}$ of mixed quantum states in $\mathcal{D}\rbra{\mathbb{C}^d \otimes \mathbb{C}^d}$ with
$\mathcal{P}^{\textup{yes}}=\cbra{\ketbra{\psi}{\psi}_{\textup{AB}}:\ket{\psi}_{\textup{AB}}\textup{ is maximally entangled}}$
and $\mathcal{P}^{\textup{no}}=\cbra{I_{\textup{A}}/d\otimes I_{\textup{B}}/d}$.
Then $\mathcal{P}$ does not satisfy condition (ii),
and there is also no local tester for $\mathcal{P}$,
because the reduced density operator of any state in $\mathcal{P}^{\textup{yes}}$ or $\mathcal{P}^{\textup{no}}$
on Alice's system is $I_{\textup{A}}/d$.

It is natural to consider whether \cref{thm:key} can be extended to the case when the tested bipartite quantum state
is a mixed state. As will be shown later in \cref{sec-432205},
in this more general case,
as long as $\mathcal{P}^{\textup{yes}}$ only consists of pure states,
one-way LOCC tester (with details explained in~\cref{sec-432205}) suffices to achieve the optimality.

Moreover, \cref{thm:key} leads to a series of implications and applications:
\begin{itemize}
    \item We reveal the limitations of purifications in mixed state testing (in \cref{sec:mixed-state-testing});
    \item We identify a canonical tester for entanglement spectra (in \cref{sec:entanglement-spectrum-testing});
    \item We prove new sample lower bounds for bipartite/multipartite state testing (in \cref{sec:sample-intro});
    \item We prove new query lower bounds for unitary property testing (in \cref{sec:query-intro}).
\end{itemize}

\subsubsection{Mixed state testing with purifications} \label{sec:mixed-state-testing}

Testing properties of mixed states has been extensively studied in the literature (cf.\ \cite[Section 4.2]{MdW16}).
However, it is unclear if purified samples of mixed states would aid testing, as raised in \cref{q:purification}.

For a property $\mathcal{Q} = \rbra{\mathcal{Q}^{\textup{yes}}, \mathcal{Q}^{\textup{no}}}$ of mixed states, its \textit{purified} version is defined by $\mathsf{Purify}\rbra{\mathcal{Q}} = \rbra{\mathsf{Purify}\rbra{\mathcal{Q}^{\textup{yes}}},\mathsf{Purify}\rbra{\mathcal{Q}^{\textup{no}}}}$, where for a set $\mathcal{R}$ of mixed states,
\begin{equation*}
\mathsf{Purify}\rbra{\mathcal{R}} = \set{\ket{\psi}_{\textup{AB}} \in \HAB}{\tr_{\textup{B}}\rbra{\ketbra{\psi}{\psi}_{\textup{AB}}} \in \mathcal{R}}.
\end{equation*}
The consideration of
the purified version $\mathsf{Purify}\rbra{\mathcal{Q}}$ can be understood by the observation that a purification should contain more information than the mixed state itself (and therefore may be exploited in the property testing).
An immediate relation is that the sample complexity (see \cref{def:sample-complexity}) of $\mathsf{Purify}\rbra{\mathcal{Q}}$ is no greater than that of $\mathcal{Q}$ as one can always ignore the ancilla system and solve $\mathsf{Purify}\rbra{\mathcal{Q}}$ directly via the tester for $\mathcal{Q}$. 
Surprisingly, it turns out that the extra information in the purifications do not benefit mixed state testing, thereby giving a negative answer to \cref{q:purification}.

\begin{corollary} [Purified samples offer no advantage in mixed state testing] \label{corollary:purification-useless}
For every property $\mathcal{Q}$ of mixed quantum states, $\mathsf{Purify}\rbra{\mathcal{Q}}$ and $\mathcal{Q}$ have the same sample complexity with any soundness and completeness.
\end{corollary}
\begin{proof} [Proof sketch]
    This can be shown by \cref{thm:key}
    and by noting that $\mathsf{Purify}\rbra{\mathcal{Q}}$ is unitarily invariant on the ancilla qubits of purifications. 
\end{proof}

\cref{corollary:purification-useless} bridges two property testing scenarios: bipartite state testing and mixed state testing. 
This connection allows us to lift a series of sample lower bounds for mixed state testing to those for bipartite/multipartite state testing, which will be introduced in \cref{sec:sample-intro}. 
We note that, when restricted to the incoherent measurement setting, the relation between the sample complexities of $\mathsf{Purify}(\mathcal{Q})$ and $\mathcal{Q}$ might be different (see Point \ref{sec-310104} in \cref{sec-4201540}).

\subsubsection{Entanglement spectrum testing} \label{sec:entanglement-spectrum-testing}

Testing properties of entanglement spectra is a basic problem in bipartite state testing.
Examples include the aforementioned questions: testing the Schmidt rank (\cref{q:schmidt-rank}) and maximal entanglement (\cref{q:max-entangled}) of bipartite pure states.
In the following, we characterize a \textit{canonical} tester for entanglement spectra of bipartite pure states, which can be taken to consist of \textit{locally} performing weak Schur sampling~\cite{CHW07,MdW16} (see \cref{sec-2282154} for a brief introduction) and classically postprocessing the results. 

\begin{corollary} [Optimal canonical tester for entanglement spectra] 
    There is an optimal canonical tester for entanglement spectra consisting of locally performing weak Schur sampling and classically postprocessing the results with any soundness and completeness. 
\end{corollary}
\begin{proof} [Proof sketch]
    This can be shown by \cref{thm:key}
    and by noting that any property of the entanglement spectrum of a bipartite pure state in $\HAB$ is unitarily invariant both on $\HA$ and on $\HB$, together with the canonical tester for unitarily invariant properties of mixed states given in \cite[Lemma 20]{MdW16}. 
\end{proof}

\subsubsection{New sample lower bounds} \label{sec:sample-intro}

We obtain new sample lower bounds for testing bipartite pure states by lifting the sample lower bounds for their corresponding properties of mixed states. 
The first three give answers to \cref{q:schmidt-rank,q:mps,q:max-entangled},
of which the importance has been already introduced in \cref{sub:moti}.
The fourth can be applied to prove a new quantum query lower bound in \cref{sec:query-intro} later.

\paragraph{Schmidt rank}
We show the \textit{first} general lower bound for testing the Schmidt rank of bipartite pure states, thereby giving a positive answer to \cref{q:schmidt-rank}.\footnote{A closely related problem was considered in \cite[Lemma 13]{aaronson2022quantum}.
They showed that,  
determining whether a $d$-dimensional bipartite state has Schmidt rank at most $2^{\polylog \rbra{\log\rbra{d}}}$ (which is negligible compared to $d$) requires sample complexity $\omega\rbra{\polylog\rbra{d}}$. 
} 

\begin{corollary} [Schmidt rank, \cref{thm:schmidt-rank} restated] \label{corollary:schmidt-rank-intro}
    Any tester for determining whether the Schmidt rank of a bipartite pure state is at most $r$ or $\varepsilon$-far (in trace distance) requires sample complexity $\Omega\rbra{r/\varepsilon^2}$.
    Moreover, when perfect completeness is required, it requires sample complexity $\Omega\rbra{r^2/\varepsilon^2}$.
\end{corollary}

Using the method proposed in \cite{OW15} for testing the rank of mixed states, 
we can test the Schmidt rank of bipartite pure states with perfect completeness with sample complexity $O\rbra{r^2/\varepsilon^2}$ (see also \cite{soleimanifar2022testing}), which is shown to be optimal by \cref{corollary:schmidt-rank-intro}. 

Moreover, \cref{corollary:schmidt-rank-intro} can be generalized to multipartite case later in \cref{corollary:mps-intro}. 

\paragraph{Bond dimension of matrix product states}
We show lower bounds for testing MPS (with open boundary condition~\cite{PGVWC07}, defined by \cref{eq:mat-prd-state}).

\begin{corollary} [Bond dimension of MPS, \cref{thm-414017} restated] \label{corollary:mps-intro}
    For any \(n\geq 2\) and \(r\geq 2\), any tester for determining whether an $n$-partite pure state is an MPS of bond dimension $r$ or $\varepsilon$-far (in trace distance) requires sample complexity $\Omega\rbra{(\sqrt{nr}+r)/\varepsilon^2}$.
    Moreover, when perfect completeness is required, it requires sample complexity $\Omega\rbra{(\sqrt{nr}+r^2)/\varepsilon^2}$. 
\end{corollary}

Previously in \cite{soleimanifar2022testing}, they showed an $\Omega\rbra{\sqrt{n}/\varepsilon^2}$ lower bound for testing $n$-partite MPS of bond dimension $r$, while leaving the $r$-dependence open. Later in \cite{aaronson2022quantum}, they showed an $\Omega\rbra{\sqrt{r}}$ lower bound as long as $r \leq 2^{n/8}$.
In \cref{corollary:mps-intro}, we improve both the lower bounds in \cite{soleimanifar2022testing,aaronson2022quantum} by a factor of \(\sqrt{r}\) and our lower bound requires no restrictions between $r$ and $n$. 

When perfect completeness is required, a tester in \cite{soleimanifar2022testing} was proposed with sample complexity $O\rbra{nr^2/\varepsilon^2}$. 
In \cref{corollary:mps-intro}, we show a matching lower bound $\Omega\rbra{r^2/\varepsilon^2}$ with respect to both $r$ and $\varepsilon$ for any \(n\geq 2\) and \(r\geq 2\). 

\paragraph{Maximal entanglement}
We show a matching lower bound for testing whether a bipartite pure state is maximally entangled, implying that the tester of \cite{OW15} for mixedness with sample complexity $O\rbra{d/\varepsilon^2}$ is optimal for testing maximal entanglement, thereby giving a negative answer to \cref{q:max-entangled}.\footnote{In \cite{hayashi2006study,AHL+14}, they consider testing whether a bipartite pure state is the specific maximally entangled state $\frac{1}{\sqrt{d}} \sum_i \ket{i} \ket{i}$ with LOCC or few qubits communication. The lower bound in \cref{corollary:maximal-entanglement-intro} does not apply to their case, because the property they considered is not unitarily invariant on either part.}
\begin{corollary} [Maximal entanglement, \cref{thm:maximal-entangle} restated] \label{corollary:maximal-entanglement-intro}
    Any tester for determining whether a bipartite pure state in Hilbert space $\mathbb{C}^d \otimes \mathbb{C}^d$ is maximally entangled or $\varepsilon$-far (in trace distance) requires sample complexity
    $\Omega\rbra{d/\varepsilon^2}$.
\end{corollary}

Compared to the full entanglement spectrum,
the entanglement entropy is a more concise quantification of the strength of entanglement. 
We also show a lower bound for testing whether the entanglement entropy is low or high. 

\begin{corollary} [Entanglement entropy, \cref{thm:entanglement-entropy} restated] \label{corollary:s-maximum-entropy}
    For any constant $\alpha > 0$, any tester for determining whether the $\alpha$-R\'enyi entanglement entropy of a bipartite pure state in Hilbert space $\mathbb{C}^d \otimes \mathbb{C}^d$ is low ($\leq a$) or high ($\geq b$) requires sample complexity
    $\Omega\rbra{d/\Delta+d^{1/\alpha-1}/\Delta^{1/\alpha}}$, where $\Delta = b - a$. 
    Note that $\alpha = 1$ means the case for the von Neumann entanglement entropy.
\end{corollary}

The sample lower bound in \cref{corollary:s-maximum-entropy} will later be used to prove new query lower bounds in \cref{sec:query-intro}. 

\paragraph{Uniform Schmidt coefficients}

We show a matching lower bound for distinguishing between the cases when the entanglement spectrum is uniform on $r$ or $r+\Delta$ non-zero Schmidt coefficients, implying that the tester of \cite{OW15} for distinguishing between the cases when the spectrum is uniform on $r$ and $r+\Delta$ eigenvalues with sample complexity $O\rbra{r^2/\Delta}$ is optimal for the bipartite case. 

\begin{corollary} [Uniform Schmidt coefficients, \cref{thm:uniform-schmidt} restated] \label{corollary:uniform-schmidt-intro}
    Any tester for determining whether the entanglement spectrum of a bipartite pure state is uniform on $r$ or $r + \Delta$ Schmidt coefficients requires sample complexity $\Omega^*\rbra{r^2/\Delta}$.\footnote{$\Omega^*\rbra{\cdot}$ suppresses quasi-polylogarithmic factors, e.g., $r^{o\rbra{1}}$.}
\end{corollary}

\paragraph{Productness}
A multipartite pure state is product if it can be written as a tensor product of local states. 
A state that is not product is actually entangled. 
Testing the productness was first discussed in \cite{MKB05} and later extensively analyzed in \cite{HM10,soleimanifar2022testing,JW24}.
It is known that productness can be tested with sample complexity $O\rbra{1/\varepsilon^2}$ by the tester in \cite{HM10} (cf.\ \cite{MdW16}). 
We show a matching lower bound for testing the productness, which was already noted in \cite{soleimanifar2022testing}.\footnote{The matching lower bound was noted in the paragraph after \cite[Proposition 1.2]{soleimanifar2022testing}. 
Here, we prove it by a different reduction from purity testing. 
Note that a product state is actually an MPS of bond dimension $r = 1$. However, the lower bound $\Omega\rbra{r/\varepsilon^2}$ given in \cref{corollary:mps-intro} does not imply the lower bound $\Omega\rbra{1/\varepsilon^2}$ for testing productness because the condition $r \geq 2$ is required in the proof of \cref{thm:mps-lb}.}

\begin{corollary} [Productness, \cref{thm:product} restated] \label{corollary:productness-intro}
Any tester for determining whether a bipartite pure state is a product state or $\varepsilon$-far (in trace distance) requires sample complexity
$\Omega\rbra{1/\varepsilon^2}$.
\end{corollary}

\subsubsection{New query lower bounds} \label{sec:query-intro}

Beyond the aforementioned new quantum sample lower bounds, 
we also obtain new quantum query lower bounds. 

\paragraph{Entanglement entropy}
The entanglement entropy problem with query access to a reflection oracle studied in \cite{SY23} is given as follows. 
Let $U_{\textup{AB}} = I_{\textup{AB}} - 2\ketbra{\psi}{\psi}_{\textup{AB}}$ be a unitary operator on $\HAB$, where $\ket{\psi}_{\textup{AB}} \in \HAB$. 
For $\alpha > 0$ and $0 < a < b \leq \ln\rbra{d}$, let $\textsc{QEntanglementEntropy}^{\alpha,a,b}$ be the problem to determine which is the case (promised that it is in either case):
\begin{itemize}
    \item Low entropy: $\mathrm{E}_\alpha\rbra{\ket{\psi}_{\textup{AB}}} \leq a$;
    \item High entropy: $\mathrm{E}_\alpha\rbra{\ket{\psi}_{\textup{AB}}} \geq b$,
\end{itemize}
where $\mathrm{E}_\alpha\rbra{\cdot}$ is the $\alpha$-R\'enyi entanglement entropy. 
The quantum query complexity for the entanglement entropy problem was recently studied in the literature (for the case of $\alpha = 2$):
\begin{itemize}
    \item $\mathsf{Q}\rbra{\textsc{QEntanglementEntropy}^{2,a,b}} = \Omega\rbra{e^{a/4}}$ is shown in \cite{SY23}, which implies an $\Omega\rbra{\sqrt[4]{d}}$ lower bound when $a$ is close to $\ln\rbra{d}$;
    \item $\mathsf{Q}\rbra{\textsc{QEntanglementEntropy}^{2,a,b}} = \widetilde \Omega\rbra{1/\sqrt{\Delta}}$ is shown in \cite{WZ23}, where $\Delta = b-a$, and this was later improved to $\Omega\rbra{1/\sqrt{\Delta}}$ in \cite{Weg24}.
\end{itemize}
We show a query lower bound for the entanglement entropy problem for general $\alpha$ by lifting the sample lower bound given in \cref{corollary:s-maximum-entropy} via quantum sample-to-query lifting \cite{WZ23}.

\begin{corollary} [Entanglement entropy problem, \cref{thm:q-entropy} restated] \label{corollary:query-entanglement-entropy}
    For any constant $\alpha > 0$, any tester for \\ $\textsc{QEntanglementEntropy}^{\alpha,a,b}$ requires query complexity $\widetilde \Omega\rbra{\sqrt{d/\Delta}+\sqrt{d^{1/\alpha-1}/\Delta^{1/\alpha}}}$, where $\Delta = b - a$.
\end{corollary}

We can see that all prior lower bounds only consider the case of $\alpha = 2$. 
By comparison, \cref{corollary:query-entanglement-entropy} gives lower bounds for any $\alpha > 0$; especially, for the case of $\alpha = 2$, it gives a lower bound $\widetilde \Omega\rbra{\sqrt{d/\Delta}}$, which covers the aforementioned lower bounds $\Omega\rbra{\sqrt[4]{d}}$ and $\widetilde \Omega\rbra{1/\sqrt{\Delta}}$ by \cite{SY23,WZ23,Weg24}.
Moreover, it is worth noting that the query lower bound in \cref{corollary:query-entanglement-entropy} is multiplicative in $d$ and $\Delta$.

\subsubsection{Robust version---Extension to bipartite mixed states}\label{sec-432205}

In the above we are only concerned about bipartite pure quantum states.
What if the tested quantum states are not guaranteed to be pure? This scenario is of interest when the quantum states to be tested are subject to noise, while the target $\mathcal{P}^{\textup{yes}}$ contains only noiseless (pure) states. 
In this setting, we prove the following theorem as an extension of \cref{thm:key}:
the unitary invariance\footnote{Here, unitarily invariant properties in \cref{def:local-unitary-invariance} can be naturally extended to mixed states (see \cref{def:general-unitarily-invariant}).} on one part still implies an optimal one-way LOCC tester (see \cref{def-45215}) over all global testers.

\begin{theorem}[Unitary invariance implies optimal LOCC tester for mixed input, \cref{lemma-3311612} restated]
    \label{thm:ext}
    For any property $\mathcal{P} = \rbra{\mathcal{P}^{\textup{yes}}, \mathcal{P}^{\textup{no}}}$ of bipartite mixed states
    that is unitarily invariant on $\HB$,
    if $\mathcal{P}^{\textup{yes}}$ only consists of pure states, 
    then there is a one-way LOCC tester for $\mathcal{P}$ 
    with \emph{optimal} sample complexity (over all global testers). 
\end{theorem}

As concrete examples, \Cref{thm:ext} implies that the \textit{robust} versions of Schmidt rank, bond dimension of MPS, maximal entanglement, entanglement entropy, uniform Schmidt coefficients, and productness can be optimally tested using one-way LOCC testers.

We note that LOCC protocols (cf.~\cite{watrous2018theory,chitambar2014everything}) have been widely studied in the context of property testing of bipartite quantum states. 
For example, the verification of bipartite pure state was considered in \cite{hayashi2006study,AHL+14,wang2019optimal,zhu2019efficient,pallister2018optimal}, i.e., verifying whether a given quantum device successfully produces a designated bipartite state. 
In \cite{hayashi2006study}, the designated state is the standard maximally entangled state,
which was further studied by~\cite{AHL+14} using communication with few qubits.
Other tasks have also been studied, e.g., purity testing~\cite{matsumoto2010test} and inner product estimation~\cite{anshu2022distributed}.

\subsection{Techniques}
Now let us give a brief overview of the main techniques used in this paper to obtain the optimalities of local testers for pure states and one-way LOCC testers for mixed states. Our results are based on an extended Schur-Weyl duality on bipartite systems (see \cref{lemma-3241548}), which decomposes the tensor product of \(N\) bipartite systems into irreducible spaces of the group \(\mathbb{S}_N\times \mathbb{S}_N \times\mathbb{U}_d\times \mathbb{U}_d\) (the direct product of symmetric groups on \(N\) letters and unitary groups on \(d\)-dimensional Hilbert space), where the first \(\mathbb{S}_N\) and the first \(\mathbb{U}_d\) act on Alice's system, and the others act on Bob's system.
Specifically, 
\begin{equation*}
\begin{split}
&\HAB^{\otimes N}\stackrel{(\mathbb{S}_N\times \mathbb{S}_N)\times (\mathbb{U}_d\times \mathbb{U}_d)}{\cong} 
\bigoplus_{\lambda_1,\lambda_2\vdash (N,d)} 
\mathcal{V}_{\lambda_1,\textup{A}}\otimes\mathcal{V}_{\lambda_2,\textup{B}} \otimes \mathcal{W}_{\lambda_1,\textup{A}} \otimes \mathcal{W}_{\lambda_2,\textup{B}},
\end{split}
\end{equation*}
where \(\mathcal{V}_{\lambda,\textup{A}}=\mathcal{V}_{\lambda,\textup{B}}\coloneqq\mathcal{V}_\lambda\) and \(\mathcal{W}_{\lambda,\textup{A}}=\mathcal{W}_{\lambda,\textup{B}}\coloneqq\mathcal{W}_{\lambda}\) are irreducible representations of \(\mathbb{S}_N\) and \(\mathbb{U}_d\), respectively.

\paragraph{Local tester for pure states} 
Our construction builds on two key observations. 
\begin{enumerate}
\item Our first observation is that if the property is unitarily invariant on Bob's system, we can safely ``discard'' those irreducible spaces of \(\mathbb{U}_d\) possessed by Bob (see \cref{eq-450110}), and force it to be maximally mixed states. This does help to decouple the bipartite states to some extent.
However, it is not yet sufficient for us to locally test the states on Alice's part, since Bob's irreducible spaces of \(\mathbb{S}_N\) may still be entangled with Alice's system. \label{item:ob1}
\item Our second observation is that the \textit{pure} tested state obeys a normal form on these irreducible spaces (see \cref{lemma-3262205}), due to the invariance of the tensor product of \(N\) pure samples under the group actions of \((\pi,\pi)\in\mathbb{S}_N\times\mathbb{S}_N\), called simultaneous-permutations. Then, we can see that the tensor state in any irreducible representation of \(\mathbb{S}_N\times\mathbb{S}_N\) must be maximally entangled in the standard basis, and therefore the isotypic representations (with respect to \(\mathbb{S}_N\)) of Alice and Bob can only be coupled when they are isomorphic (this fact was previously observed in \cite{matsumoto2007universal} and is closely related to the \((\textup{GL}_n,\textup{GL}_m)\)-duality~\cite{howe1987gl} in classical invariant theory). Therefore, we can define a local tester that mimics the behavior of the global tester acting and post-selecting on these maximally entangled states (see \cref{eq-450110}), which also has exactly the same sample complexity as the global tester. \label{item:ob2}
\end{enumerate}
Specifically, suppose that $\mathcal{T}$ is a (global) tester with sample complexity $N$ for a property $\mathcal{P}$ of bipartite pure states that is unitarily invariant on $\HB$.
Then, one can construct the following local tester on $\HA$:
\begin{equation}\label{eq-450110}
\begin{split}
    \hat{\mathcal{T}} &\coloneqq \bigoplus_{\lambda\vdash (N,d)}\frac{1}{\dim(\mathcal{W}_\lambda)\dim(\mathcal{V}_\lambda)}  
    \cdot \overbrace{ I_{\mathcal{V}_{\lambda,\TA}} \!\otimes \!\!\underbrace{\tr_{\mathcal{W}_{\lambda,\TB}}}_{\textup{Observation~\ref{item:ob1}}} \!\Bigg[ \underbrace{\bbra{I_{\mathcal{V}_\lambda}} 
 \int_{U\in\mathbb{U}_d} \!\!\!\! U_{\textup{B}}^{\otimes N} \mathcal{T} U_{\textup{B}}^{\dag \otimes N}   \kett{I_{\mathcal{V}_\lambda}} }_{\textup{Observation~\ref{item:ob2}}} \Bigg] }^{\HA^{\otimes N}}
\otimes \overbrace{I_{\textup{B}}}^{\HB^{\otimes N}},
\end{split}
\end{equation}
and it can be shown that $\hat{\mathcal{T}}$ is a local tester for $\mathcal{P}$ with sample complexity $N$. 

\textit{Applications in proving new lower bounds}. 
The local tester allows us to investigate the lower bounds for property testing problems regarding bipartite pure states by reductions of property testing of mixed states (as shown in \cref{sec:sample-intro}). 
For example, the rank of mixed states reduces to the Schmidt rank of bipartite pure states.
Furthermore, these new sample lower bounds for testing bipartite pure states imply new query lower bounds when the oracle is a reflection operator by the sample-to-query lifting \cite{WZ23} 
(as shown in \cref{sec:query-intro}). 

\paragraph{One-way LOCC tester for mixed states} 
When the provided states are mixed, the previous normal form no longer exists, as the tensor product of samples of a mixed state only commutes with the actions of simultaneous-permutations but is not invariant under them. Nevertheless, we can still make use of this property. That is, we can define a new tester that commutes with both the actions of \(\mathbb{U}_d\) on Bob's system and the actions of simultaneous-permutations (see \cref{eq-421128}). 
Notably, we prove that this tester can be embedded with a purity-tester~\cite{matsumoto2010test} without affecting its success probability (see \cref{lemma:442015}); here, the purity-tester can be viewed as the projector onto the subspace spanned by the samples of bipartite pure states, or the symmetric subspace if we consider the bipartite system as a whole system. 
The embedded tester turns out to have a normal form that fixes a maximally entangled state on each isomorphic pair of the irreducible spaces of \(\mathbb{S}_N\), and such maximally entangled state can be approximated by one-way LOCC as suggested in \cite{hayashi2006study}. Therefore, the embedded tester can be implemented through one-way LOCC in a similar manner, with only an overhead of a constant factor in the sample complexity (see \cref{lemma-43153}).

\subsection{Related Work}
\subsubsection{Comparison with \texorpdfstring{\cite{soleimanifar2022testing}}{[SW22]}} The result in \cite[Theorem 5.2]{soleimanifar2022testing} showed an optimal local tester for a bipartite state discrimination task. We note that their result requires unitary invariance on \textit{both} parts, and implies that purification does not help for testing spectrum properties of mixed states. By contrast, our \cref{thm:key} requires only \textbf{single-part} unitary invariance, thereby implying that purification does not help for testing \textbf{any properties} of mixed states. Consequently, our result recovers \cite[Theorem 5.2]{soleimanifar2022testing} (see \cref{corollary-4262150} from \cref{theorem-424241}), and further leads to the following advantages. 
    \begin{itemize}
        \item Better lower bounds: We are able to prove a better lower bound of $\Omega\rbra{\rbra{\sqrt{nr}+r}/\varepsilon^2}$ (see \cref{corollary:mps-intro}) for testing matrix product states than the lower bound $\Omega\rbra{\sqrt{n}/\varepsilon^2}$ in \cite[Theorem 1.4]{soleimanifar2022testing}. More technical details can be found in \cref{sec-722031} (especially \cref{lemma-4242204}). 
        \item More applications: 
        We are able to show that having access to purified samples offers no advantage in any property testing task of mixed states (see \cref{corollary:purification-useless}). Moreover, our result was recently applied to complexity theory for quantum promise problems~\cite{CCHS24}, and quantum channel tomography~\cite{chen2025quantum} (see recent developments in \cref{sec-5251957}).
        \item Robust extension: Our results can be extended to the \textit{robust} version of property testing where the unknown bipartite states can be noisy (allowed to be mixed). In this case, the LOCC tester emerges as the optimal tester (see \Cref{thm:ext}) in place of the local tester for the noiseless case.
    \end{itemize}

\subsubsection{Recent Developments}\label{sec-5251957}
After the work described in this paper, 
there have been several recent developments.

Liu, Gong, Du, and Cai \cite{LGDC24} further investigated the property testing of a mixed state with and without its purifications in the incoherent measurement setting or with limited quantum memory. 
In particular, their Theorems 1 and 2 show that the purified samples can provide certain exponential advantages over the mixed samples when (i) the ancilla system for the purification is of constant dimension and (ii) only incoherent measurements are allowed.
Intuitively, such advantage relies on that when the ancilla dimension is constant, certain quantities like \(\tr(\rho^t)\) for integer \(t\) can be estimated by performing measurements only on the ancilla system.

Lovitz and Lowe \cite{LL24} further established nearly tight sample upper and lower bounds for testing tree tensor network states (TTNS) with one-sided error (i.e., with perfect completeness; see \cref{def:tester}).
In particular, their Theorem 1 implies a matching (up to a logarithmic factor) lower bound \(\Omega(nr^2/(\varepsilon^2 \log n))\) for the MPS testing with one-sided error, thereby showing the optimality of the one-sided error algorithm in \cite{soleimanifar2022testing}. However, it is still open whether there exists a lower bound for the case of two-sided error (i.e., the default setting \cite{MdW16}; see \cref{def:tester}) better than the \(\Omega((\sqrt{nr}+r)/\varepsilon^2)\) provided in our \cref{corollary:mps-intro}.

Chia, Chung, Huang, and Shih \cite{CCHS24} established the separation of quantum promise complexity classes ($\mathsf{pBQP/poly}\subsetneq \mathsf{pBQP/qpoly}$ and $\mathsf{mBQP/poly}\subsetneq \mathsf{mBQP/qpoly}$), and resolved an open problem in \cite{morimae2024unconditionally} concerning unconditional secure quantum commitment schemes, where our result (\cref{theorem-424241}) is used in their proof as an important step.

Tang, Wright, and Zhandry~\cite{TWZ25} algorithmically strengthened our \Cref{corollary:purification-useless} by explicitly constructing an algorithm for generating random purifications of a mixed state (see also \cite{girardi2025random} for an alternative construction). 
Intuitively, their random purification and our local test can be thought of as dual concepts in the Schr{\"o}dinger and Heisenberg pictures, respectively.
The local test and random purification techniques show applications in quantum state tomography~\cite{pelecanos2025mixedstatetomographyreduces} and quantum channel tomography~\cite{mele2025optimallearningquantumchannels,chen2025quantum}. 
These techniques were further extended to fermionic/bosonic random purification~\cite{walter2025random,mele2025random} and local test/random dilation of quantum channels~\cite{chen2025quantum,GMZFL25,yoshida2025random}.
Building on these techniques, the authors of \cite{TWZ25} further constructed an algorithm that simulates $q$ queries to a state-preparation unitary oracle for a mixed state $\rho$, using $O(q^2)$ samples of $\rho$, strengthening the quantum sample-to-query lifting theorem \cite{WZ23}. 
It turns out that these techniques lead to a series of complexity bounds in quantum property testing~\cite{CWZ25}.
In particular, the query lower bound $\widetilde{\Omega}\rbra{\sqrt{d/\Delta}}$ for $\alpha = 1$ in \cref{corollary:query-entanglement-entropy} was improved to $\Omega\rbra{\sqrt{d/\Delta} + \log\rbra{d}/\rbra{\Delta\sqrt{\log\rbra{1/\Delta}\log\log\rbra{1/\Delta}}}}$ in \cite[Theorem 2.34]{CWZ25}.

\subsection{Discussion}\label{sec-4201540}

In this paper we prove that for properties of bipartite pure quantum states,
the unitary invariance on one part implies the existence of a local tester on the other part that achieves the optimal sample complexity over all global testers.
An implication is that purified samples offer no advantage in property testing of mixed states, which is then applied to obtain a series of new quantum sample lower bounds.
In particular, we prove the first general lower bound for testing Schmidt rank, thereby answering an open question in~\cite{MdW16}. 
We improve the prior best lower bounds for testing MPS by \cite{soleimanifar2022testing,aaronson2022quantum}; furthermore, when perfect completeness is required, we show a matching lower bound with respect to the rank and precision.
We also show new sample lower bounds for testing maximal entanglement, entanglement entropy and uniform Schmidt coefficients.
Beyond the sample complexity, we are also able to prove a new quantum query lower bound for the entanglement entropy problem considered in~\cite{SY23}.
Finally, we extend our central result to the case of mixed state input,
and show that one-way LOCC is sufficient to realize the optimal tester. 

We list several open questions for future works as follows:
\begin{enumerate}
    \item
        \cref{thm:key} actually implies the following implication:
        \begin{equation*}
        \begin{split}
            &\text{unitary invariance on one part} 
            \Longrightarrow \, \text{an optimal local tester on the other part},
        \end{split}
        \end{equation*}
        which strengthened the folklore implication in \cref{eq:dual}.
        In \cref{thm:ext}, we have shown that when the tested quantum states are mixed
        and $\mathcal{P}^{\textup{yes}}$ only consists of pure states,
        the unitary invariance on one part implies the existence of an optimal (over all global testers)
        one-way LOCC tester.
        However, the optimality of the LOCC tester is up to a constant factor; can we strengthen the result to remove the constant factor?
    \item 
        In addition to the first question,
        one might be wondering if \cref{thm:key} can be extended to the property testing of mixed states
        in a more general sense,
        without restricting $\mathcal{P}^{\textup{yes}}$ to only consisting of pure states.
        In other words,
        for any property $\mathcal{P}$ of bipartite mixed states unitarily invariant on one part,
        is there an optimal (over all global testers) and canonical tester with some locality constraints for testing $\mathcal{P}$?
    
    \item 
        \cref{corollary:schmidt-rank-intro,corollary:mps-intro} provide the lower bounds $\Omega\rbra{r}$ for testing the Schmidt rank and \(\Omega(\sqrt{nr}+r)\) for the bond dimension of MPS (with two-sided error), respectively. 
        However, the current best upper bounds for these two problems are $O\rbra{r^2}$ due to \cite[Theorem 1.11]{OW15}, and \(O(nr^2)\) due to \cite[Theorem 1.3]{soleimanifar2022testing}.
        There are still gaps between the upper and lower bounds.
        One interesting question is whether our lower bounds can be further improved.
        It is worth noting that in \cite{LL24}, they showed a matching (up to a logarithmic factor) lower bound \(\Omega(nr^2/(\varepsilon^2 \log n))\) for testing MPS with \textit{one-sided} error. 
    
    \item 
        The applications of our results are not limited to the sample complexity of property testing.
        For example, 
        in \cref{corollary:query-entanglement-entropy}, 
        a new quantum query lower bound was proved for the entanglement entropy problem,
        by applying the quantum sample-to-query lifting~\cite{WZ23} to \cref{corollary:s-maximum-entropy}.
        Can we use similar approaches to show more quantum query lower bounds?

    \item\label{sec-310104}
        The power and limitations of incoherent (single-copy) testers have been studied for different tasks (e.g., \cite{aharonov2022quantum,chen2022exponential,anshu2022distributed,chen2023unitarity,fawzi2023quantum,bluhm2024hamiltonian,liu2024separation}). 
        The optimal local tester (see \Cref{thm:key}) constructed in our paper requires a coherent (multi-copy) measurement on the local samples possessed by Alice (or Bob). Therefore, one may ask whether the optimality still holds in the incoherent setting. That is, for unitarily invariant properties, can an incoherent local tester achieve the optimal sample complexity over all incoherent global testers, or, 
        do purified samples help in incoherent setting (c.f.\ \cref{corollary:purification-useless})?
        We note that there are some recent developments in this direction (see \Cref{sec-5251957}).

    \item 
        In this paper, we only consider property testing of bipartite states. A natural question is whether our \cref{thm:key} can be extended to the multipartite scenarios. For example, when the property is unitarily invariant on each party in a given subset (of parties), does there exist an optimal local tester on the complementary parties? We note that a direct application of \cref{thm:key} is insufficient, since tracing out one party does not, in general, preserve purity of the remaining state.
\end{enumerate}

\subsection{Organization}

The remainder of the paper is organized as follows.
In \cref{sec:pre}, the basic notations and concepts of testing quantum states and the representation theory are introduced.
The central result (\cref{thm:key}) will be presented and proved in \cref{sec:opt-local-tester}.
The limitations of purification (\cref{corollary:purification-useless}) will be used in \cref{sec:sample-bounds}
to obtain a series of new sample lower bounds;
and in \cref{sec:query-bounds}, a new query lower bound will be further derived.
Finally, the extension of the central result (\cref{thm:ext})
will be presented and proved in \cref{sec:optimal-locc}.

\section{Preliminaries}
\label{sec:pre}

\subsection{Notations}
Suppose \(\mathcal{H}\) is a Hilbert space of dimension \(d\), equipped with a standard basis (e.g., the computational basis). 
We use $\ket{\psi}$ to denote a vector (or, a pure state) in $\mathcal{H}$, and use $\bra{\psi}$ to denote the Hermitian conjugate of $\ket{\psi}$. 
A mixed quantum state given by an ensemble $\cbra{\rbra{\lambda_j, \ket{\psi}_j}}_j$ of pure quantum states, where $\sum_j \lambda_j = 1$ and $\lambda_j \geq 0$, can be described by a density operator $\rho = \sum_j \lambda_j \ketbra{\psi_j}{\psi_j}$. 

For two mixed quantum states $\rho$ and $\sigma$, their fidelity is defined by $\mathrm{F}\rbra{\rho, \sigma} = \tr\sbra{\sqrt{\sqrt{\sigma}\rho\sqrt{\sigma}}}$ and their trace distance is defined by $d_{\tr}\rbra{\rho, \sigma} = \frac 1 2 \tr\rbra{\abs{\rho-\sigma}}$.
The two measures have the following relation.
\begin{lemma}[Fuchs–van de Graaf, {\cite[Theorem 1]{FvdG99}}] \label{lemma:fi-td}
    For mixed quantum states $\rho$ and $\sigma$, we have
    \begin{equation*}
    1 - \mathrm{F}\rbra{\rho, \sigma} \leq d_{\tr}\rbra{\rho, \sigma} \leq \sqrt{1 - \mathrm{F}\rbra{\rho, \sigma}^2}.
\end{equation*}
\end{lemma}

A linear map \(X \colon \mathcal{H}\rightarrow\mathcal{H}\) can be turned into a \(d\times d\) complex-valued matrix with respect to the standard basis.
The set of all \(d\times d\) complex-valued matrices can be formed as a \(d^2\)-dimensional vector space \(\mathbb{C}^{d^2}\) by simply flattening the matrices to vectors (row-major). For a \(d\times d\) matrix \(X\), we use \(\kett{X}\) to denote the corresponding element in this vector space. For example,
\begin{equation}\label{eq-327123}
\kett{\ketbra{\psi}{\phi}}=\ket{\psi}\ket{\phi^*}, \quad\quad \kett{XYZ^\dag}=X\otimes Z^* \kett{Y},
\end{equation}
where \(\ket{\phi^*}\) and \(Z^*\) are the entry-wise complex conjugate of the pure state \(\ket{\phi}\) and linear operator \(Z\) (with respect to the standard basis), respectively. The inner product of this vector space is thus defined by \(\bbrakett{A}{B}=\tr(A^\dag B)\). 
We will use \(I_{\mathcal{H}}\) to denote the identity map on \(\mathcal{H}\). Thus, \(\kett{I_{\mathcal{H}}}=\sum_{i}\ket{i}\ket{i}\) is a (non-normalized) maximally entangled state with respect to the standard basis.
\begin{remark}
For convenience, we may use the vertical form \(\begin{matrix}X\\Y\end{matrix}\) to denote the tensor \(X\otimes Y\). For example, \cref{eq-327123} can be written as
\begin{equation*}
\kett{\ketbra{\psi}{\phi}}=\begin{matrix}\ket{\psi}\\\ket{\phi^*}\end{matrix}\,, \quad\quad \kett{XYZ^\dag}=\begin{matrix}X\\ Z^*\end{matrix} \kett{Y}.
\end{equation*}
\end{remark}

\subsection{Quantum State Testing} \label{sec:def-testers}

Let $\mathcal{D}\rbra{\mathcal{H}}$ denote the set of density operators on Hilbert space $\mathcal{H}$. 
A property of mixed quantum states in $\mathcal{D}\rbra{\mathcal{H}}$ is a pair of disjoint sets $\mathcal{P} = \rbra{\mathcal{P}^{\textup{yes}}, \mathcal{P}^{\textup{no}}}$ with $\mathcal{P}^{\textup{yes}} \cap \mathcal{P}^{\textup{no}} = \emptyset$ and $\mathcal{P}^{\textup{yes}} \cup \mathcal{P}^{\textup{no}} \subseteq \mathcal{D}\rbra{\mathcal{H}}$. 
The notion of quantum testers for properties of quantum states is given as follows. 

\begin{definition} [Testers for properties of quantum states]
    \label{def:tester}
    A tester $\mathcal{T}$ with sample complexity $N$ is described by a linear operator $0 \sqsubseteq \mathcal{T} \sqsubseteq I$ on $\mathcal{H}^{\otimes N}$, where $\sqsubseteq$ is the L\"{o}wner order and $I$ is the identity operator. 
    For $0 \leq s < c \leq 1$, $\mathcal{T}$ is called a $\rbra{c,s}$-tester for property $\mathcal{P} = \rbra{\mathcal{P}^{\textup{yes}}, \mathcal{P}^{\textup{no}}}$, if
\begin{itemize}
    \item 
    \textbf{Completeness}: For every $\rho \in \mathcal{P}^{\textup{yes}}$, $\mathcal{T}$ accepts with probability $\tr\sbra{\mathcal{T} \rho^{\otimes N}} \geq c$,
    \item 
    \textbf{Soundness}: For every $\rho \in \mathcal{P}^{\textup{no}}$, $\mathcal{T}$ accepts with probability $\tr\sbra{\mathcal{T} \rho^{\otimes N}} \leq s$.
\end{itemize}
In particular,
\begin{itemize}
    \item A tester for $\mathcal{P}$ (in the default setting~\cite{MdW16}) means a $\rbra{2/3,1/3}$-tester for $\mathcal{P}$,
    \item A tester for $\mathcal{P}$ with perfect completeness means a $\rbra{1,1/3}$-tester for $\mathcal{P}$.
\end{itemize}
\end{definition}
\begin{definition} [Sample complexity] \label{def:sample-complexity}
For $0 \leq s < c \leq 1$, the $\rbra{c,s}$-sample complexity of property $\mathcal{P}$, denoted by $\mathsf{S}_{c,s}\rbra{\mathcal{P}}$, is the minimum sample complexity over all $\rbra{c,s}$-testers for $\mathcal{P}$. 
In particular, we write $\mathsf{S}\rbra{\mathcal{P}} \coloneqq \mathsf{S}_{2/3,1/3}\rbra{\mathcal{P}}$ (in the default setting) and $\mathsf{S}_1\rbra{\mathcal{P}} \coloneqq \mathsf{S}_{1,1/3}\rbra{\mathcal{P}}$ (with perfect completeness).
\end{definition}

\begin{remark}
    The notation $\mathsf{S}_1\rbra{\cdot}$ for the sample complexity with perfect completeness is due to the same reason for the notation $\mathsf{QMA}_1$ for the class of $\mathsf{QMA}$ with perfect completeness \cite{Bra11,GN16}. 
    It trivially holds that $\mathsf{S}\rbra{\mathcal{P}} \leq \mathsf{S}_1\rbra{\mathcal{P}}$.
\end{remark}

For two properties $\mathcal{P}$ and $\mathcal{Q}$,
we denote $\mathcal{P}\subseteq \mathcal{Q}$ if $\mathcal{P}^{X}\subseteq \mathcal{Q}^{X}$ for $X\in \{\textup{yes},\textup{no}\}$.
It is easy to see that if $\mathcal{P}\subseteq \mathcal{Q}$, then any tester for $\mathcal{Q}$ is also a tester for $\mathcal{P}$, which implies the following fact.
\begin{fact} \label{fact:SPSQ}
    For properties $\mathcal{P}$ and $\mathcal{Q}$, if $\mathcal{P}\subseteq \mathcal{Q}$, then $\mathsf{S}_{c,s}\rbra{\mathcal{P}} \leq \mathsf{S}_{c,s}\rbra{\mathcal{Q}}$.
\end{fact}

The property induced by a set $\mathcal{R}$ of mixed states in $\mathcal{D}\rbra{\mathcal{H}}$ with sensitivity parameter $\varepsilon$ is denoted by $\mathcal{R}_\varepsilon = \rbra{\mathcal{R}^{\textup{yes}}, \mathcal{R}^{\textup{no}}}$, where $\mathcal{R}^{\textup{yes}} = \mathcal{R}$ and 
\begin{equation*}
\mathcal{R}^{\textup{no}} = \mathcal{R}^{\geq \varepsilon} =\set{\rho \in \mathcal{D}\rbra{\mathcal{H}}}{d_{\tr}\rbra{\rho, \sigma} \geq \varepsilon \textup{ for every } \sigma \in \mathcal{R}}.
\end{equation*}
Here, $d_{\tr}\rbra{\cdot, \cdot}$ is the trace distance. Note that $\mathcal{R}^{\textup{yes}} \cup \mathcal{R}^{\textup{no}} \subseteq \mathcal{D}\rbra{\mathcal{H}}$. 

\paragraph{Bipartite state testing}
Let $\HAB = \HA \otimes \HB$ be the Hilbert space of bipartite pure states with $\HA \cong \HB \cong \mathbb{C}^d$,\footnote{Throughout this paper, we assume that $\HA$ and $\HB$ have the same dimension $d$ for simplicity. Actually, all results in this paper can be generalized to the general case when $\HA$ and $\HB$ have different dimensions.}
where the subspaces $\HA$ and $\HB$ are possessed by Alice and Bob, respectively.
For clarity, a pure state in $\HA$ is associated with the subscript $\textup{A}$, e.g., $\ket{\psi}_{\textup{A}}$. Moreover, the conjugate of $\ket{\psi}_{\textup{A}}$ is denoted as $\bra{\psi}_{\textup{A}}$, and the density operator of $\ket{\psi}_{\textup{A}}$ is denoted as $\ketbra{\psi}{\psi}_{\textup{A}}$. 
Let $\mathcal{P} = \rbra{\mathcal{P}^{\textup{yes}}, \mathcal{P}^{\textup{no}}}$ be a property of bipartite mixed states in $\mathcal{D}\rbra{\HAB}$. 
If both $\mathcal{P}^{\textup{yes}}$ and $\mathcal{P}^{\textup{no}}$ only consist of bipartite pure states, then $\mathcal{P}$ is called a property of bipartite pure states; in this case, we may assume that $\mathcal{P}^{\textup{yes}}$ and $\mathcal{P}^{\textup{no}}$ are subsets of $\HAB$. 

We introduce the notion of unitarily invariant properties of bipartite (mixed) states as follows. 

\begin{definition} [Unitarily invariant properties] \label{def:general-unitarily-invariant}
    A property $\mathcal{P} = \rbra{\mathcal{P}^{\textup{yes}}, \mathcal{P}^{\textup{no}}}$ of bipartite mixed states in $\mathcal{D}\rbra{\HAB}$ is said to be unitarily invariant on $\HA$, if $\rbra{U_{\TA} \otimes I_{\TB}} \rho \rbra{U_{\TA}^\dag \otimes I_{\TB}} \in \mathcal{P}^X$ for any $\rho \in \mathcal{P}^X$, $X \in \cbra{\textup{yes}, \textup{no}}$, and any unitary operator $U_{\TA}$ on $\HA$. 
    The unitary invariance on $\HB$ is defined similarly.
\end{definition}

\cref{def:general-unitarily-invariant} extends the definition of unitarily invariant properties (of pure states) in \cref{def:local-unitary-invariance} to mixed states. 

\begin{definition} [Local testers for bipartite states]
    A tester $\mathcal{T}$ acting on $\HAB^{\otimes N}$
    is said to be a \emph{local} tester on $\HA$, if $\mathcal{T} = M_{\textup{A}} \otimes I_{\textup{B}}$, where $M_{\textup{A}}$ is a linear operator on $\HA^{\otimes N}$ with $0_{\textup{A}} \sqsubseteq M_{\textup{A}} \sqsubseteq I_{\textup{A}}$. Here, $I_{\textup{A}}$ and $I_{\textup{B}}$ denote the identity operators on $\HA^{\otimes N}$ and $\HB^{\otimes N}$, respectively.
\end{definition}

We also give the definition of one-way LOCC tester as follows. 

\begin{definition}[One-way LOCC testers for bipartite states]\label{def-45215}
A tester \(\mathcal{T}\) acting on $\HAB^{\otimes N}$ 
is said to be a one-way LOCC tester from Bob to Alice, if it can be implemented through the following strategy:
\begin{enumerate}
\item Bob first performs a measurement on \(\HB^{\otimes N}\) and send a classical message to Alice,
\item Alice then performs a measurement on \(\HA^{\otimes N}\) and accepts or rejects.
\end{enumerate}
\end{definition}

The property induced by a set $\mathcal{Q}$ of bipartite pure states in $\HAB$ with sensitivity parameter $\varepsilon$ is denoted by $\mathcal{Q}_\varepsilon = \rbra{\mathcal{Q}^{\textup{yes}}, \mathcal{Q}^{\textup{no}}}$, where $\mathcal{Q}^{\textup{yes}} = \mathcal{Q}$ and 
\begin{equation*}
\begin{split}
\mathcal{Q}^{\textup{no}} &= \mathcal{Q}^{\geq \varepsilon}=\{\ket{\psi}_{\textup{AB}} \in \HAB\,:\, d_{\tr}\rbra{\ket{\psi}_{\textup{AB}}, \ket{\phi}_{\textup{AB}}} \geq \varepsilon  
 \textup{ for every } \ket{\phi}_{\textup{AB}} \in \mathcal{Q}\}.
\end{split}
\end{equation*}
Note that $\mathcal{Q}^{\textup{yes}} \cup \mathcal{Q}^{\textup{no}} \subseteq \HAB$ contains only pure states.

\subsection{Basic Representation Theory}
Herein, we recall some basics of representation theory.
\begin{definition}
A complex representation of a group \(G\) is a pair \((\mu,\mathcal{H})\), where \(\mathcal{H}\) is a Hilbert space, and \(\mu: G\rightarrow \textup{GL}(\mathcal{H})\) is a group homomorphism.
Here, $\textup{GL}\rbra{\mathcal{H}}$ denotes the general linear group of $\mathcal{H}$, i.e., the group of invertible linear maps $\mathcal{H} \to \mathcal{H}$. 
\end{definition}
\begin{remark}
If the context is clear, we will refer to a representation \((\mu,\mathcal{H})\) of \(G\) by either the group homomorphism \(\mu\), or the space \(\mathcal{H}\). We also call \(\mu(g)\) the \textit{action} of \(g\in G\) on \(\mathcal{H}\). Moreover, if the group action is clear in the context or does not particularly concern us, we will directly write \(g\) for \(\mu(g)\).
\end{remark}
For a group \(G\), a \textit{sub-representation} of \((\mu,\mathcal{H})\) is a representation \((\mu',\mathcal{H}')\), where \(\mathcal{H}'\) is a subspace of \(\mathcal{H}\) and \(\mu'(g)\) is simply the restriction of \(\mu(g)\) to \(\mathcal{H}'\). 
\begin{definition}
A representation \(\mathcal{H}\) of \(G\) is irreducible if the only sub-representations of \(\mathcal{H}\) are \(\{0\}\) and \(\mathcal{H}\) itself.
\end{definition}
A \textit{representation homomorphism} between two representations \((\mu_1,\mathcal{H}_1),(\mu_2,\mathcal{H}_2)\) of group \(G\) is a linear operator \(F:\mathcal{H}_1\rightarrow \mathcal{H}_2\) which commutes with the action of \(G\), i.e., 
\begin{equation*}
F\mu_1(g)=\mu_2(g)F.
\end{equation*}
A \textit{representation isomorphism} is a representation homomorphism that is also an isomorphism of vector space (full-rank linear map). Two representations \(\mathcal{H}_1\) and \(\mathcal{H}_2\) of a group \(G\) are said to be \textit{isomorphic} if there exists a representation isomorphism between them, and we write \(\mathcal{H}_1\stackrel{G}{\cong}\mathcal{H}_2\).

\begin{proposition}[Schur's lemma, see, e.g.~{\cite[Proposition 2.3.9]{etingof2011introduction}}] \label{prop:202404032229}
Let \(\mathcal{H}_1,\mathcal{H}_2\) be irreducible representations of a group \(G\). If \(F: \mathcal{H}_1\rightarrow \mathcal{H}_2\) is a non-zero homomorphism of representations, then \(F\) is an isomorphism.
\end{proposition}

In this paper, we mainly focus on representations of finite groups and compact Lie groups. Any finite-dimensional representation of a finite group or a compact group is isomorphic to a \textit{unitary representation}. 
For example, if \(G\) is finite and \(\mu\) is a representation of \(G\), then we can define a new representation \(\mu'\) by
\begin{equation*}
\mu'(g)\!=\!M^{\frac{1}{2}} \mu(g) M^{-\frac{1}{2}}, \textup{where}\, M\!=\!\frac{1}{|G|}\sum_{h\in G} \mu(h)^\dag\mu(h).
\end{equation*}
Simple calculations show that \(\mu'\) is unitary, and \(M^{1/2}\) is an isomorphism. Similarly, for compact Lie groups, the summation is replaced by integral over the Haar measure~\cite{kirillov2008introduction}. 

\begin{remark}
Throughout this paper, we assume that all representations are unitary. 
\end{remark}

\begin{corollary}\label{lemma-41116}
Suppose \((\mu_1,\mathcal{H}_1)\) and \((\mu_2,\mathcal{H}_2)\) are two non-isomorphic irreducible representations of a finite group \(G\) and \(\mathcal{V}_1,\mathcal{V}_2\) are vector spaces. Suppose \(F: \mathcal{V}_1\otimes \mathcal{H}_1\rightarrow \mathcal{V}_2\otimes \mathcal{H}_2\) is a linear map. Then,
\begin{equation}\label{eq-41117}
\frac{1}{|G|}\sum_{g\in G} (I_{\mathcal{V}_2}\otimes \mu_2(g)) F (I_{\mathcal{V}_1}\otimes \mu_1(g))^\dag=0.
\end{equation}
If \(G\) is a compact Lie group, the above property still holds except that the average sum is replaced by the Haar integral over \(G\).
\end{corollary}
\begin{proof}
Suppose \(F=\sum_{i}X_i\otimes Y_i\) where \(X_i: \mathcal{V}_1\rightarrow \mathcal{V}_2\) and \(Y_i: \mathcal{H}_1\rightarrow \mathcal{H}_2\) are linear maps. We have
\begin{equation*}
\begin{split}
&\frac{1}{|G|}\sum_{g\in G} (I_{\mathcal{V}_2}\otimes \mu_2(g)) X_i\otimes Y_i (I_{\mathcal{V}_1}\otimes \mu_1(g))^\dag 
=\frac{1}{|G|}\sum_{g\in G} (I_{\mathcal{V}_2}\otimes \mu_2(g)) X_i\otimes Y_i (I_{\mathcal{V}_1}\otimes \mu_1(g^{-1})) 
=0, 
\end{split}
\end{equation*}
where the first equality is by the unitary representation, and the second equality is by Schur's lemma (\cref{prop:202404032229}) and the fact that \(\mathcal{H}_1\ncong\mathcal{H}_2\) are both irreducible. Then, \cref{eq-41117} follows immediately.

If \(G\) is a compact Lie group, the proof is similar, except that the average summation becomes Haar integral.
\end{proof}

\begin{corollary}\label{lemma-3262332}
Suppose \((\mu,\mathcal{H})\) is an irreducible representation of a finite group \(G\) and \(\mathcal{V}_1,\mathcal{V}_2\) are vector spaces. Suppose \(F: \mathcal{V}_1\otimes \mathcal{H}\rightarrow \mathcal{V}_2\otimes \mathcal{H}\) is a linear map. Then,
\begin{equation}\label{eq-327110}
\frac{1}{|G|}\sum_{g\in G} (I_{\mathcal{V}_2}\otimes \mu(g)) F (I_{\mathcal{V}_1}\otimes \mu(g))^\dag=\frac{1}{\dim(\mathcal{H})} \tr_{\mathcal{H}}(F)\otimes I_{\mathcal{H}}.
\end{equation}
If \(G\) is a compact Lie group, the above property still holds except that the average sum is replaced by the Haar integral over \(G\).
\end{corollary}
\begin{proof}
Suppose \(F=\sum_{i}X_i\otimes Y_i\) where \(X_i: \mathcal{V}_1\rightarrow \mathcal{V}_2\) and \(Y_i: \mathcal{H}\rightarrow \mathcal{H}\) are linear maps. We have
\begin{equation*}
\begin{split}
\frac{1}{|G|}\sum_{g\in G} (I_{\mathcal{V}_2}\otimes \mu(g)) X_i\otimes Y_i (I_{\mathcal{V}_1}\otimes \mu(g))^\dag 
=&\frac{1}{|G|}\sum_{g\in G} (I_{\mathcal{V}_2}\otimes \mu(g)) X_i\otimes Y_i (I_{\mathcal{V}_1}\otimes \mu(g^{-1}))\\
=&\frac{\tr(Y_i)}{\dim(\mathcal{H})} X_i\otimes I_{\mathcal{H}}.
\end{split}
\end{equation*}
where the first equality is by the unitary representation, and the second equality is by the Schur's lemma (\cref{prop:202404032229}) and the fact that \(\mathcal{H}\) is irreducible.
Then, \cref{eq-327110} follows immediately. 

If \(G\) is a compact Lie group, the proof is similar, except that the average summation becomes Haar integral.
\end{proof}

\begin{proposition}\label{prop-432248}
If \(\mathcal{H}_1\), \(\mathcal{H}_2\) are irreducible representations of groups \(G_1\) and \(G_2\), respectively, then
\(\mathcal{H}_1\otimes \mathcal{H}_2\)
is also an irreducible representation of the group \(G_1\times G_2\), and the group action is simply:
\begin{equation*}
(g_1, g_2) \ket{v_1}\otimes \ket{v_2}=g_1\ket{v_1} \otimes g_2\ket{v_2},
\end{equation*}
where \(g_1\in G_1,g_2\in G_2,\ket{v_1}\in\mathcal{H}_1,\ket{v_2}\in\mathcal{H}_2\).
\end{proposition}

\subsection{Schur-Weyl Duality on Bipartite Systems}
\subsubsection{Young diagram}
A Young diagram \(\lambda\) with \(N\) boxes and \(r\) rows can be identified by a partition \((\lambda_1,\ldots,\lambda_r)\) of \(N\) such that \(\sum_{i} \lambda_i=N\) and \(\lambda_1\geq\cdots \geq \lambda_r>0\). For example, the Young diagram with \(8\) boxes and \(3\) rows, identified by the partition \((4,3,1)\) is:
\begin{equation*}
\vcenter{\hbox{\scalebox{0.9}{\begin{ytableau}~&~&~&~\\~&~&~\\~\end{ytableau}}}}.
\end{equation*}
We will write \(\lambda\vdash N\) to mean that \(\lambda\) is a Young diagram with \(N\) boxes, and write \(\lambda\vdash (N,d)\) to mean that \(\lambda\) is with \(N\) boxes and with no more than \(d\) rows.

Given a Young diagram \(\lambda\vdash N\), a standard Young tableau \(T\) of shape \(\lambda\) can be identified by a filling of the \(N\) boxes of \(\lambda\) by the integers \(1,\ldots,N\), using each number once and the integers increase from left to right and from top to bottom. For example, one valid standard Young tableau with shape \((4,3,1)\) is
\begin{equation*}
\vcenter{\hbox{\scalebox{0.9}{\begin{ytableau}1&3&4&6\\2&7&8\\5\end{ytableau}}}}.
\end{equation*}

\subsubsection{Schur-Weyl duality}
Now, consider the Hilbert space \((\mathbb{C}^d)^{\otimes N}\). This space admits representations of the symmetric group \(\mathbb{S}_N\) and unitary group \(\mathbb{U}_d\). The unitary group acts by simultaneous ``rotation'' as \(U^{\otimes N}\) for any \(U\in \mathbb{U}_d\) and the symmetric group acts by permuting tensor factors:
\begin{equation*}
P(\pi)\ket{\psi_1}\cdots\ket{\psi_N}=\ket{\psi_{\pi^{-1}(1)}}\cdots\ket{\psi_{\pi^{-1}(N)}},
\end{equation*}
where \(\pi\in\mathbb{S}_N\). Two actions \(U^{\otimes N}\) and \(P(\pi)\) commute with each other, and hence \((\mathbb{C}^d)^{\otimes N}\) admits a representation of group \(\mathbb{S}_N\times \mathbb{U}_d\). More specifically, we have:
\begin{proposition}[Schur-Weyl duality~\cite{fulton2013representation}]\label{prop-432239}
\begin{equation*}
(\mathbb{C}^{d})^{\otimes N}\stackrel{\mathbb{S}_N\times \mathbb{U}_d}{\cong}\bigoplus_{\lambda \vdash (N,d)}\mathcal{V}_\lambda\otimes \mathcal{W}_\lambda,
\end{equation*}
where \(\mathcal{V}_\lambda\) and \(\mathcal{W}_\lambda\) are irreducible representations of \(\mathbb{S}_N\) and \(\mathbb{U}_d\), respectively, and \(\lambda\vdash (N,d)\) refers to the Young diagram with \(N\) boxes and with no more than \(d\) rows.
\end{proposition}

\begin{remark}
For clarity, we will use \(P_\lambda(\pi)\) to denote the action of \(\pi\in\mathbb{S}_N\) on the irreducible representation \(\mathcal{V}_\lambda\), and use \(Q_\lambda(U)\) to denote the action of \(U\in\mathbb{U}_d\) on the irreducible representation \(\mathcal{W}_\lambda\).
\end{remark}

\begin{remark}\label{remark-3252205}
We assume that all representations of the symmetric group are real representations, since they are realizable over the field of real numbers~\cite{serre1977linear}.
\end{remark}

\subsubsection{Schur-Weyl duality on bipartite systems}
Now, we consider the Hilbert space \(\mathcal{H}_{\textup{AB}}^{\otimes N}=(\HA\otimes \HB)^{\otimes N}\cong (\mathbb{C}^{d}\otimes \mathbb{C}^d)^{\otimes N}\).

Denote by \(\mathbb{S}_N\) the symmetric group on \(N\) letters. We define the action of \(\mathbb{S}_N\times \mathbb{S}_N\) on the space \(\HAB^{\otimes N}=(\HA\otimes \HB)^{\otimes N}\) as
\begin{equation*}
\begin{split}
&(\pi,\sigma) \cdot\ket{\psi_{A,1}}\ket{\psi_{B,1}}\cdots \ket{\psi_{A,N}}\ket{\psi_{B,N}}
=\,\ket{\psi_{A,\pi^{-1}(1)}}\ket{\psi_{B,\sigma^{-1}(1)}}\cdots \ket{\psi_{A,\pi^{-1}(N)}}\ket{\psi_{B,\sigma^{-1}(N)}},
\end{split}
\end{equation*}
where \((\pi,\sigma)\in \mathbb{S}_N\times \mathbb{S}_N\). Therefore, the action of \((\pi,\sigma)\) is equivalent to \(P_{\TA}(\pi)\otimes P_{\TB}(\sigma)\) where \(P_{\TA}(\pi)\) and \(P_{\TB}(\sigma)\) acts on \(\HA^{\otimes N}\) and \(\HB^{\otimes N}\), respectively.
Moreover, if we identify the group \(\mathbb{S}_N\) as a subgroup of \(\mathbb{S}_N\times \mathbb{S}_N\) by:
\begin{equation*}
\mathbb{S}_N\ni \pi \mapsto (\pi,\pi)\in \mathbb{S}_N\times\mathbb{S}_N,
\end{equation*}
then we obtain an action of \(\mathbb{S}_N\) on \(\HA^{\otimes N}\otimes \HB^{\otimes N}\) called the \textit{simultaneous-permutation}:
\begin{equation*}
\begin{split}
&(\pi,\pi)\cdot \ket{\psi_{A,1}}\ket{\psi_{B,1}}\cdots \ket{\psi_{A,N}}\ket{\psi_{B,N}}
= \,\ket{\psi_{A,\pi^{-1}(1)}}\ket{\psi_{B,\pi^{-1}(1)}}\cdots \ket{\psi_{A,\pi^{-1}(N)}}\ket{\psi_{B,\pi^{-1}(N)}},
\end{split}
\end{equation*}
where \(\pi\in\mathbb{S}_N\). In other words, the simultaneous-permutation of \(\pi\in\mathbb{S}_N\) on \(\HAB\) is thus \(P_\TA(\pi)\otimes P_{\TB}(\pi)\).

Denote by \(\mathbb{U}_d\) the unitary group of dimension \(d\). We define the action of \(\mathbb{U}_d\times \mathbb{U}_d\) on the space \(\HAB^{\otimes N}=(\HA\otimes \HB)^{\otimes N}\) as
\begin{equation*}
\begin{split}
&(U,V)\cdot\ket{\psi_{A,1}}\ket{\psi_{B,1}}\cdots \ket{\psi_{A,N}}\ket{\psi_{B,N}}
=\, U\ket{\psi_{A,1}}V\ket{\psi_{B,1}}\cdots U\ket{\psi_{A,N}}V\ket{\psi_{B,N}},
\end{split}
\end{equation*}
where \((U,V)\in \mathbb{U}_d\times \mathbb{U}_d\). Therefore, the action of \((U, V)\) is equivalent to \(U^{\otimes N}\otimes V^{\otimes N}\) where \(U^{\otimes N}\) and \(V^{\otimes N}\) act on \(\HA^{\otimes N}\) and \(\HB^{\otimes N}\), respectively.

\begin{lemma}[Schur-Weyl duality on bipartite system]\label{lemma-3241548}
The bipartite system can be decomposed into a direct sum of irreducible representations as follows:
\begin{equation}\label{eq-432252}
\begin{split}
&\HAB^{\otimes N}\stackrel{(\mathbb{S}_N\times \mathbb{S}_N)\times (\mathbb{U}_d\times \mathbb{U}_d)}{\cong}
\bigoplus_{\lambda_1,\lambda_2\vdash (N,d)} \mathcal{V}_{\lambda_1,\textup{A}}\otimes\mathcal{V}_{\lambda_2,\textup{B}} \otimes \mathcal{W}_{\lambda_1,\textup{A}} \otimes \mathcal{W}_{\lambda_2,\textup{B}},
\end{split}
\end{equation}
where \(\mathcal{V}_{\lambda,\textup{A}}=\mathcal{V}_{\lambda,\textup{B}}\coloneqq\mathcal{V}_\lambda\) and \(\mathcal{W}_{\lambda,\textup{A}}=\mathcal{W}_{\lambda,\textup{B}}\coloneqq\mathcal{W}_{\lambda}\) are irreducible representations of \(\mathbb{S}_N\) and \(\mathbb{U}_d\), respectively, and the spaces notated with \(\textup{A}\) and \(\textup{B}\) are possessed by Alice and Bob, respectively.
\end{lemma}
\begin{proof}
Note that \(\mathcal{H}_{\textup{AB}}=\mathcal{H}_\TA\otimes \mathcal{H}_\TB\) and \(\mathcal{H}_{\TA}=\mathcal{H}_{\TB}=\mathbb{C}^{d}\). Then, the decomposition into irreducible representations shown in \cref{eq-432252} can be seen by combining \cref{prop-432239} with \cref{prop-432248}.
\end{proof}

\begin{remark}
In \cref{lemma-3241548}, we write \(\mathcal{V}_{\lambda,\TA}=\mathcal{V}_{\lambda,\TB}\) (or \(\mathcal{W}_{\lambda,\TA}=\mathcal{W}_{\lambda,\TB}\)) to mean that 
\(\mathcal{V}_{\lambda,\TA}\), \(\mathcal{V}_{\lambda,\TB}\) (or \(\mathcal{W}_{\lambda,\TA}\), \(\mathcal{W}_{\lambda,\TB}\)) are mathematically the same thing. However, they physically belong to Alice and Bob, respectively. If we are only interested in their intrinsic mathematical structures, we may omit the notation \(\TA\) or \(\TB\) in the subscriptions.
\end{remark}

\begin{remark}
The similar notation style will be used in other situations. For example, suppose \(U\in\mathbb{U}_d\) is a unitary matrix. Then, we will use \(U_\TB\) to mean that this unitary is acting on Bob's system, and use \(Q_{\lambda,\TA}(U)\) to denote an irreducible action of \(U\) corresponding to \(\lambda\) on Alice's system. Therefore, for a mathematical object \(O\), adding a subscript \(\TA\) or \(\TB\) is not to change its intrinsic meaning but only to emphasize which system it is related to.
\end{remark}

\subsubsection{Weak Schur sampling}\label{sec-2282154}
Here, we introduce the weak Schur sampling for testing (single-part) mixed states.
Suppose we have \(N\) samples of an unknown \(d\)-dimensional mixed quantum state \(\rho\). Consider the task of testing a property $\mathcal{P}$ of \(\rho\).
Without loss of generality, the tester can be described by a POVM \(\{M,I-M\}\) applied on \(\rho^{\otimes N}\) and ``accept'' is returned if the measurement outcome is $M$. Note that \(\rho^{\otimes N}\) is invariant under permutations, i.e.,
for any \(\pi\in\mathbb{S}_N\), \(P(\pi) \rho^{\otimes N}P(\pi)^\dag =\rho^{\otimes N}\).
This means that we can average over the symmetric group \(\mathbb{S}_N\) to obtain a permutation invariant tester.
Furthermore, if the property is unitarily invariant,
we can also average over the unitary group \(\mathbb{U}_d\) to obtain a unitarily invariant tester with the performance no worse than the original one. 
Specifically, we define the canonical permutation-invariant and unitary-invariant tester \(\overline{M}\) as:
\[\overline{M}= \frac{1}{N!}\sum_{\pi\in\mathbb{S}_N}P(\pi)\int_{U\in \mathbb{U}_d}\left[ U^{\otimes N} M U^{\dag\otimes N}\right]P(\pi)^\dag.\]
It can be shown that the tester \(\overline{M}\) is at least as powerful as the original tester \(M\) (see, e.g., \cite{MdW16}).
Note that \(\overline{M}\) commutes with both \(U^{\otimes N}\) and \(P(\pi)\) for any \(U\in\mathbb{U}_d\) and \(\pi\in\mathbb{S}_N\). 
By the Schur-Weyl duality (see \cref{prop-432239}), we have
\[\overline{M}=\bigoplus_{\lambda\vdash (N,d)} c_{\lambda}\cdot  I_{\mathcal{V}_\lambda}\otimes I_{\mathcal{W}_\lambda},\]
where \(0\leq c_{\lambda}\leq 1\).
Then, we can see that the tester \(\overline{M}\) applied on \(\rho^{\otimes N}\) is equivalent to 
\begin{enumerate}
    \item sample a \(\lambda\vdash (N,d)\) from the distribution \(\{\tr(I_{\mathcal{V}_\lambda}\otimes Q_{\lambda}(\rho))\}_\lambda=\{\dim(\mathcal{V}_\lambda)\cdot s_\lambda(\alpha)\}_\lambda\), where \(s_\lambda\) is the Schur polynomial and \(\alpha=(\alpha_1,\ldots,\alpha_d)\) are the eigenvalues of the state \(\rho\),
    \item accept with probability $c_\lambda$ and reject otherwise.
\end{enumerate}
It is worth noting that the second step is entirely classical, while the first step is a quantum measurement independent of the specific testing task, which is called \textit{weak Schur sampling}~\cite{CHW07}.

\section{Optimal Local Tester for Bipartite Pure States}
\label{sec:opt-local-tester}

In this section, we will show how to construct an optimal local tester for unitarily invariant properties of bipartite pure states, and thus prove \cref{thm:key}. 
We formally state it as follows. 

\begin{theorem} [Optimal local tester] \label{lemma-3191738}
Let \(\mathcal{P}\) be a property of bipartite pure states in \(\HAB\). If \(\mathcal{P}\) is unitarily invariant on \(\HB\), then, for any parameters \(0\leq s<c\leq 1\), there exists a local \((c,s)\)-tester on \(\HA\) for \(\mathcal{P}\) that achieves the sample complexity \(\mathsf{S}_{c,s}\rbra{\mathcal{P}}\).
\end{theorem}

\subsection{Construction (Proof of Theorem~\ref{lemma-3191738})}

To prove \cref{lemma-3191738}, let $\mathcal{T}$ be a (possibly global) \((c,s)\)-tester for $\mathcal{P}$ with sample complexity $\mathsf{S}_{c,s}\rbra{\mathcal{P}}$, then we use the following theorem to show the existence of an optimal local tester that also achieves sample complexity $\mathsf{S}_{c,s}\rbra{\mathcal{P}}$.
\begin{theorem}\label{theorem-4222259}
Suppose $\mathcal{P}$ is a property of bipartite pure states in $\HAB$ that is unitarily invariant on $\HB$. For any \(0\leq s<c\leq 1\), if \(\mathcal{T}\) is a \((c,s)\)-tester for \(\mathcal{P}\) with sample complexity \(N\), then there is a local \((c,s)\)-tester \(\hat{\mathcal{T}}\) for \(\mathcal{P}\) on \(\HA\) that also has sample complexity \(N\).
\end{theorem}
\begin{proof}
Our construction of the local \((c,s)\)-tester $\hat{\mathcal{T}}$ on $\HA$ is given as follows. 

\begin{enumerate}
    \item We first construct a new tester
    \begin{equation} \label{eq:def-tildeT}
    \widetilde{\mathcal{T}}\coloneqq\int_{U\in\mathbb{U}_d} U_{\textup{B}}^{\otimes N} \mathcal{T} U_{\textup{B}}^{\dag \otimes N},
    \end{equation}
    where \(U\) is a Haar random unitary matrix, and \(U_{\textup{B}}\) means \(U\) is acting on the system \(\HB\).
    \item Then, we construct a local tester \(\hat{\mathcal{T}}\) on $\HA$ by localizing the tester $\widetilde{\mathcal{T}}$:
    \begin{align}
    &\hat{\mathcal{T}} \coloneqq   
    \underbrace{\left( \bigoplus_{\lambda\vdash (N,d)}\!\!\frac{1}{\dim(\mathcal{W}_\lambda)\dim(\mathcal{V}_\lambda)}\cdot\begin{matrix}
        I_{\mathcal{V}_{\lambda,\TA}} \\ \tr_{\mathcal{W}_{\lambda,\TB}}\!\!\big[\bbra{I_{\mathcal{V}_\lambda}}\widetilde{\mathcal{T}}\kett{I_{\mathcal{V}_\lambda}}\big]
    \end{matrix} \right)}_{\HA^{\otimes N}}  
    \otimes \underbrace{I_{\textup{B}}}_{\HB^{\otimes N}},\label{eq:def-hatT}
    \end{align}
    where $\mathcal{V}_{\lambda, \textup{A}} = \mathcal{V}_{\lambda, \textup{B}} \coloneqq \mathcal{V}_\lambda$ and \(\mathcal{W}_{\lambda,\textup{A}}=\mathcal{W}_{\lambda,\textup{B}}\coloneqq\mathcal{W}_{\lambda}\) are irreducible representations of \(\mathbb{S}_N\) and \(\mathbb{U}_d\) that appear in the decomposition of $\HAB^{\otimes N}$ in \cref{lemma-3241548}. 
    Here, $I_{\mathcal{V}_\lambda}$ is the identity operator on $\mathcal{V}_{\lambda}$.
\end{enumerate}

To verify our construction, we first show that $\widetilde{\mathcal{T}}$ is a \((c,s)\)-tester for $\mathcal{P}$ with sample complexity $N$ and derive a formula for the acceptance probability of $\widetilde{\mathcal{T}}$ in \cref{sec-417206}. 
Since $\hat{\mathcal{T}}$ is local on $\HA$ (by its definition), the proof is completed by showing that $\hat{\mathcal{T}}$ is a \((c,s)\)-tester for $\mathcal{P}$ with sample complexity $N$. 
This is done by showing that $\hat{\mathcal{T}}$ behaves identically to the tester $\widetilde{\mathcal{T}}$ in \cref{sec-417207}.
\end{proof}

\subsection{\texorpdfstring{Acceptance Probability of $\widetilde{\mathcal{T}}$}{Acceptance Probability of tilde-T}}\label{sec-417206}

In this subsection, we show that $\widetilde{\mathcal{T}}$ is a \((c,s)\)-tester for $\mathcal{P}$ with sample complexity \(N\), and derive an expression for its acceptance probability. 
Technical lemmas are postponed in \cref{sec:tech-lemmas}.

\begin{lemma} \label{lemma:202404041522}
    Let $\widetilde{\mathcal{T}}$ be the tester defined by \cref{eq:def-tildeT}.
    Then,
    \begin{enumerate}
        \item $\widetilde{\mathcal{T}}$ is a \((c,s)\)-tester for $\mathcal{P}$ with sample complexity $N$. 
        \item $\widetilde{\mathcal{T}}$ accepts $\ket{\psi}_{\textup{AB}}$ with probability 
    \begin{equation}\label{eq-414052}
    \begin{split}
        &\tr\left[\widetilde{\mathcal{T}}\ketbra{\psi}{\psi}_{\textup{AB}}^{\otimes N}\right] = \sum_{\lambda\vdash (N,d)} \frac{1}{\dim(\mathcal{W}_\lambda)}  
        \cdot \tr\!\left[\tr_{\mathcal{W}_{\lambda,\TB}}\!\big[\bbra{I_{\mathcal{V}_\lambda}}\widetilde{\mathcal{T}}\kett{I_{\mathcal{V}_\lambda}}\big]\!\cdot\!\tr_{\mathcal{W}_{\lambda,\TB}}(\ketbra{w_{\lambda}}{w_{\lambda}}) \right],
    \end{split}
    \end{equation}
    where \(\ket{w_{\lambda}}\in\mathcal{W}_{\lambda,\textup{A}}\otimes \mathcal{W}_{\lambda,\textup{B}}\) is defined in the decomposition of $\ket{\psi}^{\otimes N}_{\textup{AB}}$ in \cref{lemma-3262205}.
    \end{enumerate}
\end{lemma}

\begin{proof}
\textbf{Item 1}. 
We first show that $\widetilde{\mathcal{T}}$ is a \((c,s)\)-tester for $\mathcal{P}$ with sample complexity $N$.
If \(\ket{\psi}_{\textup{AB}}\in\mathcal{P}^{\textup{yes}}\), then by the unitary invariance on \(\HB\),
\begin{equation*}
\begin{split}
\!c&\leq \!\!\int_{U\in\mathbb{U}_d}\!\!\!\tr\!\left[{\mathcal{T}}U_\textup{B}^{\otimes N}\ketbra{\psi}{\psi}^{\otimes N}_{\textup{AB}}U_\textup{B}^{\dag\otimes N} \right]
=\int_{U\in\mathbb{U}_d}\!\!\!\tr\!\left[U_\textup{B}^{\dag\otimes N} {\mathcal{T}}U_\textup{B}^{\otimes N}\ketbra{\psi}{\psi}^{\otimes N}_{\textup{AB}}\right]=\tr\!\left[\widetilde{\mathcal{T}} \ketbra{\psi}{\psi}^{\otimes N}_{\textup{AB}}\right],
\end{split}
\end{equation*}
Similarly, if \(\ket{\psi}_{\textup{AB}}\in\mathcal{P}^{\textup{no}}\), we have 
\begin{equation*}
\begin{split}
\!\!s&\geq \!\!\int_{U\in\mathbb{U}_d}\!\!\!\tr\!\left[{\mathcal{T}}U_\textup{B}^{\otimes N}\ketbra{\psi}{\psi}^{\otimes N}_{\textup{AB}}U_\textup{B}^{\dag\otimes N} \right] 
=\int_{U\in\mathbb{U}_d}\!\!\!\tr\!\left[U_\textup{B}^{\dag\otimes N} {\mathcal{T}}U_\textup{B}^{\otimes N}\ketbra{\psi}{\psi}^{\otimes N}_{\textup{AB}}\right] =  \tr\!\left[\widetilde{\mathcal{T}} \ketbra{\psi}{\psi}^{\otimes N}_{\textup{AB}}\right],
\end{split}
\end{equation*}
Therefore, \(\widetilde{\mathcal{T}}\) is also a tester for \(\mathcal{P}\).

\textbf{Item 2}. 
Now, we will derive a formula for the acceptance probability of $\widetilde{\mathcal{T}}$. 
Note that \(\widetilde{\mathcal{T}}\) is invariant under the conjugation of \(U_{\textup{B}}^{\otimes N}\) for \(U\in\mathbb{U}_d\), i.e.,
\begin{equation*}
U_{\textup{B}}^{\otimes N}\widetilde{\mathcal{T}}U_{\textup{B}}^{\dag \otimes N} =\widetilde{\mathcal{T}},
\end{equation*}
and thus
\begin{equation*}
\int_{U\in\mathbb{U}_d}U_{\textup{B}}^{\otimes N}\widetilde{\mathcal{T}}U_{\textup{B}}^{\dag \otimes N} =\widetilde{\mathcal{T}}.
\end{equation*}
Now, suppose \(\widetilde{\mathcal{T}}\) acts on \(N\) copies of a bipartite state \(\ket{\psi}_{\textup{AB}}\). We can see that
\begin{equation}\label{eq-3272137}
\tr\!\left[\widetilde{\mathcal{T}}\ketbra{\psi}{\psi}_{\textup{AB}}^{\otimes N}\right]=\tr\!\left[\widetilde{\mathcal{T}}\int_{U\in\mathbb{U}_d}U_{\textup{B}}^{\otimes N}\ketbra{\psi}{\psi}^{\otimes N}_{\textup{AB}}U_{\textup{B}}^{\dag \otimes N}\right].
\end{equation}
By \cref{lemma-3262205}, we have
\begin{equation}\label{eq-329053}
\ketbra{\psi}{\psi}^{\otimes N}_{\textup{AB}}=\sum_{\lambda_1,\lambda_2\vdash (N,d)}\begin{matrix}\kettbbra{I_{\mathcal{V}_{\lambda_1}}}{I_{\mathcal{V}_{\lambda_2}}}\\ \ketbra{w_{\lambda_1}}{w_{\lambda_2}}\end{matrix},
\end{equation}
where \(\ket{w_{\lambda}}\in\mathcal{W}_{\lambda,\textup{A}}\otimes \mathcal{W}_{\lambda,\textup{B}}\) and \(\kett{I_{\mathcal{V}_{\lambda}}}\in\mathcal{V}_{\lambda,\textup{A}}\otimes \mathcal{V}_{\lambda,\textup{B}}\) is a (non-normalized) maximally entangled state. Then,
\begin{equation}\label{eq-327126}
\begin{split}
&\int_{U\in\mathbb{U}_d}U_{\textup{B}}^{\otimes N}\ketbra{\psi}{\psi}^{\otimes N}_{\textup{AB}}U_{\textup{B}}^{\dag \otimes N}
=\sum_{\lambda_1,\lambda_2\vdash (N,d)}\int_{U\in\mathbb{U}_d}\begin{matrix}\kettbbra{I_{\mathcal{V}_{\lambda_1}}}{I_{\mathcal{V}_{\lambda_2}}}\\ Q_{\lambda_1,\TB}(U)\ketbra{w_{\lambda_1}}{w_{\lambda_2}}Q_{\lambda_2,\TB}(U)^\dag\end{matrix},
\end{split}
\end{equation}
where \(Q_{\lambda,\TB}(U)\) acts on \(\mathcal{W}_{\lambda,B}\) (i.e., \(Q_{\lambda,\TB}(U)\ket{w_{\lambda}}=Q_{\lambda,\TA}(I)\otimes Q_{\lambda,\TB}(U)\ket{w_{\lambda}}\)). If \(\lambda_1\neq \lambda_2\), then we can easily see that
\begin{equation*}
\int_{U\in\mathbb{U}_d} Q_{\lambda_1,\TB}(U)\ketbra{w_{\lambda_1}}{w_{\lambda_2}}Q_{\lambda_2,\TB}(U)^\dag=0,
\end{equation*}
by \cref{lemma-41116}, where \(\ketbra{w_{\lambda_1}}{w_{\lambda_2}}\) is treated as a linear map: \(\mathbb{C}^{\dim(\mathcal{W}_{\lambda_2})}\otimes \mathcal{W}_{\lambda_2}\rightarrow\mathbb{C}^{\dim(\mathcal{W}_{\lambda_1})}\otimes \mathcal{W}_{\lambda_1}\) between representation spaces.
Next, if \(\lambda_1=\lambda_2=\lambda\), by \cref{lemma-3262332},
\begin{equation*}
\begin{split}
&\int_{U\in\mathbb{U}_d} Q_{\lambda,\TB}(U)\ketbra{w_{\lambda}}{w_{\lambda}}Q_{\lambda,\TB}(U)^\dag 
=\tr_{\mathcal{W}_{\lambda,\textup{B}}}(\ketbra{w_{\lambda}}{w_{\lambda}})\otimes \frac{I_{\mathcal{W}_{\lambda,\TB}}}{\dim(\mathcal{W}_\lambda)},
\end{split}
\end{equation*}
where \(\tr_{\mathcal{W}_{\lambda,\textup{B}}}\) denotes the partial trace on \(\mathcal{W}_{\lambda,\textup{B}}\).
Therefore, from \cref{eq-327126}, we have
\begin{equation}\label{eq-45031}
\begin{split}
&\int_{U\in\mathbb{U}_d}U_{\textup{B}}^{\otimes N}\ketbra{\psi}{\psi}^{\otimes N}_{\textup{AB}}U_{\textup{B}}^{\dag \otimes N} 
=\bigoplus_{\lambda\vdash (N,d)}\begin{matrix}\kettbbra{I_{\mathcal{V}_{\lambda}}}{I_{\mathcal{V}_{\lambda}}}\\ \tr_{\mathcal{W}_{\lambda,\TB}}(\ketbra{w_{\lambda}}{w_{\lambda}})\otimes \frac{I_{\mathcal{W}_{\lambda,\TB}}}{\dim(\mathcal{W}_\lambda)}\end{matrix}.
\end{split}
\end{equation}
Therefore, combining \cref{eq-3272137} with \cref{eq-45031}, we can express \(\tr\left[\widetilde{\mathcal{T}}\ketbra{\psi}{\psi}_{\textup{AB}}^{\otimes N}\right]\) (i.e., the probability of tester \(\widetilde{\mathcal{T}}\) accepting state \(\ket{\psi}_{\textup{AB}}^{\otimes N}\)) as
\begin{align}
\tr\left[\widetilde{\mathcal{T}}\ketbra{\psi}{\psi}_{\textup{AB}}^{\otimes N}\right] 
=&\sum_{\lambda\vdash (N,d)}\tr\left[\widetilde{\mathcal{T}}\cdot\,\,\begin{matrix}\kettbbra{I_{\mathcal{V}_{\lambda}}}{I_{\mathcal{V}_{\lambda}}}\\ \tr_{\mathcal{W}_{\lambda,\TB}}(\ketbra{w_{\lambda}}{w_{\lambda}})\otimes \frac{I_{\mathcal{W}_{\lambda,\TB}}}{\dim(\mathcal{W}_\lambda)}\end{matrix}\right] \nonumber \\
=&\sum_{\lambda\vdash (N,d)}\tr\left[\widetilde{\mathcal{T}}_\lambda\cdot\,\,\begin{matrix}\tr_{\mathcal{W}_{\lambda,\TB}}(\ketbra{w_{\lambda}}{w_{\lambda}}) \\ \frac{1}{\dim(\mathcal{W}_\lambda)}I_{\mathcal{W}_{\lambda,\TB}}\end{matrix}\right], \label{eq-3272202}
\end{align}
where \(\widetilde{\mathcal{T}}_\lambda\coloneqq\bbra{I_{\mathcal{V}_\lambda}}\widetilde{\mathcal{T}}\kett{I_{\mathcal{V}_\lambda}} 
\colon \mathcal{W}_{\lambda,\TA}\otimes\mathcal{W}_{\lambda,\TB}\rightarrow\mathcal{W}_{\lambda,\TA}\otimes\mathcal{W}_{\lambda,\TB}\) is a linear map. Then, this probability can be further simplified to
\begin{equation}\label{eq-329117}
\eqref{eq-3272202}=\!\!\sum_{\lambda\vdash (N,d)} \!\!\frac{1}{\dim(\mathcal{W}_\lambda)} \tr\!\left[\tr_{\mathcal{W}_{\lambda,\TB}}\big[\widetilde{\mathcal{T}}_\lambda\big]\cdot\tr_{\mathcal{W}_{\lambda,\TB}}(\ketbra{w_{\lambda}}{w_{\lambda}}) \right],
\end{equation}
where \(\tr_{\mathcal{W}_{\lambda,\TB}}\big[\widetilde{\mathcal{T}}_\lambda\big] \colon \mathcal{W}_{\lambda,\TA}\rightarrow\mathcal{W}_{\lambda,\TA}\) is a linear map acting on Alice's \(\mathcal{W}_\lambda\).
\end{proof}

\subsection{\texorpdfstring{Validity of \(\hat{\mathcal{T}}\)}{Validity of hat-T}}\label{sec-417207}

In this subsection, we show that $\hat{\mathcal{T}}$ is a tester for $\mathcal{P}$.

\begin{lemma} 
    Let $\hat{\mathcal{T}}$ be the tester defined by \cref{eq:def-hatT}.
    Then, $\hat{\mathcal{T}}$ is a \((c,s)\)-tester for $\mathcal{P}$ with sample complexity $N$.
\end{lemma}

\begin{proof}
As the tester $\hat{\mathcal{T}}$ acts trivially on Bob's part, we just consider the behavior of \(\hat{\mathcal{T}}\) on Alice's part:
\begin{equation}\label{eq-329110}
\hat{\mathcal{T}}=\bigoplus_{\lambda\vdash (N,d)}\frac{1}{\dim(\mathcal{W}_\lambda)\dim(\mathcal{V}_\lambda)}\,\cdot\,I_{\mathcal{V}_{\lambda,\TA}}\otimes \tr_{\mathcal{W}_{\lambda,\TB}}\big[\widetilde{\mathcal{T}}_\lambda\big].
\end{equation}
To show that \(\hat{\mathcal{T}}\) is a valid measurement, i.e., \(\hat{\mathcal{T}}\sqsubseteq I\), it suffices to show that \(\|\tr_{\mathcal{W}_{\lambda,\TB}}\big[\widetilde{\mathcal{T}}_\lambda\big]\|\leq \dim(\mathcal{W}_\lambda)\dim(\mathcal{V}_\lambda)\). 
Since \(\widetilde{\mathcal{T}}\sqsubseteq I\), we have
\begin{equation*}
\left\|\widetilde{\mathcal{T}}_\lambda\right\|=\left\|\bbra{I_{\mathcal{V}_\lambda}}\widetilde{\mathcal{T}}\kett{I_{\mathcal{V}_\lambda}} \right\|\leq \bbrakett{I_{\mathcal{V}_\lambda}}{I_{\mathcal{V}_\lambda}}=\dim(\mathcal{V}_\lambda),
\end{equation*}
where \(\|\cdot\|\) is the operator norm. Next, we have
\begin{align}
\left\|\tr_{\mathcal{W}_{\lambda,\TB}}\big[\widetilde{\mathcal{T}}_\lambda\big]\right\|&=\left\|\sum_{\ket{x}\in\mathcal{W}_{\lambda,\TB}} \bra{x}\widetilde{\mathcal{T}}_{\lambda}\ket{x}\right\|\leq \sum_{\ket{x}\in\mathcal{W}_{\lambda,\TB}}\left\| \bra{x}\widetilde{\mathcal{T}}_{\lambda}\ket{x}\right\| \nonumber \\
&\leq\dim(\mathcal{V}_\lambda)\cdot\dim(\mathcal{W}_\lambda), \nonumber 
\end{align}
where the summation is over an orthonormal basis \(\{\ket{x}\}\) of \(\mathcal{W}_{\lambda,\TB}\) and therefore \(\ket{x}\) is a vector on Bob's part.

Next, we show that \(\hat{\mathcal{T}}\) is a local tester for the property \(\mathcal{P}\) on Alice's part (in fact, the local tester \(\hat{\mathcal{T}}\) behaves exactly the same as the global tester \(\widetilde{\mathcal{T}}\)). First, we can see from \cref{eq-329053} that
\begin{align}
\tr_{\textup{B}}\left[\ketbra{\psi}{\psi}_{\textup{AB}}^{\otimes N}\right] 
=&\tr_{\textup{B}}\left[\sum_{\lambda_1,\lambda_2\vdash (N,d)}\begin{matrix}\kettbbra{I_{\mathcal{V}_{\lambda_1}}}{I_{\mathcal{V}_{\lambda_2}}}\\ \ketbra{w_{\lambda_1}}{w_{\lambda_2}}\end{matrix}\right] \nonumber \\
=&\sum_{\substack{\lambda\vdash (N,d)\\ \ket{v_\lambda}\in\mathcal{V}_{\lambda,\TB}\\ \ket{w_\lambda}\in\mathcal{W}_{\lambda,\TB}}}\begin{matrix}\bra{v_\lambda}\\\bra{w_\lambda}\end{matrix}\left[\sum_{\lambda_1,\lambda_2\vdash (N,d)}\begin{matrix}\kettbbra{I_{\mathcal{V}_{\lambda_1}}}{I_{\mathcal{V}_{\lambda_2}}}\\ \ketbra{w_{\lambda_1}}{w_{\lambda_2}}\end{matrix}\right]\begin{matrix}\ket{v_\lambda}\\\ket{w_\lambda}\end{matrix}, \label{eq-329103}
\end{align}
where \(\{\ket{v_\lambda}\}\) and \(\{\ket{w_\lambda}\}\) are orthonormal bases of \(\mathcal{V}_{\lambda,\TB}\) and \(\mathcal{W}_{\lambda,\TB}\), respectively. Then, we can see that
\begin{equation}\label{eq-329109}
\eqref{eq-329103}=\bigoplus_{\lambda\vdash (N,d)} \begin{matrix} I_{\mathcal{V}_{\lambda,\TA}}\\ \tr_{\mathcal{W}_{\lambda,\TB}}\left(\ketbra{w_\lambda}{w_\lambda}\right) \end{matrix}.
\end{equation}
Therefore, combining \cref{eq-329109} with \cref{eq-329110}, we have
\begin{align}
\tr\!\left(\hat{\mathcal{T}}\tr_{\textup{B}}\left[\ketbra{\psi}{\psi}_{\textup{AB}}^{\otimes N}\right]\right)
=&\sum_{\lambda\vdash (N,d)}\frac{1}{\dim(\mathcal{W}_\lambda)\dim(\mathcal{V}_\lambda)} 
 \cdot \tr\!\left[\begin{matrix}I_{\mathcal{V}_{\lambda,\TA}} \\ \tr_{\mathcal{W}_{\lambda,\TB}}\!\!\big[\widetilde{\mathcal{T}}_\lambda\big]\cdot\tr_{\mathcal{W}_{\lambda,\TB}}\!\left(\ketbra{w_\lambda}{w_\lambda}\right) \end{matrix}\right] \nonumber \\
=&\sum_{\lambda\vdash (N,d)} \!\!\frac{1}{\dim(\mathcal{W}_\lambda)} \tr\!\left[\tr_{\mathcal{W}_{\lambda,\TB}}\!\big[\widetilde{\mathcal{T}}_\lambda\big]\cdot\tr_{\mathcal{W}_{\lambda,\TB}}(\ketbra{w_{\lambda}}{w_{\lambda}}) \right]\!,\! \nonumber
\end{align}
which is exactly the same as the RHS of \cref{eq-329117} (and also \cref{eq-414052}). This means the local tester \(\hat{\mathcal{T}}\) behaves identically to the tester \(\widetilde{\mathcal{T}}\). Thus \(\hat{\mathcal{T}}\) is a local \((c,s)\)-tester for \(\mathcal{P}\) with sample complexity \(N\).
\end{proof}

\subsection{Technical Lemmas} \label{sec:tech-lemmas}

\begin{lemma}\label{lemma-3252341}
Suppose \(\lambda\vdash N\) is a Young diagram with \(N\) boxes, and \((P_\lambda,\mathcal{V}_\lambda)\) is the irreducible representation of \(\mathbb{S}_N\) corresponding to \(\lambda\). Then, we have
\begin{equation}\label{eq-3252235}
\frac{1}{N!}\sum_{\pi\in\mathbb{S}_N} P_{\lambda}(\pi)\otimes P_{\lambda}(\pi)=\frac{1}{\dim(\mathcal{V}_\lambda)}\kettbbra{I_{\mathcal{V}_\lambda}}{I_{\mathcal{V}_\lambda}},
\end{equation}
where \(\kett{I_{\mathcal{V}_\lambda}}\) is a (non-normalized) maximally entangled state on \(\mathcal{V}_\lambda\otimes\mathcal{V}_\lambda\).
Moreover, for \(\lambda_1,\lambda_2\vdash N,\lambda_1\neq \lambda_2\), we have
\begin{equation}\label{eq-3252254}
\frac{1}{N!}\sum_{\pi\in\mathbb{S}_N} P_{\lambda_1}(\pi)\otimes P_{\lambda_2}(\pi)=0.
\end{equation}
\end{lemma}

\begin{proof}
Since \(P_\lambda(\pi)\) is a real-valued matrix (see \cref{remark-3252205}), we have
\begin{equation*}
\frac{1}{N!}\sum_{\pi\in\mathbb{S}_N}P_{\lambda}(\pi)\otimes P_{\lambda}(\pi)=\frac{1}{N!}\sum_{\pi\in\mathbb{S}_N}P_{\lambda}(\pi)\otimes P_{\lambda}(\pi)^*.
\end{equation*}
It is easy to see that this sum results in an orthogonal projector onto the subspace that is invariant under the action \(P_{\lambda}(\pi)\otimes P_{\lambda}(\pi)^*\) for all \(\pi\in\mathbb{S}_N\) (see, e.g., the proof of Proposition 1 in \cite{harrow2013church}). Due to the following equivalence: for any $\pi$,
\begin{equation}\label{eq-3252233}
P_{\lambda}(\pi)\otimes P_{\lambda}(\pi)^* \kett{K}\!=\!\kett{K}\Longleftrightarrow P_{\lambda}(\pi) K P_{\lambda}(\pi)^{\dag}\!=\!K,
\end{equation}
we can thus consider the fixed point problem on the RHS of \cref{eq-3252233}. Then, by Schur's lemma (\cref{prop:202404032229}), the only possible solution is \(K=cI_\lambda\). Therefore, \cref{eq-3252235} follows immediately.

Similarly, \cref{eq-3252254} results in an orthogonal projector onto the subspace that is invariant under the action \(P_{\lambda_1}(\pi)\otimes P_{\lambda_2}(\pi)^*\). Thus we consider the fixed point problem \(\forall \pi, P_{\lambda_1}(\pi)K\mu_{\lambda_2}(\pi)^\dag=K\), which, again by the Schur's lemma (\cref{prop:202404032229}), only has the trivial solution \(K=0\) since \(\lambda_1,\lambda_2\) are different irreducible representations. Then, \cref{eq-3252254} follows immediately.
\end{proof}

\begin{lemma}[see, e.g., \cite{matsumoto2007universal}]\label{lemma-3262205}
Let \(\ket{\psi}_{\textup{AB}}\in \HAB\) be a bipartite state, then $\ket{\psi}_{\textup{AB}}^{\otimes N}$ can be written as
\begin{equation}\label{eq-326050}
\ket{\psi}_{\textup{AB}}^{\otimes N}= \bigoplus_{\lambda\vdash (N,d)}\kett{I_{\mathcal{V}_\lambda}}\otimes\ket{w_\lambda},
\end{equation}
where \(\kett{I_{\mathcal{V}_\lambda}}\) is a (non-normalized) maximally entangled state on \(\mathcal{V}_{\lambda,\TA}\otimes \mathcal{V}_{\lambda,\TB}\) and \(\ket{w_\lambda}\in\mathcal{W}_{\lambda,\TA}\otimes \mathcal{W}_{\lambda,\TB}\).
\end{lemma}

\begin{proof}
We first decompose \(\ket{\psi}^{\otimes N}\) along those irreducible representations in \cref{lemma-3241548}
\begin{equation*}
\ket{\psi}^{\otimes N}=\bigoplus_{\lambda_1,\lambda_2\vdash (N,d)} \ket{\psi_{\lambda_1,\lambda_2}},
\end{equation*}
where $\ket{\psi_{\lambda_1,\lambda_2}}\in\mathcal{V}_{\lambda_1,\TA}\otimes\mathcal{V}_{\lambda_2,\TB} \otimes \mathcal{W}_{\lambda_1,\TA} \otimes \mathcal{W}_{\lambda_2,\TB}$.
Note that \(\ket{\psi}^{\otimes N}\) is invariant under the action \(P_\TA(\pi)\otimes P_\TB(\pi)\) for \(\pi \in \mathbb{S}_N\) (as a subgroup of \(\mathbb{S}_N\times \mathbb{S}_N\)), where \(P_\TA(\pi)\) and \(P_\TB(\pi)\) act on \(\HA^{\otimes N}\) and \(\HB^{\otimes N}\), respectively. Therefore, for every Young diagram \(\lambda_1,\lambda_2\vdash (N,d)\), we have
\begin{equation*}
P_{\lambda_1,\TA}(\pi)\otimes P_{\lambda_2,\TB}(\pi)\ket{\psi_{\lambda_1,\lambda_2}}=\ket{\psi_{\lambda_1,\lambda_2}},
\end{equation*}
where \(P_{\lambda_1,\TA},P_{\lambda_2,\TB}\) are irreducible representations of \(\mathbb{S}_N\) corresponding to \(\lambda_1,\lambda_2\), respectively, and acting on Alice's system and Bob's system, respectively. This means
\begin{equation*}
\frac{1}{N!}\sum_{\pi\in\mathbb{S}_N} P_{\lambda_1,\TA}(\pi)\otimes P_{\lambda_2,\TB}(\pi)\ket{\psi_{\lambda_1,\lambda_2}}=\ket{\psi_{\lambda_1,\lambda_2}}.
\end{equation*}
By \cref{lemma-3252341}, we see that \(\ket{\psi_{\lambda_1,\lambda_2}}=0\) for \(\lambda_1\neq\lambda_2\). Moreover, for \(\lambda_1=\lambda_2=\lambda\), we have
\begin{equation*}
\frac{1}{\dim(\mathcal{V}_\lambda)}\kettbbra{I_{\mathcal{V}_\lambda}}{I_{\mathcal{V}_\lambda}} \otimes I_{\mathcal{W}_{\lambda,\TA}}\otimes I_{\mathcal{W}_{\lambda,\TB}}\,\,\ket{\psi_{\lambda,\lambda}}=\ket{\psi_{\lambda,\lambda}}.
\end{equation*}
which means \(\ket{\psi_{\lambda,\lambda}}=\kett{I_{\mathcal{V}_\lambda}}\ket{w_\lambda}\) for some \(\ket{w_\lambda}\in\mathcal{W}_{\lambda,\TA}\otimes\mathcal{W}_{\lambda,\TB}\). Therefore, \cref{eq-326050} follows.
\end{proof}

\subsection{Optimal Local Tester in Average Case}
In this section, we use the same techniques in the proof of \cref{theorem-4222259} to give another result with respect to average-case performance.
To this end, we define the notion of unitarily invariant probability distributions on properties of quantum states, and extend the quantum testers to the average case.
\begin{definition} [Unitarily invariant distributions]
    Let $D$ be a probability distribution on bipartite pure states in \(\HAB\). 
    Then, $D$ is said to be unitarily invariant on $\HB$, if for any measurable subset \(S\) of bipartite pure states in \(\HAB\) 
    and any \(U\in\mathbb{U}_d\),
\begin{equation*}
\Pr_{\ket{\psi}_{\textup{AB}} \sim D} \sbra[\Big]{\ket{\psi}_{\textup{AB}} \in S}\! =\! \Pr_{\ket{\psi}_{\textup{AB}} \sim D} \sbra[\Big]{\ket{\psi}_{\textup{AB}} \in \rbra*{I_\TA \otimes U_\TB} S},
\end{equation*}
where \(\rbra{I_\TA \otimes U_\TB} S = \set{\rbra{I_\TA \otimes U_\TB}\ket{\psi}_{\textup{AB}}}{\ket{\psi}_{\textup{AB}}\in S}\).
\end{definition}
\begin{definition} [Average-case testers]
    Let $\mathcal{P} = \rbra{\mathcal{P}^{\textup{yes}}, \mathcal{P}^{\textup{no}}}$ be a property of bipartite pure states, and $D^{\textup{yes}}$ and $D^{\textup{no}}$ be probability distributions on $\mathcal{P}^{\textup{yes}}$ and $\mathcal{P}^{\textup{no}}$, respectively. 
    For $0 \leq s < c \leq 1$, a tester $\mathcal{T}$ is called an average-case $\rbra{c, s}$-tester for $\mathcal{P}$ with respect to  $\rbra{D^{\textup{yes}},D^{\textup{no}}}$ with sample complexity $N$, if 
    \begin{equation*} 
    \begin{split}
    \mathop{\mathbb{E}}_{\ket{\psi}_\TAB\sim D^{\textup{yes}}}\sbra*{\tr\rbra*{\mathcal{T}\ketbra{\psi}{\psi}_{\TAB}^{\otimes N}}}=c, \qquad
    \mathop{\mathbb{E}}_{\ket{\phi}_\TAB\sim D^{\textup{no}}}\sbra*{\tr\rbra*{\mathcal{T}\ketbra{\phi}{\phi}_{\TAB}^{\otimes N}}}=s.
    \end{split}
    \end{equation*}
    Moreover, if $c = 2/3$ and $s = 1/3$, simply call $\mathcal{T}$ an average-case tester. 
\end{definition}
Then, we have the following result which states that the construction in \cref{eq:def-hatT} also leads to the optimality of local tester in an average-case sense.
\begin{theorem} [Optimal average-case local tester] \label{theorem-424241}
Suppose \(\mathcal{P}=(\mathcal{P}^{\textup{yes}},\mathcal{P}^{\textup{no}})\) is a property of bipartite pure states in \(\HAB\) that is unitarily invariant on \(\HB\), and suppose \(D^{\textup{yes}},D^{\textup{no}}\) are probability distributions on \(\mathcal{P}^{\textup{yes}},\mathcal{P}^{\textup{no}}\) that are unitarily invariant on \(\HB\). 
Let $\mathcal{T}$ be an average-case $\rbra{c, s}$-tester for $\mathcal{P}$ with respect to $\rbra{D^{\textup{yes}},D^{\textup{no}}}$ with sample complexity $N$ for some parameters \(0\leq s<c\leq 1\).
Then, there is a local tester \(\hat{\mathcal{T}}\) on \(\HA\) with sample complexity \(N\) that is also an average-case $\rbra{c, s}$-tester for $\mathcal{P}$ with respect to $\rbra{D^{\textup{yes}},D^{\textup{no}}}$.
\end{theorem}
\begin{proof}
Since \(D^{\textup{yes}}\) is unitarily invariant on \(\HB\), then it is easy to see that
\begin{equation*}
\mathop{\mathbb{E}}_{\ket{\psi}_\TAB\sim D^{\textup{yes}}}\!\sbra*{\ketbra{\psi}{\psi}_{\TAB}^{\otimes N}}\!\!=\!\mathop{\mathbb{E}}_{\ket{\psi}_\TAB\sim D^{\textup{yes}}}\!\sbra*{\int_{U\in\mathbb{U}_d}\!\!\! U_\TB^{\otimes N}\ketbra{\psi}{\psi}_{\TAB}^{\otimes N}U_\TB^{\dag\otimes N}\!}\!.
\end{equation*}
For any \(\ket{\psi}_{\textup{AB}}\in \mathcal{P}^{\textup{yes}}\), by \cref{lemma-3262205}, we can write \(\ket{\psi}_{\textup{AB}}^{\otimes N}=\sum_{\lambda\vdash (N,d)}\kett{I_{\mathcal{V}_\lambda}}\ket{w_{\lambda,\psi}}\).
Through a similar argument in the proof of \cref{lemma:202404041522} (from \cref{eq-327126} to \cref{eq-45031}), 
we have
\begin{equation*}
\begin{split}
&\int_{U\in\mathbb{U}_d}U_{\TB}^{\otimes N} \ketbra{\psi}{\psi}^{\otimes N}_\TAB U_\TB^{\dag\otimes N}  
=\bigoplus_{\lambda\vdash (N,d)} \begin{matrix}\kettbbra{I_{\mathcal{V}_{\lambda}}}{I_{\mathcal{V}_{\lambda}}}\\ \tr_{\mathcal{W}_{\lambda,\TB}}(\ketbra{w_{\lambda,\psi}}{w_{\lambda,\psi}})\otimes \frac{I_{\mathcal{W}_{\lambda,\TB}}}{\dim(\mathcal{W}_\lambda)}\end{matrix},
\end{split}
\end{equation*}
Then,
\begin{equation*}
\begin{split}
&\mathop{\mathbb{E}}_{\ket{\psi}_\TAB\sim D^{\textup{yes}}}\sbra*{\ketbra{\psi}{\psi}_{\TAB}^{\otimes N}}
=\bigoplus_{\lambda\vdash (N,d)} \begin{matrix}\kettbbra{I_{\mathcal{V}_{\lambda}}}{I_{\mathcal{V}_{\lambda}}}\\ \mathop{\mathbb{E}}\limits_{\ket{\psi}_{\TAB}\sim D^{\textup{yes}}}\sbra*{\tr_{\mathcal{W}_{\lambda,\TB}}(\ketbra{w_{\lambda,\psi}}{w_{\lambda,\psi}})}\otimes \frac{I_{\mathcal{W}_{\lambda,\TB}}}{\dim(\mathcal{W}_\lambda)}\end{matrix}.
\end{split}
\end{equation*}
Similarly, we have
\begin{equation*}
\begin{split}
&\mathop{\mathbb{E}}_{\ket{\phi}_\TAB\sim D^{\textup{no}}}\sbra*{\ketbra{\phi}{\phi}_{\TAB}^{\otimes N}} 
=\bigoplus_{\lambda\vdash (N,d)} \begin{matrix}\kettbbra{I_{\mathcal{V}_{\lambda}}}{I_{\mathcal{V}_{\lambda}}}\\ \mathop{\mathbb{E}}\limits_{\ket{\phi}_{\TAB}\sim D^{\textup{no}}}\sbra*{\tr_{\mathcal{W}_{\lambda,\TB}}(\ketbra{w_{\lambda,\phi}}{w_{\lambda,\phi}})}\otimes \frac{I_{\mathcal{W}_{\lambda,\TB}}}{\dim(\mathcal{W}_\lambda)}\end{matrix}.
\end{split}
\end{equation*}
Therefore, one can easily verify that the construction of \(\hat{\mathcal{T}}\) in \cref{eq:def-hatT} is a local tester on \(\HA\) that achieves the same average-case performance as \(\mathcal{T}\).
\end{proof}

\cref{theorem-424241} directly leads to the following result regarding the trace distance between average states.
\begin{corollary}\label{corollary-4241439}
Suppose \(\mathcal{P}=(\mathcal{P}^{\textup{yes}},\mathcal{P}^{\textup{no}})\), \(D^{\textup{yes}}\), \(D^{\textup{no}}\) are as defined in \cref{theorem-424241}. Let
\begin{equation}\label{eq-4242359}
\rho\coloneqq \mathop{\mathbb{E}}_{\ket{\psi}_\TAB\sim D^{\textup{yes}}}\sbra*{\ketbra{\psi}{\psi}_{\TAB}^{\otimes N}},\qquad
\sigma\coloneqq \mathop{\mathbb{E}}_{\ket{\phi}_\TAB\sim D^{\textup{no}}}\sbra*{\ketbra{\phi}{\phi}_{\TAB}^{\otimes N}}.
\end{equation}
Then, \(d_{\tr}(\rho,\sigma)=d_{\tr}(\tr_\TB[\rho],\tr_{\TB}[\sigma])\).
\end{corollary}
\begin{proof}
    By the Helstrom-Holevo bound (\cref{thm:HH-bound}), there is a tester $\mathcal{T}$ (note that $\cbra{\mathcal{T}, I-\mathcal{T}}$ is a POVM) such that 
    \begin{equation*}
        \tr\sbra*{\mathcal{T} \rho} - \tr\sbra*{\mathcal{T} \sigma} =  d_{\tr}\rbra{\rho, \sigma}.
    \end{equation*}
    By \cref{theorem-424241}, there is a tester $\hat{\mathcal{T}} = \hat{\mathcal{T}}_\TA \otimes I_\TB$ such that $\tr\sbra{\hat{\mathcal{T}}_\TA \tr_\TB\sbra{\rho}} = \tr\sbra{\hat{\mathcal{T}} \rho} = \tr\sbra{\mathcal{T} \rho}$ and $\tr\sbra{\hat{\mathcal{T}}_\TA \tr_\TB\sbra{\sigma}} = \tr\sbra{\hat{\mathcal{T}} \sigma} = \tr\sbra{\mathcal{T} \sigma}$.
    Then,  
    \begin{equation} \label{eq:04250008}
        \tr\sbra*{\hat{\mathcal{T}}_\TA \tr_\TB\sbra{\rho}} - \tr\sbra*{\hat{\mathcal{T}}_\TA \tr_\TB\sbra{\sigma}} = d_{\tr}\rbra{\rho, \sigma}.
    \end{equation}
    Again by the Helstrom-Holevo bound (\cref{thm:HH-bound}), we have
    \begin{equation} \label{eq:04250009}
        \tr\sbra*{\hat{\mathcal{T}}_\TA \tr_\TB\sbra{\rho}} - \tr\sbra*{\hat{\mathcal{T}}_\TA \tr_\TB\sbra{\sigma}} \leq  d_{\tr}(\tr_\TB[\rho],\tr_{\TB}[\sigma]).
    \end{equation}
    By \cref{eq:04250008,eq:04250009}, we have $d_{\tr}\rbra{\tr_\TB[\rho],\tr_{\TB}[\sigma]} \geq d_{\tr}\rbra{\rho, \sigma}$.
    
    On the other hand, $d_{\tr}\rbra{\tr_\TB[\rho],\tr_{\TB}[\sigma]} \leq d_{\tr}\rbra{\rho, \sigma}$ due to the contractivity of partial trace.
    Therefore, we conclude that \(d_{\tr}(\rho,\sigma)=d_{\tr}(\tr_\TB[\rho],\tr_{\TB}[\sigma])\).
\end{proof}

In addition, the prior result in \cite[Theorem 5.2]{soleimanifar2022testing} (which was used to prove their lower bound for testing MPS) can be viewed as a direct application of \cref{theorem-424241}.
\begin{corollary} [{\cite[Theorem 5.2]{soleimanifar2022testing}}]\label{corollary-4262150}
    Suppose $\ket{\psi}_{\textup{AB}}, \ket{\phi}_{\textup{AB}} \in \HAB$ and let
    \begin{equation*}
    \begin{split}
        \rho =  \int_{U \in \mathbb{U}_d} \int_{V \in \mathbb{U}_d} \left(\rbra{U_{\TA} \otimes V_{\TB}} \cdot \ketbra{\psi}{\psi}_{\textup{AB}} \cdot \rbra{U_{\TA}^\dag \otimes V_{\TB}^\dag}\right)^{\otimes N},\\
        \sigma =  \int_{U \in \mathbb{U}_d} \int_{V \in \mathbb{U}_d} \left( \rbra{U_{\TA} \otimes V_{\TB}} \cdot \ketbra{\phi}{\phi}_{\textup{AB}} \cdot \rbra{U_{\TA}^\dag \otimes V_{\TB}^\dag}\right)^{\otimes N}.
    \end{split}
    \end{equation*}
    If $\rho \neq \sigma$, then any measurement for distinguishing between \(\rho\) and \(\sigma\) can, without loss of generality, act only on either \(\HA\) or \(\HB\).
\end{corollary}

\begin{proof}
    Let \(\mathcal{P}=(\mathcal{P}^{\textup{yes}},\mathcal{P}^{\textup{no}})\) where 
    \begin{equation*}
    \begin{split}
        &\mathcal{P}^{\textup{yes}} = \set{\rbra{U_\TA \otimes V_\TB}\ket{\psi}_{\TAB}}{U, V \in \mathbb{U}_d}, 
        \qquad \mathcal{P}^{\textup{no}} = \set{\rbra{U_\TA \otimes V_\TB}\ket{\phi}_{\TAB}}{U, V \in \mathbb{U}_d}.
    \end{split}
    \end{equation*}
    Let $D^{\textup{yes}}$ and $D^{\textup{no}}$ be the probability distributions on $\mathcal{P}^{\textup{yes}}$ and $\mathcal{P}^{\textup{no}}$, respectively, induced by the Haar measure over $\mathbb{U}_d \times \mathbb{U}_d$. 
    Note that \(D^{\textup{yes}},D^{\textup{no}}\) are unitarily invariant both on \(\HA\) and on \(\HB\).
    Then by \cref{theorem-424241}, for any average-case tester $\mathcal{T}$, there is a local tester $\hat{\mathcal{T}}_\TA$ on $\HA$ and a local tester $\hat{\mathcal{T}}_\TB$ on $\HB$ with the same sample complexity and the same average-case performance.
\end{proof}

Note that \cref{theorem-424241} applies to any distribution that is unitarily invariant only on one part, which can be further used to give a stronger lower bound for testing MPS (see \cref{lemma-4242204}).

\section{Quantum Sample Lower Bounds}
\label{sec:sample-bounds}

In this section, we derive a series of sample lower bounds for the property testing of bipartite quantum states through the tools proposed in \cref{sec:opt-local-tester}.

\subsection{Schmidt Rank}

We consider the problem of testing whether the Schmidt rank of a bipartite pure state $\ket{\psi}_{\textup{AB}}$ is at most $r$. Formally, we define
\begin{equation*}
\begin{split}
\textsc{SchmidtRank}^r = 
\set{\!\ket{\psi}_{\textup{AB}} \!\in\! \HAB}{\ket{\psi}_{\textup{AB}}\textup{ has Schmidt rank at most } r\!}\!.\!
\end{split}
\end{equation*}
We have the lower bounds for testing the Schmidt rank. 
\begin{corollary} [Schmidt rank] \label{thm:schmidt-rank}
    Any tester for determining whether the Schmidt rank of a bipartite pure state is at most $r$ or $\varepsilon$-far (in trace distance) requires sample complexity $\mathsf{S}\rbra{\textsc{SchmidtRank}^r_\varepsilon} = \Omega\rbra{r/\varepsilon^2}$. If perfect completeness is required, 
    then we have the lower bound \(\mathsf{S}_1\rbra{\textsc{SchmidtRank}^r_\varepsilon} = \Omega(r^2/\varepsilon^2)\). 
\end{corollary}

To prove \cref{thm:schmidt-rank}, we need the Eckart-Young theorem. 
\begin{lemma} [Eckart-Young  \cite{eckart1936approximation}, see also {\cite[Lemma 3.1]{soleimanifar2022testing}}] \label{lemma:Eckart-Young}
    Let $\ket{\psi}_{\textup{AB}} \in \HAB \coloneqq \HA \otimes \HB$ with $\Abs{\ket{\psi}_{\textup{AB}}} = 1$ where $\HA = \HB = \mathbb{C}^d$.
    Suppose that the Schmidt decomposition of $\ket{\psi}_{\textup{AB}}$ is
    \begin{equation*}
        \ket{\psi}_{\textup{AB}} = \sum_{j=1}^d \sqrt{\lambda_j} \ket{\phi_j}_{\textup{A}} \ket{\gamma_j}_{\textup{B}},
    \end{equation*}
    where $\lambda_1 \geq \lambda_2 \geq \dots \geq \lambda_d$.
    Then, 
    \begin{equation*}
        \max_{\ket{\eta}_{\textup{AB}} \in \textsc{SchmidtRank}^r} \abs*{\braket{\eta}{\psi}_{\textup{AB}}}^2 = \sum_{j=1}^r \lambda_j.
    \end{equation*}
\end{lemma}

Then, we can prove \cref{thm:schmidt-rank} as follows.

\begin{proof} [Proof of \cref{thm:schmidt-rank}]
    For the special case that $r = 1$, see the proof of \cref{thm:product}.
    In the following, we only have to consider the case when $r \geq 2$. 
    We relate the bipartite state testing of $\textsc{SchmidtRank}^r$ to the corresponding mixed state testing of $\textsc{Rank}^r$.
    Specifically, let 
    \begin{equation*}
        \textsc{Rank}^r = \set{\rho_{\textup{A}} \in \mathcal{D}\rbra{\HA}}{\rank\rbra{\rho_{\textup{A}}} \leq r}.
    \end{equation*}
    In the following, we first show that $\mathsf{Purify}\rbra{\textsc{Rank}^r_{\varepsilon^2}} \subseteq \textsc{SchmidtRank}^r_\varepsilon$.
    
    \textbf{Reduction}.
    Let $\mathcal{P}=\textsc{Rank}^r$ and
    $\mathcal{Q}=\textsc{SchmidtRank}^r$.
    It suffices to prove $\mathsf{Purify}\rbra{\mathcal{P}}\subseteq\mathcal{Q}$,
    and $\mathsf{Purify}\rbra{\mathcal{P}^{\geq \varepsilon^2}}\subseteq \mathcal{Q}^{\geq \varepsilon}$.
    \begin{enumerate}
        \item 
            $\mathsf{Purify}\rbra*{\mathcal{P}}\subseteq \mathcal{Q}$
            can be easily seen from the definition of Schmidt rank.
            That is, for any mixed state $\rho_{\textup{A}}\in \mathcal{D}(\HA)$
            and any of its purification $\ket{\psi}_{\textup{AB}}\in \HAB$ with $\tr_{\textup{B}}\rbra{\ketbra{\psi}{\psi}_{\textup{AB}}}=\rho_{\textup{A}}$,
            we have $\rank\rbra{\rho_{\textup{A}}}\leq r$ if and only if the Schmidt rank of $\ket{\psi}_{\textup{AB}}$ is $\leq r$.
        \item
            To prove $\mathsf{Purify}\rbra{\mathcal{P}^{\geq \varepsilon^2}}\subseteq \mathcal{Q}^{\geq \varepsilon}$, suppose 
            \begin{equation*}
                \rho_\TA=\sum_{i=1}^d\lambda_i \ketbra{\psi_i}{\psi_i}_\TA\in \mathcal{P}^{\geq \varepsilon^2}.
            \end{equation*}
            We define a state 
            \begin{equation*}
                \rho_{\TA,r} = \frac{1}{\sum_{i=1}^r \lambda_i}\sum_{i=1}^r \lambda_i\ketbra{\psi_i}{\psi_i}_\TA \in \mathcal{P}.
            \end{equation*}
            Then, since \(\rho_{\TA,r}\in\mathcal{P}\) and \(\rho_\TA\in\mathcal{P}^{\geq \varepsilon^2}\), we have
            \begin{equation}\label{eq-414011}
            \varepsilon^2\leq d_{\tr}(\rho_{\TA,r},\rho_\TA)=1-\sum_{i=1}^r\lambda_i.
            \end{equation}
            Suppose \(\ket{\phi}_{\textup{AB}}\in \HAB\) is a purification of \(\rho_\TA\). Then 
            \begin{equation*}
                \ket{\phi}_{\textup{AB}}=\sum_{i=1}^d\sqrt{\lambda_i}\ket{\psi_i}_\TA\ket{\gamma_i}_\TB
            \end{equation*}
            for some orthonormal basis \(\{\ket{\gamma_i}_\TB\}\) of \(\HB\). 
            By the Eckart-Young theorem (\cref{lemma:Eckart-Young}), we have \(
                \max_{\ket{\eta}_{\textup{AB}}\in\mathcal{Q}} \abs*{ \braket{\eta}{\phi}_{\textup{AB}}}^2=\sum_{i=1}^r\lambda_i
            \),
            which, together with \cref{eq-414011}, implies
            \begin{equation*}
            \min_{\ket{\eta}_{\textup{AB}}\in\mathcal{Q}} d_{\tr}(\ket{\eta}_{\textup{AB}},\ket{\phi}_{\textup{AB}})=\sqrt{1-\sum_{i=1}^r\lambda_i}\,\geq\, \varepsilon.
            \end{equation*}
            This proves that \(\mathsf{Purify}(\mathcal{P}^{\geq \varepsilon^2})\subseteq \mathcal{Q}^{\geq \varepsilon}\).
    \end{enumerate}
    From the above, we conclude that $\mathsf{Purify}\rbra{\textsc{Rank}^r_{\varepsilon^2}} \subseteq \textsc{SchmidtRank}^r_\varepsilon$.
    
    \textbf{Lower bound}. 
    Because $\mathsf{Purify}\rbra{\textsc{Rank}^r_{\varepsilon^2}} \subseteq \textsc{SchmidtRank}^r_\varepsilon$, and by \cref{fact:SPSQ},
    we immediately have $\mathsf{S}(\textsc{SchmidtRank}^r_{\varepsilon}) \geq \mathsf{S}\rbra{\mathsf{Purify}\rbra{\textsc{Rank}^r_{\varepsilon^2}}}$ and $\mathsf{S}_1(\textsc{SchmidtRank}^r_{\varepsilon}) \geq \mathsf{S}_1\rbra{\mathsf{Purify}\rbra{\textsc{Rank}^r_{\varepsilon^2}}}$.
    We also note that $\mathsf{S}\rbra{\mathsf{Purify}\rbra{\textsc{Rank}^r_{\varepsilon^2}}}=\mathsf{S}\rbra{\textsc{Rank}^r_{\varepsilon^2}}$ and $\mathsf{S}_1\rbra{\mathsf{Purify}\rbra{\textsc{Rank}^r_{\varepsilon^2}}}=\mathsf{S}_1\rbra{\textsc{Rank}^r_{\varepsilon^2}}$ by \cref{corollary:purification-useless}.
    On the other hand, it was shown in \cite[Theorem 3]{CHW07} and \cite[Theorem 1.11]{OW15} that $\mathsf{S}\rbra{\textsc{Rank}^r_{\varepsilon}} = \Omega\rbra{r/\varepsilon}$, and it was shown in \cite[Theorem 1.11]{OW15} that \(\mathsf{S}_1\rbra{\textsc{Rank}^r_{\varepsilon}} = \Omega\rbra{r^2/\varepsilon}\). 
    Therefore, we have $\mathsf{S}\rbra{\textsc{SchmidtRank}^r_\varepsilon} = \Omega\rbra{r/\varepsilon^2}$ and $\mathsf{S}_1\rbra{\textsc{SchmidtRank}^r_\varepsilon} = \Omega\rbra{r^2/\varepsilon^2}$.
\end{proof}

The proof of \cref{thm:schmidt-rank} is based on \cite[Theorem 1.11]{OW15}, which requires the assumption that $r \geq 2$. 
For $r = 1$, the task becomes productness testing, which will be discussed separately in \cref{sec:productness}. 
We also note that the lower bound $\mathsf{S}\rbra{\textsc{SchmidtRank}^r_\varepsilon} = \Omega\rbra{r/\varepsilon^2}$ for general testers is implied by the lower bound for testing MPS (see \cref{thm:mps-lb}), while using different techniques. 

\subsection{Bond Dimension of Matrix Product States}\label{sec-722031}
We consider the problem of testing whether an \(n\)-partite state \(\ket{\psi}\) is a matrix product state (MPS) of bond dimension \(r\) (with open boundary condition~\cite{PGVWC07}). Formally, we define
\begin{equation}
\label{eq:mat-prd-state}
\begin{split}
\textsc{MPS}^{r,n} = &\Bigg\{ \ket{\psi}_{1,\ldots,n}= \sum_{i_1,\ldots,i_n} A^{(1)}_{i_1}\cdots A^{(n)}_{i_n} \ket{i_1\cdots i_n} : \\
&\textup{\(A_i^{(j)}\) is an \(r\times r\) matrix for \(2\leq j\leq n-1\), } 
 \textup{\(A_i^{(1)},A_i^{(n)}\) are \(1\times r\) and \(r\times 1\) matrices} \Bigg\}.
\end{split}
\end{equation}
We have the following lower bound for this problem.
\begin{corollary} [Bond dimension of MPS] \label{thm-414017}
    For any $n \geq 2$ and $r\geq 2$, any tester for determining whether an $n$-partite pure state is an MPS of bond dimension $r$ or $\varepsilon$-far (in trace distance) requires sample complexity $\mathsf{S}\rbra{\textsc{MPS}^{r,n}_\varepsilon} = \Omega\rbra{(\sqrt{nr}+r)/\varepsilon^2}$. If perfect completeness is required, then we have the lower bound \(\mathsf{S}_1\rbra{\textsc{MPS}^{r,n}_\varepsilon} = \Omega\rbra{(\sqrt{nr}+r^2)/\varepsilon^2}\).
\end{corollary}
\begin{proof}
    We show $\mathsf{S}\rbra{\textsc{MPS}^{r,n}_\varepsilon} = \Omega\rbra{(\sqrt{nr}+r)/\varepsilon^2}$ in \cref{thm:mps-lb}. 
    Then, we show \(\mathsf{S}_1\rbra{\textsc{MPS}^{r,n}_\varepsilon} = \Omega\rbra{r^2/\varepsilon^2}\) in \cref{lemma-427010}. 
    The proof is completed by further noting that $\mathsf{S}_1\rbra{\textsc{MPS}^{r,n}_\varepsilon} \geq \mathsf{S}\rbra{\textsc{MPS}^{r,n}_\varepsilon}$.
\end{proof}

\subsubsection{Lower bounds for general testers}

\begin{lemma} \label{thm:mps-lb}
For any \(n\geq 2\) and \(r\geq 2\), any tester for determining whether an \(n\)-partite pure state is an MPS of bond dimension \(r\) or \(\varepsilon\)-far (in trace distance) requires sample complexity \(\mathsf{S}\rbra{\textsc{MPS}^{r,n}_\varepsilon}=\Omega((\sqrt{nr}+r)/\varepsilon^2)\).
\end{lemma}

We note that when $n = 2$, \cref{thm:mps-lb} implies the lower bound $\mathsf{S}\rbra{\textsc{SchmidtRank}^r_\varepsilon} = \Omega\rbra{r/\varepsilon^2}$ for general testers given in \cref{thm:schmidt-rank}, while using a different proof.
To prove \cref{thm:mps-lb}, we need the following lemma. 

\begin{lemma} [Overlap of tensor products, {\cite[Proposition 5.1]{soleimanifar2022testing}}] \label{lemma:gEYthm}
    Let $\ket{\psi} \in \mathbb{C}^{d_1} \otimes \mathbb{C}^{d_2} \otimes \dots \otimes \mathbb{C}^{d_k}$ be a $k$-partite state.
    Then, for $\ell \geq 1$ and $r \geq 1$, the $k\ell$-partite state $\ket{\psi}^{\otimes \ell}$ satisfies 
    \begin{equation*}
        \max_{\ket{\varphi} \in \textsc{MPS}^{r, k\ell}} \abs*{\bra{\varphi} \cdot \ket{\psi}^{\otimes \ell}}^2 = \rbra*{\max_{\ket{\varphi} \in \textsc{MPS}^{r, k}}\abs*{\braket{\varphi}{\psi}}^2}^{\ell}.
    \end{equation*}
\end{lemma}

Then, we can prove \cref{thm:mps-lb} as follows.

\begin{proof} [Proof of \cref{thm:mps-lb}]
Our proof follows the main idea given in \cite{soleimanifar2022testing}, but using a slightly modified hard example and an improved analysis (see \cref{lemma-4242204}).

Without loss of generality, we assume that $n$ is even and \(d\geq 2r\).
First, we define the bipartite states:
\begin{equation}\label{eq-4242148}
\begin{split}
\ket{\psi}_{\textup{AB}}&\!\coloneqq\!\sqrt{1-\theta}\ket{0}_\TA\ket{0}_\TB\!+\!\sum_{i=1}^{r-1} \sqrt{\frac{\theta}{r-1}}\ket{i}_\TA\ket{i}_\TB\in\HAB,\\
\ket{\phi}_{\textup{AB}}&\!\coloneqq\!\sqrt{1-\theta}\ket{0}_\TA\ket{0}_\TB\!+\!\sum_{i=1}^{d-1}\sqrt{\frac{\theta}{d-1}}\ket{i}_\TA\ket{i}_\TB\in\HAB,
\end{split}
\end{equation}
where \(\theta\coloneqq 8\varepsilon^2/n\).
Then, we define the \(n\)-partite states:
\begin{equation*}
\begin{split}
\ket{\Psi}&=\ket{\psi}_\TAB^{\otimes n/2}\, \in \,\mathcal{H}_{\textup{A}_1\textup{B}_1\ldots \textup{A}_{n/2}\textup{B}_{n/2}},\\
\ket{\Phi}&=\ket{\phi}_\TAB^{\otimes n/2}\, \in\, \mathcal{H}_{\textup{A}_1\textup{B}_1\ldots \textup{A}_{n/2}\textup{B}_{n/2}}.
\end{split}
\end{equation*}
By the Eckart-Young theorem (\cref{lemma:Eckart-Young}), we have (note that $\textsc{SchmidtRank}^r = \textsc{MPS}^{r,2}$)
\begin{equation*}
\max_{\ket{\gamma}_{\textup{AB}} \in \textsc{MPS}^{r,2}} \abs*{\braket{\gamma}{\phi}_{\textup{AB}}}^2 = 1-\theta+\frac{r-1}{d-1}\theta\leq 1-\frac{\theta}{2},
\end{equation*}
where the inequality is because \(d\geq 2r\). 
Then, we have
\begin{equation*}
\begin{split}
&\max_{\ket{\Gamma} \in \textsc{MPS}^{r,n}} \abs*{\braket{\Gamma}{\Phi}}^2 = \rbra*{1-\theta+\frac{r-1}{d-1}\theta}^{n/2} 
\leq \left(1-\frac{\theta}{2}\right)^{n/2}=\left(1-\frac{4\varepsilon^2}{n}\right)^{n/2}\leq 1-\varepsilon^2,
\end{split}
\end{equation*}
where the first equality is by \cref{lemma:gEYthm}, the last inequality is because \((1-x)^n\leq 1-xn/2\) for \(0\leq x\leq 1/n\).
Then, we can conclude that the trace distance between \(\ket{\Phi}\) and any state in \(\textsc{MPS}^{r,n}\) is no less than \(\varepsilon\), which means \(\ket{\Phi}\in\textsc{MPS}^{r,n,\geq \varepsilon}\). 
On the other hand, it is easy to see that \(\ket{\Psi}\in \textsc{MPS}^{r,n}\). 

Note that for any local unitaries \(U_i,V_i\) acting on \(\mathcal{H}_{\textup{A}_i},\mathcal{H}_{\textup{B}_i}\), respectively, and for any \(\ket{\Gamma}\in\textsc{MPS}^{r,n}\), we have \((U_i\otimes V_i)\ket{\Gamma}\in \textsc{MPS}^{r,n}\) (and the same also holds for \(\textsc{MPS}^{r,n,\geq\varepsilon}\)). Therefore, any tester for \(\textsc{MPS}^{r,n}_{\varepsilon}\) with sample complexity \(N\) is able to distinguish between the following states with success probability \(\geq 2/3\):
\begin{align}
&\rho\coloneqq\int_{U_i\in G,V_i\in\mathbb{U}_d}\!\Big((U_1\otimes V_1)\otimes\cdots\otimes (U_{n/2}\otimes V_{n/2} ) \ketbra{\Psi}{\Psi} (U^\dag_1\otimes V^\dag_1)\otimes\cdots\otimes (U^\dag_{n/2}\otimes V^\dag_{n/2}) \Big)^{\otimes N}, \nonumber \\
&\sigma\coloneqq\int_{U_i\in G,V_i\in\mathbb{U}_d}\Big((U_1\otimes V_1)\otimes\cdots\otimes (U_{n/2}\otimes V_{n/2} ) \ketbra{\Phi}{\Phi}  (U^\dag_1\otimes V^\dag_1)\otimes\cdots\otimes (U^\dag_{n/2}\otimes V^\dag_{n/2}) \Big)^{\otimes N},\nonumber 
\end{align}
where 
\begin{equation*}
G\coloneqq \set{U\in\mathbb{U}_d}{U\ket{0}=\ket{0}},
\end{equation*}
and \(U_i\) and \(V_i\) act on \(\mathcal{H}_{\textup{A}_i}\) and \(\mathcal{H}_{\textup{B}_i}\), respectively, and are integrated over the Haar measures over \(G\) and over \(\mathbb{U}_d\), respectively. This means \(d_{\tr}(\rho,\sigma)\geq 1/3\). Note that \(\rho,\sigma\) can be written as \(\rho=\rho_0^{\otimes n/2}\), \(\sigma=\sigma_0^{\otimes n/2}\), where
\begin{equation}\label{eq-4242147}
\begin{split}
\rho_0&\coloneqq \!\int_{U\in G,V\in \mathbb{U}_d} \Big((U_\TA\otimes V_\TB) \ketbra{\psi}{\psi}_\TAB(U_\TA^\dag\otimes V_\TB^\dag)\Big)^{\otimes N},\\
\sigma_0&\coloneqq \!\int_{U\in G,V\in \mathbb{U}_d} \Big((U_\TA\otimes V_\TB)\ketbra{\phi}{\phi}_\TAB (U^\dag_\TA\otimes V_\TB^\dag)\Big)^{\otimes N}.
\end{split}
\end{equation}
This implies \(d_{\tr}(\rho_0^{\otimes n/2},\sigma_0^{\otimes n/2})\geq 1/3\). 
On the other hand,
\begin{align}
1/3 &\leq d_{\tr}(\rho_0^{\otimes n/2},\sigma_0^{\otimes n/2}) \nonumber \\
&\leq \sqrt{1-\mathrm{F}(\rho_0^{\otimes n/2},\sigma_0^{\otimes n/2})^2} \label{eq-04252359}\\
&=\sqrt{1-\mathrm{F}(\rho_0,\sigma_0)^{n}} \nonumber \\
&\leq \sqrt{1-(1-d_{\tr}(\rho_0,\sigma_0))^n},\label{eq-04252351}
\end{align}
where \cref{eq-04252359,eq-04252351} is due to \cref{lemma:fi-td}. 
By \cref{lemma-4242204}, we have \(d_{\tr}(\rho_0,\sigma_0)\leq 4\sqrt{2}N\theta/r \cdot \min\{N\theta, 1\}\), which implies
\begin{align}
    \frac{8}{9} 
    & \geq \rbra*{1 - d_{\tr}(\rho_0,\sigma_0)}^n  \geq 1 - n d_{\tr}(\rho_0,\sigma_0) \nonumber \\
    & \geq 1 - \min\left\{256\sqrt{2} \cdot \frac{N^2\varepsilon^4}{nr},32\sqrt{2}\cdot\frac{N\varepsilon^2}{r}\right\},\nonumber
\end{align}
where we have used \(\theta= 8\varepsilon^2/n\). Therefore, 
\begin{equation*}
N \geq \max\left\{\frac{\sqrt{nr}}{3\cdot2^{17/4}\cdot\varepsilon^{2}},\frac{r}{3^2\cdot 2^{11/2}\cdot \varepsilon^2}\right\}\!=\!\Omega\left(\frac{\sqrt{nr}+r}{\varepsilon^2}\right)\!.
\end{equation*}
\end{proof}

In the proof of \cref{thm:mps-lb}, we require the following lemma.

\begin{lemma}\label{lemma-4242204}
Let \(\rho_0,\sigma_0\) be defined by \cref{eq-4242147}, in which \(\ket{\psi}_\TAB,\ket{\phi}_\TAB\) are defined by \cref{eq-4242148}. Then, \(d_{\tr}(\rho_0,\sigma_0)\leq 4\sqrt{2}N\theta/r \cdot \min\{N\theta, 1\}\).
\end{lemma}

To show \cref{lemma-4242204}, we need the following lemma.

\begin{lemma} [Adapted from {\cite[Equation (38)]{CHW07}}] \label{lemma-04261111}
    Let $\tau_r$ and $\tau_d$ defined by 
\begin{equation} \label{eq:def-tau-r-d}
    \tau_r\coloneqq\sum_{i=1}^{r-1}\frac{1}{r-1}\ketbra{i}{i}, \quad \tau_d\coloneqq\sum_{i=1}^{d-1}\frac{1}{d-1}\ketbra{i}{i}.
\end{equation}
    Then, for $N \geq 1$, we have
    \begin{equation*}
    \begin{split}
        &d_{\tr}\rbra*{\int_{U\in G}\rbra*{U\tau_r U^\dag}^{\otimes N},\int_{U\in G}\rbra*{U\tau_d U^\dag}^{\otimes N}}
        \leq \sqrt{2} \rbra*{\frac{N}{r-1} + \frac{N}{d-1}},
    \end{split}
    \end{equation*}
    where \(G\coloneqq \set{U\in \mathbb{U}_d}{U\ket{0}=\ket{0}}\).
\end{lemma}

Then, we can prove \cref{lemma-4242204} as follows.

\begin{proof} [Proof of \cref{lemma-4242204}]
We define a property \(\mathcal{P}=(\mathcal{P}^{\textup{yes}},\mathcal{P}^{\textup{no}})\), where
\begin{equation*}
\begin{split}
\mathcal{P}^{\textup{yes}}\coloneqq \set{(U_\TA\otimes V_\TB)\ket{\psi}_\TAB}{U\in G,V\in \mathbb{U}_d},\\
\mathcal{P}^{\textup{no}}\coloneqq \set{(U_\TA\otimes V_\TB)\ket{\phi}_\TAB}{U\in G,V\in\mathbb{U}_d},
\end{split}
\end{equation*}
in which \(G=\set{U\in\mathbb{U}_d}{U\ket{0}=\ket{0}}\).
Let \(D^{\textup{yes}}\) be a probability distribution on \(\mathcal{P}^{\textup{yes}}\) that is induced by the joint Haar measure over \(G\) and \(\mathbb{U}_d\), i.e., for any \(S\subseteq \mathcal{P}^{\textup{yes}}\),
\begin{equation*}
\Pr_{\ket{\eta}_\TAB\sim D^{\textup{yes}}}\bigl[\ket{\eta}_\TAB\!\in\! S\bigr]\!=\!\Pr_{(U,V)\sim \textup{Haar}(G\times \mathbb{U}_d)}\bigl[(U_\TA\otimes V_\TB)\ket{\psi}_\TAB \!\in\! S\bigr].
\end{equation*}
\(D^{\textup{no}}\) is defined similarly. It can be seen that both \(D^{\textup{yes}}\) and \(D^{\textup{no}}\) are unitarily invariant on \(\HB\). Therefore, by \cref{corollary-4241439} (note that \(\rho_0,\sigma_0\) can be written in the form of \cref{eq-4242359}), we can conclude that \(d_{\tr}(\rho_0,\sigma_0)=d_{\tr}(\tr_\TB[\rho_0],\tr_\TB[\sigma_0])\). 

On the other hand, we have
\begin{equation*}
\begin{split}
    &\tr_\TB[\rho_0]=\int_{U\in G} \left(U\tr_\TB\bigl[\ketbra{\psi}{\psi}_\TAB\bigr] U^\dag\right)^{\otimes N},\\
    &\tr_\TB[\sigma_0]=\int_{U\in G} \left(U\tr_\TB\bigl[\ketbra{\phi}{\phi}_\TAB\bigr] U^\dag\right)^{\otimes N}.
\end{split}
\end{equation*}
Note that for \(U\in G\), we have
\begin{align}
&U\tr_\TB\bigl[\ketbra{\psi}{\psi}_\TAB\bigr]U^\dag = (1-\theta) \ketbra{0}{0}+\sum_{i=1}^{r-1}\frac{\theta}{r-1}U\ketbra{i}{i}U^\dag
=(1-\theta)\ketbra{0}{0}+\theta U\tau_r U^\dag, \nonumber \\
&U\tr_\TB\bigl[\ketbra{\phi}{\phi}_\TAB\bigr]U^\dag= (1-\theta) \ketbra{0}{0}+\sum_{i=1}^{d-1}\frac{\theta}{d-1}U\ketbra{i}{i}U^\dag 
=(1-\theta)\ketbra{0}{0}+\theta U\tau_d U^\dag,\nonumber
\end{align}
where $\tau_r$ and $\tau_d$ are defined by \cref{eq:def-tau-r-d}. 
Therefore, 
\begin{align}
&d_{\tr}(\tr_\TB[\rho_0], \tr_\TB[\sigma_0]) \nonumber \\ 
=&d_{\tr}\!\bigg[\int_{U\in G} \rbra*{(1-\theta)\ketbra{0}{0}+\theta  U\tau_r U^\dag }^{\otimes N}, 
\int_{U\in G}\rbra*{(1-\theta)\ketbra{0}{0}+\theta  U\tau_d U^\dag }^{\otimes N}\bigg] \nonumber \\
\leq &\sum_{i=0}^{N} \binom{N}{i} (1-\theta)^i \theta^{N-i}\cdot d_{\tr}\!\Bigg[\begin{matrix}\ketbra{0}{0}^{\otimes i}\\\displaystyle\int_{U\in G}\left(U\tau_r U^\dag\right)^{\otimes N-i}\end{matrix}, 
\begin{matrix}\ketbra{0}{0}^{\otimes i}\\\displaystyle\int_{U\in G}\left(U\tau_d U^\dag\right)^{\otimes N-i}\end{matrix}\Bigg] \label{eq-04260035}\\
=&\sum_{i=0}^{N} \binom{N}{i} (1-\theta)^i \theta^{N-i}\cdot d_{\tr}\!\bigg[\int_{U\in G}\left(U\tau_r U^\dag\right)^{\otimes N-i},
\int_{U\in G}\left(U\tau_d U^\dag\right)^{\otimes N-i}\bigg], \label{eq-04260037}
\end{align}
where \cref{eq-04260035} is by the binomial expansion of tensor product and the triangle inequality of trace distance.
In the summation of \cref{eq-04260037}, the addends for $i = N$ and $i = N-1$ equal $0$ because 
\begin{equation*}
    \int_{U\in G} U \tau_r U^\dag = \tau_d=\int_{U\in G}U\tau_d U^\dag.
\end{equation*}
Therefore, 
\begin{align} 
    \eqref{eq-04260037} &= \sum_{i=0}^{N-2} \binom{N}{i} (1-\theta)^i \theta^{N-i}\cdot d_{\tr}\!\bigg[\int_{U\in G}\left(U\tau_r U^\dag\right)^{\otimes N-i},
    \int_{U\in G}\left(U\tau_d U^\dag\right)^{\otimes N-i}\bigg]\label{eq-425057}.
\end{align}
By \cref{lemma-04261111} and noting that $d \geq r\geq 2$, we have
\begin{align}
\eqref{eq-425057}
&\leq \sum_{i=0}^{N-2} \binom{N}{i} (1-\theta)^i \theta^{N-i} \sqrt{2}\rbra*{\frac{N-i}{d-1}+\frac{N-i}{r-1}}\nonumber \\
& \leq \frac{4\sqrt{2}}{r} \sum_{i=0}^{N-2} \binom{N}{i} (1-\theta)^i \theta^{N-i} \rbra{N-i} \nonumber \\
& = \frac{4\sqrt{2}N\theta}{r} \sum_{i=0}^{N-2} \binom{N-1}{i}(1-\theta)^i \theta^{N-1-i}  \nonumber \\
& = \frac{4\sqrt{2}N\theta}{r} \!\sbra*{\sum_{i=0}^{N-1} \!\binom{N-1}{i} (1-\theta)^i \theta^{N-1-i} \!-\! (1-\theta)^{N-1} \!} \nonumber \\
& = \frac{4\sqrt{2}N\theta}{r} \sbra*{1 - (1-\theta)^{N-1} } \nonumber  \\
& \leq \frac{4\sqrt{2}N\theta}{r}\min\{N\theta,1\}. \nonumber
\end{align}
\end{proof}

\subsubsection{Lower bounds for testers with perfect completeness}

\begin{lemma} \label{lemma-427010}
    For $n \geq 2$ and \(r\geq 2\), any tester with perfect completeness for determining whether an $n$-partite pure state is an MPS of bond dimension $r$ or $\varepsilon$-far (in trace distance) requires sample complexity \(\mathsf{S}_1\rbra{\textsc{MPS}^{r,n}_\varepsilon} = \Omega\rbra{r^2/\varepsilon^2}\).
\end{lemma}

\begin{proof}
MPS has an alternative characterization: a state \(\ket{\psi}_{1,\ldots, n}\) is an MPS of bond dimension \(r\) if and only if \(\tr_{k+1,\ldots,n}(\ketbra{\psi}{\psi}_{1,\ldots,n})\) has rank at most \(r\) for each \(1\leq k\leq n\) (see, e.g., \cite{PGVWC07,soleimanifar2022testing}). 
Therefore, it is easy to see that \(\textsc{MPS}^{r,2}=\textsc{SchmidtRank}^r\). 
Then, \cref{lemma-427010} is a direct implication of \cref{thm:schmidt-rank} for \(n=2\).
For \(n>2\), by \cref{lemma-416002}, a tester with perfect completeness for \(\textsc{MPS}_{\varepsilon}^{r,n}\) can be used as a tester with perfect completeness for \(\textsc{MPS}_{\varepsilon}^{r,2}\).
Therefore, \cref{lemma-427010} holds for all \(n\).
\end{proof}

\begin{lemma}\label{lemma-416002}
A tester (with perfect completeness) for \(\textsc{MPS}_{\varepsilon}^{r,n}\) can be used as a tester (with perfect completeness) for \(\textsc{MPS}_{\varepsilon}^{r,2}\).
That is, $\mathsf{S}\rbra{\textsc{MPS}_{\varepsilon}^{r,n}} \geq \mathsf{S}\rbra{\textsc{MPS}_{\varepsilon}^{r,2}}$ and $\mathsf{S}_1\rbra{\textsc{MPS}_{\varepsilon}^{r,n}} \geq \mathsf{S}_1\rbra{\textsc{MPS}_{\varepsilon}^{r,2}}$.
\end{lemma}
\begin{proof}
Suppose \(\ket{\psi}_{1,2}\) is a state on a bipartite system (labelled by \(1\) and \(2\)). We simply pad \(\ket{\psi}_{1,2}\) with \(\ket{0}_{3,\ldots,n}\) and then use the tester \(\mathcal{T}\) for \(\textsc{MPS}_{\varepsilon}^{r,n}\).
\begin{enumerate}
    \item If \(\ket{\psi}_{1,2}\in \textsc{MPS}^{r,2}\), it is obvious that \(\ket{\psi}_{1,2} \ket{0}_{3,\ldots,n}\in \textsc{MPS}^{r,n}\). Then, \(\mathcal{T}\) accepts \(\ket{\psi}_{1,2}\ket{0}_{3,\ldots,n}\) with probability \(\geq 2/3\) (or \(1\) if \(\mathcal{T}\) has perfect completeness).
    \item If \(\ket{\psi}_{1,2}\in \textsc{MPS}^{r,2,\geq\varepsilon}\), then \(\ket{\psi}_{1,2} \ket{0}_{3,\ldots,n}\in \textsc{MPS}^{r,n,\geq\varepsilon}\). To show this, for any \(\ket{\phi}_{1,\dots,n}\in\textsc{MPS}^{r,n}\), we define \(\ket{\phi'}_{1,2}= \bra{0}_{3,\dots,n} \cdot \ket{\phi}_{1,\dots,n}\), which is a (subnormalized) pure state on the first two subsystems. Note that
    \begin{equation*}
    \begin{split}
        \tr_{2}\left[\ketbra{\phi'}{\phi'}_{1,2}\right] 
        =&\bra{0}_{3,\ldots,n} \cdot \tr_2\left[\ketbra{\phi}{\phi}_{1,\dots,n}\right] \cdot \ket{0}_{3,\ldots,n} \\
        \sqsubseteq &\sum_{\ket{i}_{3,\ldots,n}}\bra{i}_{3,\ldots,n} \cdot \tr_{2}\left[\ketbra{\phi}{\phi}_{1,\dots,n} \right] \cdot \ket{i}_{3,\ldots,n}\\
        =&\tr_{2,\ldots,n}\left[\ketbra{\phi}{\phi}_{1,\dots,n}\right], 
    \end{split}
    \end{equation*}
    which implies that \(\rank(\tr_2[\ketbra{\phi'}{\phi'}_{1,2}])\leq \rank(\tr_{2,\ldots,n}[\ketbra{\phi}{\phi}_{1,\dots,n}])\leq r\). Therefore, \(\frac{\ket{\phi'}_{1,2}}{\Abs{\ket{\phi'}_{1,2}}}\in\textsc{MPS}^{r,2}\), and 
    \begin{equation*}
    \begin{split}
        &\abs*{\bra{\phi}_{1,\dots,n}\cdot \ket{\psi}_{1,2}\ket{0}_{3,\ldots,n}}=\abs*{\braket{\phi'}{\psi}_{1,2}} 
        \leq \Abs*{\ket{\phi'}_{1,2}}\sqrt{1-\varepsilon^2}\leq \sqrt{1-\varepsilon^2},
    \end{split}
    \end{equation*}
    where the first inequality is because \(\frac{\ket{\phi'}_{1,2}}{\Abs{\ket{\phi'}_{1,2}}}\in \textsc{MPS}^{r,2}\) and \(\ket{\psi}\in\textsc{MPS}^{r,2,\geq\varepsilon}\).
    This means that \(d_{\tr}(\ket{\phi}_{1,\dots,n},\ket{\psi}_{1,2}\ket{0}_{3,\ldots,n})\geq \varepsilon\), and thus \(\ket{\psi}_{1,2}\ket{0}_{3,\ldots,n}\in\textsc{MPS}^{r,n,\geq\varepsilon}\). Then, \(\mathcal{T}\) accepts \(\ket{\psi}_{1,2}\ket{0}_{3,\ldots,n}\) with probability \(\leq 1/3\).
\end{enumerate}
\end{proof}

\subsection{Maximal Entanglement}

We consider the problem of testing whether a bipartite pure state $\ket{\psi}_{\textup{AB}}$ is maximally entangled. Formally, we define
\begin{equation*}
\begin{split}
&\textsc{MaximalEntanglement} 
= \set{\ket{\psi}_{\textup{AB}} \in \HAB}{\ket{\psi}_{\textup{AB}} \textup{ is maximally entangled}}.
\end{split}
\end{equation*}
We have the following lower bound for this problem. 
\begin{corollary} \label{thm:maximal-entangle}
    Any tester for determining whether a bipartite pure state is maximally entangled requires sample complexity $\mathsf{S}\rbra{\textsc{MaximalEntanglement}_\varepsilon} = \Omega\rbra{d/\varepsilon^2}$.
\end{corollary}

\begin{proof}
    We relate the bipartite state testing of $\textsc{MaximalEntanglement}$ to the corresponding mixed state testing of $\textsc{Mixedness}$.
    Specifically, let 
    \begin{equation*}
        \textsc{Mixedness} = \cbra{I_{\textup{A}}/d} \subseteq \mathcal{D}\rbra{\HA}.
    \end{equation*}
    In the following we first show that $\mathsf{Purify}\rbra{\textsc{Mixedness}_{\varepsilon}}\subseteq \textsc{MaximalEntanglement}_\varepsilon$.
    
    \textbf{Reduction}.
    Let $\mathcal{P}=\textsc{Mixedness}$ and
    $\mathcal{Q}=\textsc{MaximalEntanglement}$.
    It suffices to prove $\mathsf{Purify}\rbra{\mathcal{P}}\subseteq\mathcal{Q}$,
    and $\mathsf{Purify}\rbra{\mathcal{P}^{\geq \varepsilon}}\subseteq \mathcal{Q}^{\geq \varepsilon}$.
    \begin{enumerate}
        \item 
            $\mathsf{Purify}\rbra*{\mathcal{P}}\subseteq \mathcal{Q}$
            can be easily seen from the definition of maximally entanglement.
            That is, for any mixed state $\rho_{\textup{A}}\in \mathcal{D}(\HA)$
            and any of its purification $\ket{\psi}_{\textup{AB}}\in \HAB$ with $\tr_{\textup{B}}\rbra*{\ketbra{\psi}{\psi}}=\rho_{\textup{A}}$,
            we have $\rho_{\textup{A}}=I_{\textup{A}}/d$ if and only if $\ket{\psi}_{\textup{AB}} =\sum_{j=0}^{d-1} \frac{1}{\sqrt{d}}\ket{\phi_j}_{\textup{A}}\ket{\gamma_j}_{\textup{B}}$
            for some orthonormal bases $\{\ket{\phi_j}_{\textup{A}}\}$ and $\{\ket{\gamma_j}_{\textup{B}}\}$.
        \item
            To prove $\mathsf{Purify}\rbra{\mathcal{P}^{\geq \varepsilon}}\subseteq \mathcal{Q}^{\geq \varepsilon}$,
            suppose mixed state $\rho_{\textup{A}}\in \mathcal{P}^{\geq \varepsilon}$.
            Then we have $d_{\tr}\rbra{\rho_{\textup{A}},I_{\textup{A}}/d}\geq \varepsilon$.
            Let $\ket{\psi}_{\textup{AB}}\in \HAB$ be a purification of $\rho_{\textup{A}}$ with $\tr_{\textup{B}}\rbra{\ketbra{\psi}{\psi}_{\textup{AB}}}=\rho_{\textup{A}}$.
            For any $\ket{\gamma}_{\textup{AB}}\in \mathcal{Q}$, its reduced density matrix is $\tr_{\textup{B}}\rbra{\ketbra{\gamma}{\gamma}_{\textup{AB}}}=I_{\textup{A}}/d\in \mathcal{P}$.
            By the contractivity of trace distance under trace-preserving quantum operations,
            we have
            \begin{equation*}
                d_{\tr}\rbra*{\ket{\psi}_{\textup{AB}},\ket{\gamma}_{\textup{AB}}}\geq d_{\tr}\rbra*{\rho_{\textup{A}},I_{\textup{A}}/d}\geq \varepsilon.
            \end{equation*}
            Hence, $\ket{\psi}_{\textup{AB}}\in \mathcal{Q}^{\geq \varepsilon}$, which implies $\mathsf{Purify}\rbra{\mathcal{P}^{\geq \varepsilon}}\subseteq \mathcal{Q}^{\geq \varepsilon}$.
    \end{enumerate}
    From the above, we conclude that $\mathsf{Purify}\rbra{\textsc{Mixedness}_{\varepsilon}} \subseteq  \textsc{MaximalEntanglement}_\varepsilon$.

    \textbf{Lower bound}.
    Because $\mathsf{Purify}\rbra{\textsc{Mixedness}_{\varepsilon}} \subseteq  \textsc{MaximalEntanglement}_\varepsilon$, and by \cref{fact:SPSQ},
    we immediately have $\mathsf{S}( \textsc{MaximalEntanglement}_\varepsilon)\geq\mathsf{S}\rbra{\mathsf{Purify}\rbra{\textsc{Mixedness}_{\varepsilon}}}$.
    We also note that $\mathsf{S}\rbra{\mathsf{Purify}\rbra{\textsc{Mixedness}_{\varepsilon}}}=\mathsf{S}\rbra{\textsc{Mixedness}_{\varepsilon}}$ by \cref{corollary:purification-useless}.
    On the other hand, it was shown in \cite[Theorem 1.10]{OW15} that $\mathsf{S}\rbra{\textsc{Mixedness}_\varepsilon} = \Omega\rbra{d/\varepsilon^2}$. 
    Therefore, we obtain the lower bound $\mathsf{S}\rbra{\textsc{MaximalEntanglement}_\varepsilon} = \Omega\rbra{d/\varepsilon^2}$.
\end{proof}

We also consider the problem of testing whether the $\alpha$-R\'enyi entanglement entropy $\mathrm{E}_\alpha\rbra{\ket{\psi}_{\textup{AB}}}$ of $\ket{\psi}_{\textup{AB}}$ is low ($\leq a$) or high ($\geq b$), which can be understood as a quantitative version of testing maximal 
entanglement. 
Testing entanglement entropy helps verify if a quantum system exhibits entanglement properties like area laws, an important topic in the study of quantum many-body systems. Some systems can be analytically proven to follow an area law (e.g.,~\cite{hastings2007area,ECP10,ALV12,anshu2022area}), while for some others entanglement entropy testers can be used.
In the following, we study the limit of entanglement entropy testers.

Formally, we define $\textsc{EntanglementEntropy}^{\alpha,a, b} = \rbra{\mathcal{Q}^{\leq a}, \mathcal{Q}^{\geq b}}$,
where
\begin{equation*}
\begin{split}
    &\mathcal{Q}^{\leq a} = \set{\ket{\psi}_{\textup{AB}} \in \HAB}{\mathrm{E}_\alpha\rbra{\ket{\psi}_{\textup{AB}}} \leq a}, \\ 
    &\mathcal{Q}^{\geq b} = \set{\ket{\psi}_{\textup{AB}} \in \HAB}{\mathrm{E}_\alpha\rbra{\ket{\psi}_{\textup{AB}}} \geq b}.
\end{split}
\end{equation*}
Here, $\mathrm{E}_\alpha\rbra{\cdot}$ is the $\alpha$-R\'enyi entanglement entropy induced by the $\alpha$-R\'enyi entropy of mixed states 
\begin{equation*}
    \mathrm{S}_\alpha\rbra{\rho} = \frac{1}{1-\alpha} \ln \rbra{\tr\sbra{\rho^{\alpha}}}.
\end{equation*}
We have the following lower bound for this problem. 
\begin{corollary} [Entanglement entropy] \label{thm:entanglement-entropy}
    For any constant $\alpha > 0$, any tester for determining whether the $\alpha$-R\'enyi entanglement entropy of a bipartite pure state is low ($\leq a$) or high ($\geq b$) requires sample complexity
    $\mathsf{S}\rbra{\textsc{EntanglementEntropy}^{\alpha,a, b}} = \Omega\rbra{d/\Delta+d^{1/\alpha-1}/\Delta^{1/\alpha}}$, where $\Delta = b - a$. 
\end{corollary}

\begin{proof}
    We relate the bipartite state testing of $\textsc{EntanglementEntropy}^{\alpha,a, b}$
    to the corresponding mixed state testing of $\textsc{Entropy}^{\alpha,a, b}$.
    Specifically, let 
    \begin{equation*}
        \textsc{Entropy}^{\alpha, a,b} = \rbra{\mathcal{P}^{\leq a}, \mathcal{P}^{\geq b}},
    \end{equation*}
    where 
\begin{equation*}
\begin{split}
    &\mathcal{P}^{\leq a} = \{ \rho_{\textup{A}} \in \mathcal{D}\rbra{\HA} : \mathrm{S}_\alpha\rbra{\rho_{\textup{A}}} \leq a \}, \\ 
    &\mathcal{P}^{\geq b} = \{ \rho_{\textup{A}} \in \mathcal{D}\rbra{\HA} : \mathrm{S}_\alpha\rbra{\rho_{\textup{A}}} \geq b \}.
\end{split}
\end{equation*}
    In the following we first show that $\mathsf{Purify}\rbra{\textsc{Entropy}^{\alpha, a, b}}= \textsc{EntanglementEntropy}^{\alpha,a, b}$.

    \textbf{Reduction}.
    It suffices to prove $\mathsf{Purify}\rbra{\mathcal{P}^{\leq a}}=\mathcal{Q}^{\leq a}$,
    and $\mathsf{Purify}\rbra{\mathcal{P}^{\geq b}}=\mathcal{Q}^{\geq b}$.
    These can be easily seen from the connection between the $\alpha$-R\'enyi entanglement entropy of bipartite pure states
    and the $\alpha$-R\'enyi entropy of the corresponding reduced mixed states.
    Take the low entropy case for example, 
    for any mixed state $\rho_{\textup{A}}\in \mathcal{D}(\HA)$
    and any of its purification $\ket{\psi}_{\textup{AB}}\in \HAB$ with $\tr_{\textup{B}}\rbra{\ketbra{\psi}{\psi}_{\textup{AB}}}=\rho_{\textup{A}}$,
    we have $\mathrm{S}_\alpha\rbra{\rho}\leq a$ if and only if $\mathrm{E}_\alpha\rbra{\ket{\psi}} \leq a$.
    The same statement holds for the high entropy case.
    
    \textbf{Lower bound}
    Because $\mathsf{Purify}\rbra{\textsc{Entropy}^{\alpha, a, b}}= \textsc{EntanglementEntropy}^{\alpha,a, b}$,
    we immediately have 
    $
        \mathsf{S}(  \textsc{EntanglementEntropy}^{\alpha,a, b}) = \mathsf{S}\rbra*{\mathsf{Purify}\rbra{\textsc{Entropy}^{\alpha, a, b}}}.
    $
    We also note that 
    \[\mathsf{S}\rbra{\mathsf{Purify}\rbra{\textsc{Entropy}^{\alpha, a, b}}}=\mathsf{S}\rbra{\textsc{Entropy}^{\alpha, a, b}}\] by \cref{corollary:purification-useless}.
    On the other hand, it was shown in \cite[Theorem VI.1]{WZ24} that $\mathsf{S}\rbra{\textsc{Entropy}^{\alpha, a, b}} = \Omega\rbra{d/\Delta+d^{1/\alpha-1}/\Delta^{1/\alpha}}$, where $\Delta = b - a$.
    Therefore, we have 
    $
        \mathsf{S}\rbra{\textsc{EntanglementEntropy}^{\alpha,a, b}} = \Omega\rbra{d/\Delta+d^{1/\alpha-1}/\Delta^{1/\alpha}}.
    $
\end{proof}

\subsection{Uniform Schmidt Coefficients}

We consider the problem of testing whether the entanglement spectrum of $\ket{\psi}_{\textup{AB}}$ is uniform on $r$ or $r+\Delta$ non-zero Schmidt coefficients. Formally, we define
$\textsc{UniformSchmidtCoefficients}^{r, \Delta} = \rbra{\mathcal{Q}^{r}, \mathcal{Q}^{r+\Delta}}$, where 
\begin{equation*}
\begin{split}
    \mathcal{Q}^{s} = \big\{\ket{\psi}_{\textup{AB}} \in \HAB: \ket{\psi}_{\textup{AB}} \textup{ has exactly } s \textup{ non-zero} 
    \textup{ Schmidt coefficients } 1/\sqrt{s}\big\}.
\end{split}
\end{equation*}
We have the following lower bound for this problem. 

\begin{corollary} [Uniform Schmidt coefficients] \label{thm:uniform-schmidt}
    Any tester for determining whether the entanglement spectrum of a bipartite pure state is uniform on $r$ or $r + \Delta$ Schmidt coefficients requires sample complexity $\mathsf{S}\rbra{\textsc{UniformSchmidtCoefficients}^{r, \Delta}} = \Omega^*\rbra{r^2/\Delta}$.
\end{corollary}

\begin{proof}
    We reduce the bipartite state testing of $\textsc{UniformSchmidtCoefficients}^{r, \Delta}$
    to the corresponding mixed state testing of $\textsc{UniformEigenvalues}^{r, \Delta}$,
    i.e., testing whether the spectrum of a mixed state is uniform on $r$ or $r+\Delta$ eigenvalues.
    Specifically, let 
    \begin{equation*}
        \textsc{UniformEigenvalues}^{r, \Delta} = \rbra{\mathcal{P}^r, \mathcal{P}^{r+\Delta}},
    \end{equation*}
    where 
    \begin{equation*}
        \mathcal{P}^s = \set{\rho_{\textup{A}} \in \mathcal{D}\rbra{\HA}}{\rho_{\textup{A}}^2 = \rho_{\textup{A}}/s}.
    \end{equation*}
    We first show that $\mathsf{Purify}\rbra{\textsc{UniformEigenvalues}^{r, \Delta}} = \textsc{UniformSchmidtCoefficients}^{r, \Delta}$.

    \textbf{Reduction}.
    We note that
    \begin{equation*}
    \begin{split}
        \mathcal{Q}^s &= \Bigg\{\ket{\psi}_{\textup{AB}} \in \HAB:\ket{\psi}_{\textup{AB}}=\sum_{j=0}^{s-1} \frac{1}{\sqrt{s}}\ket{\phi_j}_{\textup{A}}\ket{\gamma_j}_{\textup{B}} \text{ for orthonormal bases $\{\ket{\phi_j}_{\textup{A}}\}$ and $\{\ket{\gamma_j}_{\textup{B}}\}$}\Bigg\}.
    \end{split}
    \end{equation*}
    It suffices to prove $\mathsf{Purify}\rbra*{\mathcal{P}^{s}}=\mathcal{Q}^{s}$.
    This can be easily seen from the connection between the Schmidt coefficients of bipartite pure states
    and the eigenvalues of the corresponding reduced mixed states.
    Specifically,
    for any mixed state $\rho_{\textup{A}}\in \mathcal{D}(\HA)$
    and any of its purification $\ket{\psi}_{\textup{AB}}\in \HAB$ with $\tr_{\textup{B}}\rbra{\ketbra{\psi}{\psi}_{\textup{AB}}}=\rho_{\textup{A}}$,
    we have $\rho^2_{\textup{A}}=\rho_{\textup{A}}/s$ if and only if $\ket{\psi}_{\textup{AB}}=\sum_{j=0}^{s-1} \frac{1}{\sqrt{s}}\ket{\phi_j}_{\textup{A}}\ket{\gamma_j}_{\textup{B}}$
    for some orthonormal bases $\{\ket{\phi_j}_{\textup{A}}\}$ and $\{\ket{\gamma_j}_{\textup{B}}\}$.
    
    \textbf{Lower bound}. $\mathsf{Purify}\rbra{\textsc{UniformEigenvalues}^{r, \Delta}} = \textsc{UniformSchmidtCoefficients}^{r, \Delta}$ implies 
    \begin{equation*}
    \begin{split}
        &\mathsf{S}(\textsc{UniformSchmidtCoefficients}^{r, \Delta})
     = \, \mathsf{S}\rbra{\mathsf{Purify}\rbra{\textsc{UniformEigenvalues}^{r, \Delta}}}.
    \end{split}
    \end{equation*}
    We also note that $\mathsf{S}\rbra{\mathsf{Purify}\rbra{\textsc{UniformEigenvalues}^{r, \Delta}}}=\mathsf{S}\rbra{\textsc{UniformEigenvalues}^{r, \Delta}}$ by \cref{corollary:purification-useless}.
    On the other hand, it was shown in \cite[Theorem 1.12]{OW15} that $\mathsf{S}\rbra{\textsc{UniformEigenvalues}^{r, \Delta}} = \Omega^*\rbra{r^2/\Delta}$. 
    This yields 
    $
        \mathsf{S}\rbra{\textsc{UniformSchmidtCoefficients}^{r, \Delta}} = \Omega^*\rbra{r^2/\Delta}.
    $
\end{proof}

\subsection{Productness} \label{sec:productness}

We consider the problem of testing whether a bipartite pure state $\ket{\psi}_{\textup{AB}}$ is a product state. Formally, we define
\begin{equation*}
\begin{split}
&\textsc{Productness} = \Big\{\ket{\psi}_{\textup{AB}} \in \HAB:  \ket{\psi}_{\textup{AB}} = \ket{\phi}_\textup{A} \otimes \ket{\gamma}_\textup{B} \textup{ is a product state}\Big\}.
\end{split}
\end{equation*}
Note that a product state is actually an MPS of bond dimension $r = 1$. 
We have the following lower bound for this problem. 
\begin{corollary} [Productness] \label{thm:product}
    Any tester for determining whether a bipartite pure state is a product state or $\varepsilon$-far (in trace distance) requires sample complexity
    $\mathsf{S}\rbra{\textsc{Productness}_\varepsilon} = \Omega\rbra{1/\varepsilon^2}$.
\end{corollary}

The lower bound for testing productness was mentioned in \cite{soleimanifar2022testing}. Here, we give a proof different from theirs. 

\begin{proof}
    We relate the bipartite state testing of $\textsc{Productness}$ to the corresponding mixed state testing of $\textsc{Purity}$.
    Specifically, let
    \begin{equation*}
        \textsc{Purity} = \set{\rho_{\textup{A}} \in \mathcal{D}\rbra{\HA}}{\tr\sbra{\rho_{\textup{A}}^2} = 1}.
    \end{equation*}
    In the following we first show that $\mathsf{Purify}\rbra{\textsc{Purity}_{\varepsilon^2}}\subseteq \textsc{Productness}_\varepsilon$.
    
    \textbf{Reduction}.
    It is easy to see that \(\textsc{Purity}=\textsc{Rank}^1\) and \(\textsc{Productness}=\textsc{SchmidtRank}^1\). Then, since $\mathsf{Purify}\rbra{\textsc{Rank}^r_{\varepsilon^2}} \subseteq \textsc{SchmidtRank}^r_\varepsilon$ (as shown in \cref{thm:schmidt-rank}), we can conclude that $\mathsf{Purify}\rbra{\textsc{Purity}_{\varepsilon^2}}\subseteq \textsc{Productness}_\varepsilon$.

    \textbf{Lower bound}.
    Because $\mathsf{Purify}\rbra{\textsc{Purity}_{\varepsilon^2}}\subseteq \textsc{Productness}_\varepsilon$, and by \cref{fact:SPSQ},
    we immediately have $\mathsf{S}(\textsc{Productness}_\varepsilon)\geq\mathsf{S}\rbra{\mathsf{Purify}\rbra{\textsc{Purity}_{\varepsilon^2}}}$.
    We also note that $\mathsf{S}\rbra{\mathsf{Purify}\rbra{\textsc{Purity}_{\varepsilon^2}}}=\mathsf{S}\rbra{\textsc{Purity}_{\varepsilon^2}}$ by \cref{corollary:purification-useless}.
    On the other hand, it is shown in \cref{thm:purity} that $\mathsf{S}\rbra{\textsc{Purity}_{\varepsilon}} = \Omega\rbra{1/\varepsilon}$.
    Therefore, we have $\mathsf{S}\rbra{\textsc{Productness}_\varepsilon} = \Omega\rbra{1/\varepsilon^2}$.
\end{proof}

\section{Quantum Query Lower Bounds}
\label{sec:query-bounds}

Our results for entanglement spectrum testing of bipartite pure states contribute to not only quantum sample complexity but also quantum query complexity. 
In particular, we are able to prove a new lower bound for the quantum query complexity of the entanglement entropy problem studied in \cite{SY23}. 
The problem is defined as follows.

\begin{problem} [Entanglement entropy problem, generalizing {\cite[Definition 1.13]{SY23}}] 
    Let $0 < a < b \leq \ln\rbra{d}$ and $\alpha > 0$. 
    Then entanglement entropy problem is to determine whether the $\alpha$-R\'{e}nyi entanglement entropy of a bipartite state $\ket{\psi}_{\textup{AB}} \in \HAB = \HA \otimes \HB = \mathbb{C}^d \otimes \mathbb{C}^d$ is low ($\leq a$) or high ($\geq b$), given query access to the unitary reflection operator $I_{\textup{AB}} - 2 \ketbra{\psi}{\psi}_{\textup{AB}}$. 
    Formally, the promise problem is defined by
    $\textsc{QEntanglementEntropy}^{\alpha,a,b} = \rbra{\mathcal{P}^{\textup{low}}, \mathcal{P}^{\textup{high}}}$, where 
 \[\mathcal{P}^{\textup{low}} \!=\! \set{\!I_{\textup{AB}} - 2 \ketbra{\psi}{\psi}_{\textup{AB}}}{\ket{\psi}_{\textup{AB}} \in \HAB, \, \mathrm{E}_\alpha\rbra{\ket{\psi}_{\textup{AB}}} \leq a },\] 
 \[\mathcal{P}^{\textup{high}} \!=\! \set{\!I_{\textup{AB}} - 2 \ketbra{\psi}{\psi}_{\textup{AB}}}{\ket{\psi}_{\textup{AB}} \in \HAB, \, \mathrm{E}_\alpha\rbra{\ket{\psi}_{\textup{AB}}} \geq b }.\]
    Here, $\mathrm{E}_\alpha\rbra{\cdot}$ is the $\alpha$-R\'enyi entanglement entropy induced by the $\alpha$-R\'enyi entropy of mixed states.
\end{problem}

We give a quantum query lower bound for the entanglement entropy problem as follows.

\begin{corollary} [Quantum query lower bounds for the entanglement entropy problem] \label{thm:q-entropy}
    For any constant $\alpha > 0$, any quantum query algorithm for the entanglement entropy problem requires quantum query complexity 
    $
        \mathsf{Q}\rbra{\textsc{QEntanglementEntropy}^{\alpha,a,b}} = \widetilde \Omega\rbra{\sqrt{d/\Delta}+\sqrt{d^{1/\alpha-1}/\Delta^{1/\alpha}}},
    $
    where $\Delta = b - a$.
\end{corollary}

The proof of \cref{thm:q-entropy} is based on the sample lower bound for $\textsc{EntanglementEntropy}^{\alpha,a,b}$ in \cref{corollary:s-maximum-entropy} via the quantum sample-to-query lifting theorem \cite{WZ23}. 

\subsection{Preparation}

\paragraph{Block-encoding}

To give the proof, we need the notion of block-encoding~\cite{CGJ19,GSLW19}. 

\begin{definition} [Block-encoding]
    Suppose that $n \geq 1$ and $a \geq 0$ are integers, and $\alpha > 0$ and $\varepsilon \geq 0$ are real numbers. 
    Let $A$ be an $n$-qubit operator. 
    A unitary operator $U$ on Hilbert space $\mathbb{C}^{2^n} \otimes \mathbb{C}^{2^a}$ is said to be an $\rbra{\alpha, a, \varepsilon}$-block-encoding of $A$, if 
    \begin{equation*}
        \Abs*{\alpha \rbra*{I_n \otimes \bra{0}^{\otimes a}} U \rbra*{I_n \otimes \ket{0}^{\otimes a}} - A} \leq \varepsilon,
    \end{equation*}
    where $\Abs{\cdot}$ is the operator norm and $I_n$ is the identity operator on $n$ qubits. 
\end{definition}

We need to scale down the factor in the block-encoded operators, and we use the version from \cite{WZ23}.

\begin{lemma} [Down-scaling of block-encoded operators, {\cite[Corollary 2.6]{WZ23}}] \label{lemma:down-scaling}
    Let $n \geq 1$ and $a \geq 0$. 
    Suppose that $U$ is an $\rbra{n+a}$-qubit unitary operator that is a $\rbra{1, a, \varepsilon}$-block-encoding of $A$.
    Then, for $\alpha > 1$, there is an $\rbra{n+a+1}$-qubit unitary operator $U_{\alpha}$ that is a $\rbra{1, a+1, \varepsilon/\alpha}$-block-encoding of $A/\alpha$, using $1$ query to $U$. 
\end{lemma}

\paragraph{LCU lemma}

We need the Linear Combination of Unitaries (LCU) lemma due to \cite{CW12,BCC+15}. 
We use the version from \cite{GSLW19}.

\begin{definition} [State-preparation-pair]
    Let $\beta > 0$ and $\varepsilon \geq 0$ be real numbers and $m \geq 1$ and $b \geq 0$ be integers. 
    Let $y \in \mathbb{C}^m$ with $\Abs{y}_1 \leq \beta$. 
    A pair of unitary operators $\rbra{P_L, P_R}$ is said to be a $\rbra{\beta, b, \varepsilon}$-state-preparation-pair for $y$, if $P_L \ket{0}^{\otimes b} = \sum_{j = 0}^{2^b-1} c_j \ket{j}$ and $P_R \ket{0}^{\otimes b} = \sum_{j = 0}^{2^b-1} d_j \ket{j}$ such that $\sum_{j=0}^{m-1} \abs{\beta c_j^* d_j - y_j} \leq \varepsilon$ and $c_j^* d_j = 0$ for all $m \leq j < 2^b$.
\end{definition}

\begin{lemma} [Linear combination of block-encoded operators, {\cite[Lemma 29]{GSLW19}}]
\label{lemma:lcu}
    Let $\alpha, \beta > 0$ and $ \varepsilon_1, \varepsilon_2 \geq 0$ be real numbers, and $s \geq 1$ and $a, b \geq 0$ be integers. 
    Suppose that $A = \sum_{j=0}^{m-1} y_jA_j$ is an $s$-qubit operator with $y \in \mathbb{C}^m$, $\rbra{P_L, P_R}$ is a $\rbra{\beta, b, \varepsilon_1}$-state-preparation-pair for $y$, and $U_j$ is an $\rbra{\alpha, a, \varepsilon_2}$-block-encoding of $A_j$ for $0\leq j\leq m-1$. 
    Then, we can implement a quantum circuit that is an $\rbra{\alpha\beta, a+b, \alpha\varepsilon_1 + \alpha\beta\varepsilon_2}$-block-encoding of $A$, using $1$ query to each of $P_L$, $P_R$ and $U_j$ for $0\leq j\leq m-1$.
\end{lemma}

\paragraph{Quantum amplitude amplification}

Quantum amplitude amplification \cite{BHMT02} is a generalization of Grover search \cite{Gro96}. 
It can also be described in the language of block-encoding (cf.\ {\cite[Theorem 15]{GSLW19}}).

\begin{lemma}[Quantum amplitude amplification, adapted from {\cite[Theorem 15]{GSLW19}}] \label{lemma:ampl-amplifi}
Let $\alpha > 1$ and $\varepsilon > 0$ be real numbers, and $n \geq 1$ and $a \geq 0$ be integers.
Suppose that $W$ is an $n$-qubit unitary operator, and $U$ is a unitary operator that is a $(1,a,\varepsilon)$-block-encoding of $W/\alpha$.
Then, there is a quantum circuit $\widetilde{U}$ that is a $(1,a+2, 8\alpha \varepsilon)$-block-encoding of $W$,
by using $O(\alpha)$ queries to $U$.
\end{lemma}

\paragraph{Quantum sample-to-query lifting}

To show quantum query lower bounds, we need the quantum sample-to-query lifting theorem from \cite{WZ23}. 
First, we define the ``diamond'' quantum query complexity for quantum state testing. 

\begin{definition} [Diamond quantum query complexity] 
    Let $\mathcal{P} = \rbra{\mathcal{P}^{\textup{yes}}, \mathcal{P}^{\textup{no}}}$ be a quantum state testing problem with $\mathcal{P}^{\textup{yes}} \cup \mathcal{P}^{\textup{no}} \subseteq \mathcal{D}\rbra{\mathbb{C}^d}$. 
    A diamond quantum tester $\mathcal{T}$ for $\mathcal{P}$ is a quantum query algorithm given access to a unitary oracle $U$ (as well as controlled-$U$ and their inverses) such that for every $\rho \in \mathcal{P}^{\textup{yes}} \cup \mathcal{P}^{\textup{no}}$, and every unitary $\rbra{2, a, 0}$-block-encoding $U_\rho$ of $\rho$ for some $a \geq 0$, it holds that
    \begin{itemize}
        \item $\Pr\sbra{\mathcal{T}^{U_\rho} \textup{ accepts}} \geq 2/3$ if $\rho \in \mathcal{P}^{\textup{yes}}$. 
        \item $\Pr\sbra{\mathcal{T}^{U_\rho} \textup{ accepts}} \leq 1/3$ if $\rho \in \mathcal{P}^{\textup{no}}$. 
    \end{itemize}
    
    The \emph{diamond} quantum query complexity of $\mathcal{P}$, denoted by $\mathsf{Q}_\diamond\rbra{\mathcal{P}}$, is the minimum quantum query complexity over all quantum diamond testers for $\mathcal{P}$.
\end{definition}

\begin{theorem} [Quantum sample-to-query lifting, {\cite[Theorem 1.1]{WZ23}}] \label{thm:qlifting}
    For every quantum state testing problem $\mathcal{P} = \rbra{\mathcal{P}^{\textup{yes}}, \mathcal{P}^{\textup{no}}}$, it holds that $\mathsf{Q}_\diamond\rbra{\mathcal{P}} = \widetilde \Omega\rbra{\sqrt{\mathsf{S}\rbra{\mathcal{P}}}}$.
\end{theorem}

\subsection{Proof of Corollary~\ref{thm:q-entropy}}

The proof of \cref{thm:q-entropy} is to relate the sample complexity $\mathsf{S}\rbra{\textsc{EntanglementEntropy}^{\alpha,a,b}}$ to the query complexity $\mathsf{Q}\rbra{\textsc{QEntanglementEntropy}^{\alpha,a,b}}$, consisting of three steps:
\begin{enumerate}
    \item By \cref{corollary:s-maximum-entropy}, $\mathsf{S}\rbra{\textsc{EntanglementEntropy}^{\alpha,a, b}} = \Omega\rbra{d/\Delta+d^{1/\alpha-1}/\Delta^{1/\alpha}}$. 
    \item By \cref{thm:qlifting}, 
    $
        \mathsf{Q}_\diamond \rbra{\textsc{EntanglementEntropy}^{\alpha,a, b}}
         = \widetilde \Omega\rbra{\sqrt{\mathsf{S}\rbra{\textsc{EntanglementEntropy}^{\alpha,a, b}}}} 
         = \widetilde \Omega\rbra{\sqrt{d/\Delta}+\sqrt{d^{1/\alpha-1}/\Delta^{1/\alpha}}}.
    $
    \item Finally, note that
    $
    \mathsf{Q}\rbra{\textsc{QEntanglementEntropy}^{\alpha, a, b}} = \Theta\rbra{\mathsf{Q}_\diamond\rbra{\textsc{EntanglementEntropy}^{\alpha,a, b}}}.
    $
    This can be seen by noting that the two types of quantum oracles can be implemented by each other with a constant slowdown: $I_{\textup{AB}} - 2 \ketbra{\psi}{\psi}_{\textup{AB}}$ and a unitary operator that is a block-encoding of $\ketbra{\psi}{\psi}_{\textup{AB}}$ (see \cref{lemma:equiv-oracles}). 
\end{enumerate}

To complete the proof, we show the equivalence of the two types of quantum oracles. 

\begin{lemma} \label{lemma:equiv-oracles}
    Let $\ket{\psi}$ be an $n$-qubit pure state, and $R = I - 2\ketbra{\psi}{\psi}$ be an $n$-qubit unitary operator. 
    Then, 
    \begin{itemize}
        \item Using $O\rbra{1}$ queries to $R$, we can implement a unitary operator that is a $\rbra{2, a, 0}$-block-encoding of $\ketbra{\psi}{\psi}$ for some $a \geq 1$. 
        \item For any unitary operator $U$ that is a $\rbra{2, a, 0}$-block-encoding of $\ketbra{\psi}{\psi}$ for some $a \geq 1$, we can implement $R$ using $O\rbra{1}$ queries to $U$. 
    \end{itemize}
\end{lemma}

To show this, we need the following lemma. 

\begin{lemma} \label{lemma:block-encoding-I-2A}
    Let $n \geq 1$ and $a \geq 0$ be integers. 
    Suppose that $A = \ketbra{\psi}{\psi}$ is an $n$-qubit operator with $\ket{\psi}$ an $n$-qubit pure quantum state, and $U$ is a $\rbra{2, a, 0}$-block-encoding of $A$. 
    Then, we can implement a unitary operator that is a $\rbra{1, a+3, 0}$-block-encoding of $I - 2A$.
\end{lemma}
\begin{proof}
    Let $P_L$ and $P_R$ be $1$-qubit unitary operators such that
    \begin{equation*}
        P_L \ket{0} = \frac{1}{\sqrt{5}} \ket{0} + \frac{2}{\sqrt{5}} \ket{1}, \quad P_R \ket{0} = \frac{1}{\sqrt{5}} \ket{0} - \frac{2}{\sqrt{5}} \ket{1}.
    \end{equation*}
    It can be verified that $\rbra{P_L, P_R}$ is a $\rbra{5, 1, 0}$-state-preparation-pair for $\rbra{1, -4}$. 
    Then, with $U_1 \coloneqq I_{n+a}$ a $\rbra{1, a, 0}$-block-encoding of $I_n$ and $U_2 \coloneqq U$ a $\rbra{1, a, 0}$-block-encoding of $A/2$, by \cref{lemma:lcu}, we can implement a quantum circuit $W$ that is a $\rbra{5, a+1, 0}$-block-encoding of $I - 2A$, by using $1$ query to each of $P_L$, $P_R$, $U_1$ and $U_2$.
    By \cref{lemma:ampl-amplifi} with $\alpha \coloneqq 5$ and noting that $I - 2A$ is unitary, we can implement a quantum circuit $V$ that is a $\rbra{1, a+3, 0}$-block-encoding of $I-2A$, by using $O\rbra{1}$ queries to $W$.
    In summary, $V$ is the desired unitary operator and it uses $O\rbra{1}$ queries to $U$. 
\end{proof}

Now we are ready to prove \cref{lemma:equiv-oracles}.

\begin{proof} [Proof of \cref{lemma:equiv-oracles}]
    The proof is split into two parts: the implementation of the block-encoding of $\ketbra{\psi}{\psi}$ by $R$, and the implementation of $R$ by the block-encoding of $\ketbra{\psi}{\psi}$. 

    \textbf{The implementation of the block-encoding of $\ketbra{\psi}{\psi}$ by $R$}. Let $P_L$ and $P_R$ be $1$-qubit unitary operators such that
    \begin{equation*}
        P_L \ket{0} = \frac{1}{\sqrt{2}} \ket{0} + \frac{1}{\sqrt{2}} \ket{1}, \quad P_R \ket{0} = \frac{1}{\sqrt{2}} \ket{0} - \frac{1}{\sqrt{2}} \ket{1}.
    \end{equation*}
    It can be verified that $\rbra{P_L, P_R}$ is a $\rbra{2, 1, 0}$-state-preparation-pair for $\rbra{1, -1}$. 
    Then, with $U_1 \coloneqq I_{n}$ (which is a $\rbra{1, 0, 0}$-block-encoding of $I_n$) and $U_2 \coloneqq R$ (which is a $\rbra{1, 0, 0}$-block-encoding of $R$), by \cref{lemma:lcu}, we can implement a quantum circuit $W$ that is a $\rbra{2, 1, 0}$-block-encoding of $I_n - R = 2\ketbra{\psi}{\psi}$, by using $1$ query to each of $P_L$, $P_R$, $U_1$ and $U_2$.
    We note that $W$ is actually a $\rbra{1, 1, 0}$-block-encoding of $\ketbra{\psi}{\psi}$ and can be implemented by using $1$ query to $R$. 
    To complete the proof, we note that for $a > 1$, $W \otimes I_{a-1}$ is a $\rbra{1, a, 0}$-block-encoding of $\ketbra{\psi}{\psi}$. 
    Finally, by \cref{lemma:down-scaling}, we can implement a unitary operator that is a $\rbra{1, a+1, 0}$-block-encoding of $\ketbra{\psi}{\psi}/2$, i.e., a $\rbra{2, a+1, 0}$-block-encoding of $\ketbra{\psi}{\psi}$, by using $1$ query to $W \otimes I_{a-1}$.

    \textbf{The implementation of $R$ by the block-encoding of $\ketbra{\psi}{\psi}$}. By \cref{lemma:block-encoding-I-2A}, we can implement a unitary operator $W$ that is a $\rbra{1, a+3, 0}$-block-encoding of $I - 2\ketbra{\psi}{\psi} = R$, by using $O\rbra{1}$ queries to $U$. 
    Note that $W = R \otimes V$, where $V$ acts on $\rbra{a+3}$ ancilla qubits. 
\end{proof}

\section{Optimal One-Way LOCC for Bipartite Mixed States}
\label{sec:optimal-locc}

In this section, we will show how to construct an optimal one-way LOCC tester for unitarily invariant properties of bipartite mixed states when $\mathcal{P}^{\textup{yes}}$ only consists of pure states (of the form $\ket{\psi}\!\bra{\psi}$), and thus prove \cref{thm:ext}.
We formally state it as follows. 

\begin{theorem}\label{lemma-3311612}
Let $\mathcal{P} = \rbra{\mathcal{P}^{\textup{yes}}, \mathcal{P}^{\textup{no}}}$ be a property of bipartite mixed states in $\mathcal{D}\rbra{\HAB}$ that is unitarily invariant on $\HB$. 
If $\mathcal{P}^{\textup{yes}}$ only consists of pure states, 
then, for any parameters \(0< s < c \leq 1\), there exists a one-way LOCC 
\((c,s)\)-tester for $\mathcal{P}$ that achieves the sample complexity 
$k\cdot \mathsf{S}_{c,s}\rbra{\mathcal{P}}$, where \(k\) is a constant that depends only on \(c\) and \(s\).
\end{theorem}

\subsection{Construction (Proof of Theorem~\ref{lemma-3311612})}
To prove \cref{lemma-3311612}, let $\mathcal{T}$ be a (possibly global) \((c,s)\)-tester for $\mathcal{P}$ with sample complexity $\mathsf{S}_{c,s}\rbra{\mathcal{P}}$, where \(0<s< c\leq 1\), then we use the following theorem to show the existence of an optimal one-way LOCC tester that achieves sample complexity $200 \ln(1/s+10) (c-s)^{-4} \cdot \mathsf{S}_{c,s}\rbra{\mathcal{P}}$ (note that this overhead \(200 \ln(1/s+10) (c-s)^{-4}\) is a constant in the case of \(c=2/3\) and \(s=1/3\), which is the default setting of a quantum tester).

\begin{theorem}
Suppose $\mathcal{P} = \rbra{\mathcal{P}^{\textup{yes}}, \mathcal{P}^{\textup{no}}}$ is a property of bipartite mixed states in $\mathcal{D}\rbra{\HAB}$ such that $\mathcal{P}$ is unitarily invariant on $\HB$ and $\mathcal{P}^{\textup{yes}}$ only consists of pure states.
For any \(0<s<c\leq 1\), if \(\mathcal{T}\) is a \((c,s)\)-tester for \(\mathcal{P}\) with sample complexity \(N_0\), then there is a one-way LOCC (from Bob to Alice) \((c,s)\)-tester with sample complexity $200 \ln(1/s+10) (c-s)^{-4} \cdot N_0$.
\end{theorem}
\begin{proof}
We first check whether $N_0 \geq 2(c-s)^{-1}$. If this is not the case, then we simply pad the tester \(\mathcal{T}\) with trivial measurements, and obtain a tester (also denoted by \(\mathcal{T}\) in the rest of this paper) that has sample complexity \(2(c-s)^{-1}\). Therefore, from now on, suppose we have a \((c,s)\)-tester \(\mathcal{T}\) with sample complexity 
\begin{equation}\label{eq-423150}
    N=\max\rbra{N_0,2(c-s)^{-1}}.
\end{equation}

Note that the tester \(\mathcal{T}\) has the following form:
\begin{equation*}
\mathcal{T}=\sum_{\lambda_1,\lambda_2,\lambda_3,\lambda_4\vdash (N,d)} \mathcal{T}_{\lambda_1,\lambda_2\rightarrow\lambda_3,\lambda_4},
\end{equation*}
where \(\mathcal{T}_{\lambda_1,\lambda_2\rightarrow\lambda_3,\lambda_4}:(\mathcal{V}_{\lambda_1,\TA}\otimes\mathcal{W}_{\lambda_1,\TA})\otimes (\mathcal{V}_{\lambda_2,\TB}\otimes \mathcal{W}_{\lambda_2,\TB})\rightarrow (\mathcal{V}_{\lambda_3,\TA}\otimes \mathcal{W}_{\lambda_3,\TA})\otimes(\mathcal{V}_{\lambda_4,\TB}\otimes \mathcal{W}_{\lambda_4,\TB})\) is a linear map, in which the subspaces notated with \(\lambda_1,\lambda_3\) are possessed by Alice and those with \(\lambda_2,\lambda_4\) are possessed by Bob. 
Here, $\mathcal{V}_{\lambda, \textup{A}} = \mathcal{V}_{\lambda, \textup{B}} \coloneqq \mathcal{V}_\lambda$ and \(\mathcal{W}_{\lambda,\textup{A}}=\mathcal{W}_{\lambda,\textup{B}}\coloneqq\mathcal{W}_{\lambda}\) are irreducible representations of \(\mathbb{S}_N\) and \(\mathbb{U}_d\) that appear in the decomposition of $\HAB^{\otimes N}$ in \cref{lemma-3241548}.

Our construction of the optimal one-way LOCC tester $\hat{\mathcal{T}}$ from Bob to Alice is given as follows: 
\begin{equation} \label{eq-43159}
        \hat{\mathcal{T}}\!\!\coloneqq\!\!\!\bigoplus_{\lambda\vdash (N,d)}\!\!\frac{1}{\dim(\mathcal{W}_{\lambda})\dim(\mathcal{V_{\lambda}})}\cdot
\begin{matrix}
L_\lambda \\ \bbra{I_{\mathcal{V}_\lambda}}S_{\lambda,\lambda\rightarrow\lambda,\lambda}\kett{I_{\mathcal{V}_\lambda}}\otimes I_{\mathcal{W}_{\lambda,\TB}},
\end{matrix}
    \end{equation}
    where
    \begin{align}
        &L_\lambda  = \frac{1}{\dim(\mathcal{V}_\lambda)}\kettbbra{I_{\mathcal{V}_\lambda}}{I_{\mathcal{V}_\lambda}} +\frac{1}{\dim(\mathcal{V}_\lambda)+1}\left(I_{\mathcal{V}_\lambda\otimes\mathcal{V}_\lambda}-\frac{1}{\dim(\mathcal{V}_\lambda)}\kettbbra{I_{\mathcal{V}_\lambda}}{I_{\mathcal{V}_\lambda}}\right), \label{eq-431419} \\
        &S_{\lambda,\lambda\rightarrow\lambda,\lambda} =\frac{1}{N!}\sum_{\pi\in\mathbb{S}_N}
\begin{matrix}P_{\lambda,\TA}(\pi)\\ P_{\lambda,\TB}(\pi)\end{matrix}
\tr_{\mathcal{W}_{\lambda,\TB}}\left[\mathcal{T}_{\lambda,\lambda\rightarrow\lambda,\lambda}\right]
\begin{matrix}P_{\lambda,\TA}(\pi)^\dag\\ P_{\lambda,\TB}(\pi)^\dag\end{matrix}. \nonumber
    \end{align}
The rest part of this section is to show that $\hat{\mathcal{T}}$ is a one-way LOCC tester for $\mathcal{P}$ (although repeated experiments are required to amplify the soundness, with an overhead that depends only on \(c\) and \(s\), which will be explicitly given later).

To illustrate the main idea, we introduce intermediate testers that will be useful in the proof. 
\begin{enumerate}
    \item Let $\widetilde{\mathcal{T}}$ be the tester defined by \cref{eq:def-tildeT}. 
    \item Based on the simultaneous-permutation invariance, we construct a new tester:
    \begin{equation} \label{eq-421128}
        \widetilde{\widetilde{\mathcal{T}}}\coloneqq\frac{1}{N!}\sum_{\pi\in\mathbb{S}_N} 
        \begin{matrix}
            P_{\TA}(\pi) \\
            P_{\TB}(\pi)
        \end{matrix}
        \; \widetilde{\mathcal{T}} \; \begin{matrix}
            P_{\TA}(\pi)^\dag \\
            P_{\TB}(\pi)^\dag
        \end{matrix}.
    \end{equation}
    \item Then, we construct a new tester \(\overline{\mathcal{T}}\), with the final tester \(\hat{\mathcal{T}}\) derived as an LOCC approximation of \(\overline{\mathcal{T}}\):
    \begin{equation} \label{eq-202404041947}
        \overline{\mathcal{T}}\coloneqq\Pi \widetilde{\widetilde{\mathcal{T}}} \Pi,
    \end{equation}
    where
    \begin{equation} \label{eq-421126}
        \Pi=\bigoplus_{\lambda\vdash (N,d)}\frac{1}{\dim(\mathcal{V}_\lambda)}\kettbbra{I_{\mathcal{V}_\lambda}}{I_{\mathcal{V}_\lambda}}\otimes I_{\mathcal{W}_{\lambda,\TA}}\otimes I_{\mathcal{W}_{\lambda,\TB}},
    \end{equation}
    and \(\Pi\) can be viewed as a purity-tester (e.g., see \cite{matsumoto2010test}).
\end{enumerate}

Through an argument similar to \cref{lemma:202404041522}, we can show that \(\widetilde{\mathcal{T}}\) is also a \((c,s)\)-tester for \(\mathcal{P}\) with sample complexity \(N\). 
We characterize some useful properties of $\widetilde{\widetilde{\mathcal{T}}}$ and $\overline{\mathcal{T}}$ in \cref{sec:442030} and \cref{sec-417158}, respectively, and also show that they both are \((c,s)\)-testers for \(\mathcal{P}\) with sample complexity \(N\). 
Finally, in \cref{sec-417159}, we show that $\hat{\mathcal{T}}$ can be implemented by one-way LOCC (see \cref{lemma-43153}), and $\hat{\mathcal{T}}$ is a \((c,s/2+c/2)\)-tester for \(\mathcal{P}\) with sample complexity \(N\) (see \cref{lemma:441955}). Furthermore, we show that \(\hat{\mathcal{T}}\) can be amplified to a \((c,s)\)-tester with an overhead \(100\ln(1/s+10)(c-s)^{-3}\) (see \cref{lemma-4181503}). Therefore, the total sample complexity can be bounded by:
\begin{equation*}
100\ln(1/s+10)(c-s)^{-3}\cdot N\leq 200\ln(1/s+10)(c-s)^{-4} \cdot N_0,
\end{equation*}
where the inequality is because \(N=\max\rbra{N_0,2(c-s)^{-1}}\leq 2(c-s)^{-1}\cdot N_0\) 
(here, we did not try to optimize the constant factor $200\ln(1/s+10)(c-s)^{-4}$ for simplicity). 
\end{proof}

Before proceeding to more details, we note here that the acceptance probability of our LOCC tester is upper and lower bounded via different methods (see \cref{lemma:442015} and \cref{lemma:441955} respectively), where the lower-bound argument uses the condition that $\mathcal{P}^{\textup{yes}}$ contains only pure bipartite states which have the good property in \cref{lemma-3262205}.

\subsection{\texorpdfstring{Characterization of $\widetilde{\widetilde{\mathcal{T}}}$}{Characterization of tilde-tilde-T}} \label{sec:442030}

The following lemma shows that $\widetilde{\widetilde{\mathcal{T}}}$ is a tester for $\mathcal{P}$ and gives useful properties of $\widetilde{\widetilde{\mathcal{T}}}$. 

\begin{lemma} \label{lemma:442014}
    Let $\widetilde{\widetilde{\mathcal{T}}}$ be the tester defined by \cref{eq-421128}. Then,
    \begin{enumerate}
        \item $\widetilde{\widetilde{\mathcal{T}}}$ is a \((c,s)\)-tester for $\mathcal{P}$ with sample complexity $N$. 
        \item $\widetilde{\widetilde{\mathcal{T}}}$ can be written as: \label{item:442330}
    \begin{equation*} 
        \widetilde{\widetilde{\mathcal{T}}} = \sum_{\lambda_1,\lambda_2,\lambda_3\vdash(N,d)}\frac{1}{\dim(\mathcal{W}_{\lambda_2})} S_{\lambda_1,\lambda_2\rightarrow\lambda_3,\lambda_2}\otimes I_{\mathcal{W}_{\lambda_2,\TB}},
    \end{equation*}
    where
    \begin{equation} \label{eq-421436}
    \begin{split}
        \!&S_{\lambda_1,\lambda_2\rightarrow\lambda_3,\lambda_2}\coloneqq 
\frac{1}{N!}\!\!\sum_{\pi\in\mathbb{S}_N}\!\!
\begin{matrix}P_{\lambda_3,\TA}(\pi)\\ P_{\lambda_2,\TB}(\pi)\end{matrix}
\tr_{\mathcal{W}_{\lambda_2,\TB}}\!\left[\mathcal{T}_{\lambda_1,\lambda_2\rightarrow\lambda_3,\lambda_2}\right]
\begin{matrix}P_{\lambda_1,\TA}(\pi)^\dag\\ P_{\lambda_2,\TB}(\pi)^\dag\end{matrix}\!.\!
\end{split}
    \end{equation}
    Note that \(S_{\lambda_1,\lambda_2\rightarrow\lambda_3,\lambda_2} \colon \mathcal{V}_{\lambda_1,\TA}\otimes \mathcal{V}_{\lambda_2,\TB}\otimes \mathcal{W}_{\lambda_1,\TA}\rightarrow \mathcal{V}_{\lambda_3,\TA}\otimes \mathcal{V}_{\lambda_2,\TB}\otimes \mathcal{W}_{\lambda_3,\TA}\) is a linear map.
    \item \(\widetilde{\widetilde{\mathcal{T}}}\) commutes with \(\Pi\), i.e., $\Pi \widetilde{\widetilde{\mathcal{T}}} = \widetilde{\widetilde{\mathcal{T}}} \Pi$, where $\Pi$ is defined by \cref{eq-421126}. \label{lemma:itm:commute}
    \end{enumerate}
\end{lemma}

\begin{proof}
\textbf{Item 1}. Since \(\mathcal{P}\) is unitarily invariant on \(\HB\), through an argument similar to \cref{lemma:202404041522}, we can show that \(\widetilde{\mathcal{T}}\) is also a \((c,s)\)-tester for \(\mathcal{P}\) with sample complexity \(N\). 
It is easy to see that \(\widetilde{\widetilde{\mathcal{T}}}\) is also a \((c,s)\)-tester for \(\mathcal{P}\) with sample complexity \(N\), as \(\widetilde{\widetilde{\mathcal{T}}}\) behaves identically to \(\widetilde{\mathcal{T}}\):
\begin{equation*}
\begin{split}
\tr\left[\widetilde{\widetilde{\mathcal{T}}}\rho_{\textup{AB}}^{\otimes N}\right]
&=\tr\left[\widetilde{\mathcal{T}}\cdot \frac{1}{N!}\sum_{\pi\in\mathbb{S}_N} 
\begin{matrix}
    P_{\TA}(\pi)^\dag \\
    P_{\TB}(\pi)^\dag
\end{matrix}
\rho_{\textup{AB}}^{\otimes N} \begin{matrix}
    P_{\TA}(\pi) \\
    P_{\TB}(\pi)
\end{matrix}\right]\\ 
&=\tr\left[\widetilde{\mathcal{T}}\rho_{\textup{AB}}^{\otimes N}\right].
\end{split}
\end{equation*}
\textbf{Item 2}.
According to \cref{eq:def-tildeT}, we have
\begin{align}
\widetilde{\mathcal{T}} & \coloneqq \int_{U\in\mathbb{U}_d} U_{\textup{B}}^{\otimes N}\mathcal{T}U_{\textup{B}}^{\dag \otimes N}\nonumber \\
&=\sum_{\lambda_1,\lambda_2,\lambda_3,\lambda_4\vdash (N,d)}\int_{U\in\mathbb{U}_d}Q_{\lambda_4,\TB}(U)\mathcal{T}_{\lambda_1,\lambda_2\rightarrow\lambda_3,\lambda_4}Q_{\lambda_2,\TB}(U)^\dag \nonumber \\
&=\sum_{\lambda_1,\lambda_2,\lambda_3\vdash (N,d)} \frac{1}{\dim(\mathcal{W}_{\lambda_2})} \tr_{\mathcal{W}_{\lambda_2,\TB}}\!\left[\mathcal{T}_{\lambda_1,\lambda_2\rightarrow\lambda_3,\lambda_2}\right]\otimes I_{\mathcal{W}_{\lambda_2,\TB}}, \label{eq:442323}
\end{align}
where the third equality is due to \cref{lemma-41116} and \cref{lemma-3262332}.
By \cref{eq-421128} and \cref{eq:442323}, we can write \(\widetilde{\widetilde{\mathcal{T}}}\) as:
\begin{align}
\widetilde{\widetilde{\mathcal{T}}}&=\sum_{\lambda_1,\lambda_2,\lambda_3\vdash(N,d)}\frac{1}{\dim(\mathcal{W}_{\lambda_2})}
\Bigg(\frac{1}{N!}\sum_{\pi\in\mathbb{S}_N}
\begin{matrix}P_{\lambda_3,\TA}(\pi)\\ P_{\lambda_2,\TB}(\pi)\end{matrix}
\tr_{\mathcal{W}_{\lambda_2,\TB}}\left[\mathcal{T}_{\lambda_1,\lambda_2\rightarrow\lambda_3,\lambda_2}\right]
\begin{matrix}P_{\lambda_1,\TA}(\pi)^\dag\\ P_{\lambda_2,\TB}(\pi)^\dag\end{matrix}\Bigg)
\otimes I_{\mathcal{W}_{\lambda_2,\TB}} \nonumber \\
&=\sum_{\lambda_1,\lambda_2,\lambda_3\vdash(N,d)}\frac{1}{\dim(\mathcal{W}_{\lambda_2})} S_{\lambda_1,\lambda_2\rightarrow\lambda_3,\lambda_2}\otimes I_{\mathcal{W}_{\lambda_2,\TB}}. \nonumber
\end{align}
\textbf{Item 3}. 
First, note that
\begin{align}
\Pi\widetilde{\widetilde{\mathcal{T}}}&=\sum_{\lambda\vdash (N,d)}\frac{1}{\dim(\mathcal{V}_\lambda)}\kettbbra{I_{\mathcal{V}_\lambda}}{I_{\mathcal{V}_\lambda}}  \sum_{\lambda_1,\lambda_2,\lambda_3\vdash(N,d)}\frac{1}{\dim(\mathcal{W}_{\lambda_2})} S_{\lambda_1,\lambda_2\rightarrow\lambda_3,\lambda_2}\otimes I_{\mathcal{W}_{\lambda_2,\TB}} \nonumber \\
&=\sum_{\lambda,\lambda_1\vdash (N,d)}\frac{1}{\dim(\mathcal{V}_\lambda)\dim(\mathcal{W}_\lambda)}\kettbbra{I_{\mathcal{V}_\lambda}}{I_{\mathcal{V}_\lambda}} S_{\lambda_1,\lambda\rightarrow\lambda,\lambda}\otimes I_{\mathcal{W}_{\lambda,\TB}}, \label{eq-421915}
\end{align}
where the second equality is because \(\kettbbra{I_{\mathcal{V}_\lambda}}{I_{\mathcal{V}_\lambda}} \colon \mathcal{V}_{\lambda,\TA}\otimes \mathcal{V}_{\lambda,\TB}\rightarrow\mathcal{V}_{\lambda,\TA}\otimes\mathcal{V}_{\lambda,\TB}\) and \(S_{\lambda_1,\lambda_2\rightarrow\lambda_3,\lambda_2}\colon \mathcal{V}_{\lambda_1,\TA}\otimes \mathcal{V}_{\lambda_2,\TB}\otimes \mathcal{W}_{\lambda_1,\TA}\rightarrow \mathcal{V}_{\lambda_3,\TA}\otimes \mathcal{V}_{\lambda_2,\TB}\otimes \mathcal{W}_{\lambda_3,\TA}\), thus \(\kettbbra{I_{\mathcal{V}_\lambda}}{I_{\mathcal{V}_\lambda}}S_{\lambda_1,\lambda_2\rightarrow\lambda_3,\lambda_2}\) is non-zero only if \(\lambda_3=\lambda_2=\lambda\).

Now, we look deeper into \(S_{\lambda_1,\lambda\rightarrow\lambda,\lambda}\). By its definition (see \cref{eq-421436}), \(S_{\lambda_1,\lambda\rightarrow\lambda,\lambda}\) commutes with the actions of  simultaneous-permutation. Therefore, by Schur's lemma (\cref{prop:202404032229}), it has the following form:
\begin{equation}\label{eq-421913}
S_{\lambda_1,\lambda\rightarrow\lambda,\lambda}=\bigoplus_{\substack{\mathcal{V}_{\lambda'}\stackrel{\mathbb{S}_N}{\subset} \mathcal{V}_{\lambda_1}\otimes \mathcal{V}_{\lambda}\\ \mathcal{V}_{\lambda'}\stackrel{\mathbb{S}_N}{\subset} \mathcal{V}_{\lambda}\otimes \mathcal{V}_{\lambda}}} I_{\mathcal{V}_{\lambda'}} \otimes R_{\lambda_1,\lambda\rightarrow\lambda,\lambda}^{\lambda'},
\end{equation}
where \(\mathcal{V}_{\lambda'}\stackrel{\mathbb{S}_N}{\subset} \mathcal{V}_{\lambda_1}\otimes \mathcal{V}_{\lambda}\) means \(\mathcal{V}_{\lambda'}\) occurs in \(\mathcal{V}_{\lambda_1}\otimes\mathcal{V}_{\lambda}\) as a representation of \(\mathbb{S}_N\) (and similarly for \(\mathcal{V}_{\lambda'}\stackrel{\mathbb{S}_N}{\subset} \mathcal{V}_{\lambda}\otimes \mathcal{V}_{\lambda}\)). 
Moreover, \(R^{\lambda'}_{\lambda_1,\lambda\rightarrow\lambda,\lambda}\colon \mathcal{M}_{\lambda_1,\lambda}^{\lambda'}\rightarrow\mathcal{M}_{\lambda,\lambda}^{\lambda'}\) is a linear map, where \(\mathcal{M}_{\lambda_1,\lambda}^{\lambda'}\) is the multiplicity space of \(\mathcal{V}_{\lambda'}\) in \(\mathcal{V}_{\lambda_1}\otimes \mathcal{V}_\lambda\) (and similarly for \(\mathcal{M}_{\lambda,\lambda}^{\lambda'}\)). On the other hand, one can see that
\begin{equation*}
P_{\lambda}(\pi)\otimes P_{\lambda}(\pi) \kett{I_{\mathcal{V}_\lambda}} = \kett{P_{\lambda}(\pi) P_{\lambda}(\pi)^\textup{T}} =\kett{I_{\mathcal{V}_\lambda}},
\end{equation*}
which means \(\spanspace(\kett{I_{\mathcal{V}_\lambda}})\) is a trivial representation of \(\mathbb{S}_N\) (corresponding to the Young diagram with only one row, i.e., \(\scalebox{.45}{\begin{ytableau}~&~&\none[\cdots]&~\end{ytableau}}\)) in \(\mathcal{V}_{\lambda}\otimes\mathcal{V}_{\lambda}\). However, \(\mathcal{V}_{\scalebox{.35}{\begin{ytableau}~&~&\none[\cdots]&~\end{ytableau}}}\) does not occur in \(\mathcal{V}_{\lambda_1}\otimes\mathcal{V}_{\lambda}\) for \(\lambda_1\neq \lambda\), since otherwise there exists a non-zero \(\ket{x}\in \mathcal{V}_{\lambda_1}\otimes\mathcal{V}_{\lambda}\) such that:
\begin{equation*}
\frac{1}{N!}\sum_{\pi\in\mathbb{S}_N}P_{\lambda_1}(\pi)\otimes P_{\lambda}(\pi) \ket{x}=\ket{x},
\end{equation*}
but this contradicts \cref{eq-3252254}. Therefore, for \(\lambda_1\neq \lambda\),
\begin{equation*}
\kettbbra{I_{\mathcal{V}_\lambda}}{I_{\mathcal{V}_\lambda}} S_{\lambda_1,\lambda\rightarrow\lambda,\lambda}=0,
\end{equation*}
since \(\spanspace(\kett{I_{\mathcal{V}_\lambda}})\stackrel{\mathbb{S}_N}{\cong}\mathcal{V}_{\scalebox{.35}{\begin{ytableau}~&~&\none[\cdots]&~\end{ytableau}}}\) but \(\mathcal{V}_{\scalebox{.35}{\begin{ytableau}~&~&\none[\cdots]&~\end{ytableau}}}\) does not occur in the summation of \cref{eq-421913}. Therefore, \(\Pi\widetilde{\widetilde{\mathcal{T}}}\) (see \cref{eq-421915}) can be written as
\begin{equation*}
\Pi\widetilde{\widetilde{\mathcal{T}}}\!=\!\!\bigoplus_{\lambda\vdash (N,d)}\!\!\frac{1}{\dim(\mathcal{V}_\lambda)\dim(\mathcal{W}_\lambda)}\kettbbra{I_{\mathcal{V}_\lambda}}{I_{\mathcal{V}_\lambda}}S_{\lambda,\lambda\rightarrow\lambda,\lambda}\otimes I_{\mathcal{W}_{\lambda,\TB}}.
\end{equation*}
According to symmetry, a similar form also holds for \(\widetilde{\widetilde{\mathcal{T}}}\Pi\).

Now, to prove that \(\Pi\) commutes with \(\widetilde{\widetilde{\mathcal{T}}}\), it suffices to prove that \(\kettbbra{I_{\mathcal{V}_\lambda}}{I_{\mathcal{V}_\lambda}}\) commutes with \(S_{\lambda,\lambda\rightarrow\lambda,\lambda}\). First, by \cref{eq-421913}, we have
\begin{equation*}
S_{\lambda,\lambda\rightarrow \lambda,\lambda}=\bigoplus_{\mathcal{V}_{\lambda'}\stackrel{\mathbb{S}_N}{\subset} \mathcal{V}_{\lambda}\otimes \mathcal{V}_{\lambda}} I_{\mathcal{V}_{\lambda'}}\otimes R_{\lambda,\lambda\rightarrow\lambda,\lambda}^{\lambda'}.
\end{equation*}
Since \(\spanspace(\kett{I_{\mathcal{V}_\lambda}})\stackrel{\mathbb{S}_N}{\cong}\mathcal{V}_{\scalebox{.35}{\begin{ytableau}~&~&\none[\cdots]&~\end{ytableau}}}\), we only need to consider the case of \(\lambda'=\scalebox{.45}{\begin{ytableau}~&~&\none[\cdots]&~\end{ytableau}}\).
Note that the multiplicity of \(\mathcal{V}_{\scalebox{.35}{\begin{ytableau}~&~&\none[\cdots]&~\end{ytableau}}}\) in \(\mathcal{V}_{\lambda}\otimes\mathcal{V}_{\lambda}\) is always \(1\). This is because: suppose \(\ket{x}\in\mathcal{V}_{\lambda}\otimes\mathcal{V}_\lambda\) is in a \(\mathcal{V}_{\scalebox{.35}{\begin{ytableau}~&~&\none[\cdots]&~\end{ytableau}}}\), then it satisfies
\begin{equation*}
\frac{1}{N!}\sum_{\pi} P_{\lambda}(\pi)\otimes P_{\lambda}(\pi)\ket{x}=\ket{x},
\end{equation*}
but by \cref{eq-3252235}, \(\ket{x}\) can only be proportional to \(\kett{I_{\mathcal{V}_\lambda}}\). Therefore, \(R_{\lambda,\lambda\rightarrow\lambda,\lambda}^{\scalebox{.35}{\begin{ytableau}~&~&\none[\cdots]&~\end{ytableau}}}\) is trivially a \(1\times 1\) matrix, and \(I_{\mathcal{V}_{\scalebox{.35}{\begin{ytableau}~&~&\none[\cdots]&~\end{ytableau}}}}\) is just \(\frac{1}{\dim(\mathcal{V}_{\lambda})}\kettbbra{I_{\mathcal{V}_\lambda}}{I_{\mathcal{V}_\lambda}}\), which means \(\kettbbra{I_{\mathcal{V}_\lambda}}{I_{\mathcal{V}_\lambda}}\) commutes with \(S_{\lambda,\lambda\rightarrow\lambda,\lambda}\).

\end{proof}

\subsection{\texorpdfstring{Characterization of $\overline{\mathcal{T}}$}{Characterization of overline-T}}\label{sec-417158}

The following lemma shows that $\overline{\mathcal{T}}$ is a \((c,s)\)-tester for $\mathcal{P}$ and gives a useful identity for $\overline{\mathcal{T}}$. 

\begin{lemma} \label{lemma:442015}
    Let $\overline{\mathcal{T}}$ be the tester defined by \cref{eq-202404041947}. Then,
    \begin{enumerate}
        \item $\overline{\mathcal{T}}$ is a \((c,s)\)-tester for $\mathcal{P}$ with sample complexity $N$.
        \item $\overline{\mathcal{T}}$ can be written as:\label{item-46312}
        \begin{equation*}
        \begin{split}
        &\overline{\mathcal{T}}=\bigoplus_{\lambda\vdash (N,d)}\frac{1}{\dim(\mathcal{W}_{\lambda})\dim(\mathcal{V_{\lambda}})}
\cdot \begin{matrix}\frac{1}{\dim(\mathcal{V_{\lambda}})}\kettbbra{I_{\mathcal{V}_\lambda}}{I_{\mathcal{V}_\lambda}}\\
\bbra{I_{\mathcal{V}_\lambda}}S_{\lambda,\lambda\rightarrow \lambda,\lambda}\kett{I_{\mathcal{V}_\lambda}}\otimes I_{\mathcal{W}_{\lambda,\TB}}\end{matrix},
    \end{split}
    \end{equation*}
    where \(\bbra{I_{\mathcal{V}_\lambda}}S_{\lambda,\lambda\rightarrow \lambda,\lambda}\kett{I_{\mathcal{V}_\lambda}}\colon\mathcal{W}_{\lambda,\TA}\rightarrow\mathcal{W}_{\lambda,\TA}\) is a positive operator, and since \(\overline{\mathcal{T}}\) is a valid measurement (by its definition), the operator norm of \(\frac{1}{\dim(\mathcal{W}_\lambda)\dim(\mathcal{V}_\lambda)}\bbra{I_{\mathcal{V}_\lambda}}S_{\lambda,\lambda\rightarrow \lambda,\lambda}\kett{I_{\mathcal{V}_\lambda}}\) is less than or equal to \(1\). 
    \end{enumerate}
\end{lemma}

\begin{proof}
    \textbf{Item 1}. 
    We first show that \(\overline{\mathcal{T}}\) is also a \((c,s)\)-tester for \(\mathcal{P}\) with sample complexity \(N\). First, if \(\ket{\psi}_{\textup{AB}}\in\mathcal{P}^{\textup{yes}}\), then
\begin{equation*}
\begin{split}
&\tr\left[\overline{\mathcal{T}}\ketbra{\psi}{\psi}_{\textup{AB}}^{\otimes N}\right]=\tr\left[\widetilde{\widetilde{\mathcal{T}}}\Pi \ketbra{\psi}{\psi}^{\otimes N}_{\textup{AB}}\Pi\right]\\
=&\tr\left[\widetilde{\widetilde{\mathcal{T}}}\ketbra{\psi}{\psi}_{\textup{AB}}^{\otimes N}\right]\geq c,
\end{split}
\end{equation*}
where the second equality is because \(\ket{\psi}_{\textup{AB}}^{\otimes N}=\sum_{\lambda\vdash (N,d)}\kett{I_{\mathcal{V}_\lambda}} \otimes \ket{w_\lambda}\) by \cref{lemma-3262205}, and therefore is in the support space of \(\Pi\). On the other hand, if \(\rho_{\textup{AB}}\in\mathcal{P}^{\textup{no}}\), then
\begin{equation*}
\tr\left[\overline{\mathcal{T}}\rho_{\textup{AB}}^{\otimes N}\right]=\tr\left[\Pi\widetilde{\widetilde{\mathcal{T}}}\Pi\rho_{\textup{AB}}^{\otimes N}\right]\leq \tr\left[\widetilde{\widetilde{\mathcal{T}}}\rho_{\textup{AB}}^{\otimes N}\right]\leq s,
\end{equation*}
where the first inequality is because \(\Pi\sqsubseteq I\) and \(\Pi\) commutes with \(\widetilde{\widetilde{\mathcal{T}}}\), due to Item \ref{lemma:itm:commute} of \cref{lemma:442014}.

\textbf{Item 2}. 
It follows directly from \cref{eq-202404041947} and Item \ref{item:442330} of \cref{lemma:442014}. 
\end{proof}

\subsection{\texorpdfstring{Validity of $\hat{\mathcal{T}}$}{Validity of hat-T}}\label{sec-417159}

Note that \(\overline{\mathcal{T}}\) is still not LOCC implementable, since the part \(\frac{1}{\dim(\mathcal{V_{\lambda}})}\kettbbra{I_{\mathcal{V}_\lambda}}{I_{\mathcal{V}_\lambda}}\) is not LOCC implementable when \(\dim(\mathcal{V}_\lambda)> 1\). 
Nevertheless, as suggested by \cite{matsumoto2010test,hayashi2006study}, there is an LOCC approximation for the projector of maximally entangled state (see \cref{lemma-42239}). 
Based on this, we first show that $\hat{\mathcal{T}}$ is LOCC implementable. 

\begin{lemma}\label{lemma-43153}
The tester \(\hat{\mathcal{T}}\) defined by \cref{eq-43159} is a one-way LOCC tester from Bob to Alice.
\end{lemma}
\begin{proof}
Note that the operator norm of \(\frac{1}{\dim(\mathcal{W}_\lambda)\dim(\mathcal{V}_\lambda)}\bbra{I_{\mathcal{V}_\lambda}}S_{\lambda,\lambda\rightarrow \lambda,\lambda}\kett{I_{\mathcal{V}_\lambda}}\) is less than \(1\) by Item~\ref{item-46312} of \cref{lemma:442015}. Thus, from the definition of \(\hat{\mathcal{T}}\) (see \cref{eq-43159}), one can see that \(\hat{\mathcal{T}}\) is a valid measurement.

By \cref{lemma-42239}, \(L_\lambda\) (as an LOCC approximation of \(\frac{1}{\dim(\mathcal{V_{\lambda}})}\kettbbra{I_{\mathcal{V}_\lambda}}{I_{\mathcal{V}_\lambda}}\)) defined by \cref{eq-431419} can be written as
\begin{equation*}
L_\lambda=\int_{U_\lambda\in\mathbb{U}_{d_\lambda}} \begin{matrix}U_\lambda\\U_\lambda^*\end{matrix}\left(\sum_{i=1}^{d_\lambda} \begin{matrix}\ketbra{i}{i}_{\mathcal{V}_{\lambda,\TA}}\\\ketbra{i}{i}_{\mathcal{V}_{\lambda,\TB}}\end{matrix}\right)\begin{matrix}U_\lambda^\dag\\U_\lambda^{*\dag}\end{matrix},
\end{equation*}
where \(d_\lambda\coloneqq\dim(\mathcal{V}_\lambda)\) and \(\{\ket{i}_{\mathcal{V}_\lambda}\}_{i}\) is a standard basis of \(\mathcal{V}_\lambda\). Then, \(\hat{\mathcal{T}}\) defined by \cref{eq-43159} can be written as
\begin{equation*}
\begin{split}
\hat{\mathcal{T}}
=&\bigoplus_{\lambda\vdash (N,d)}\int_{U_\lambda\in \mathbb{U}_{d_\lambda}} \sum_{i=1}^{d_\lambda} \,\,
\frac{1}{\dim(\mathcal{W}_{\lambda})\dim(\mathcal{V_{\lambda}})} 
\begin{matrix}
\begin{matrix}U_\lambda\\U_\lambda^*\end{matrix}\left(\begin{matrix}\ketbra{i}{i}_{\mathcal{V}_{\lambda,\TA}}\\\ketbra{i}{i}_{\mathcal{V}_{\lambda,\TB}}\end{matrix}\right)\begin{matrix}U_\lambda^\dag\\U_\lambda^{*\dag}\end{matrix} \\ 
\bbra{I_{\mathcal{V}_\lambda}}S_{\lambda,\lambda\rightarrow \lambda,\lambda}\kett{I_{\mathcal{V}_\lambda}}\otimes I_{\mathcal{W}_{\lambda,\TB}}
\end{matrix}
\end{split}
\end{equation*}
Therefore, \(\hat{\mathcal{T}}\) can be implemented by the following strategy:
\begin{enumerate}
    \item For each \(\lambda\vdash (N,d)\), Bob samples a Haar random unitary \(U_\lambda\in\mathbb{U}_{d_\lambda}\), where \(d_\lambda\coloneqq\mathcal{V}_\lambda\).

    \item Bob performs the measurement 
    \begin{equation*}
    \left\{\begin{matrix} 
    U_\lambda^*\ketbra{i}{i}_{\mathcal{V}_{\lambda,\TB}} U_{\lambda}^{*\dag}\\
    I_{\mathcal{W}_{\lambda,\TB}}
    \end{matrix}\right\}_{\lambda\vdash (N,d),1\leq i\leq d_\lambda},
    \end{equation*}
    and obtain the measurement result \((\lambda_\TB,i_\TB)\).

    \item Bob sends the classical message \((\{U_\lambda\}_{\lambda\vdash (N,d)}, \lambda_\TB,i_\TB)\) to Alice.

    \item Alice performs the measurement
    \begin{equation*}
    \left\{
    \begin{matrix}
    U_{\lambda}\ketbra{i}{i}_{\mathcal{V}_{\lambda,\TA}} U_{\lambda}^\dag\\
    \frac{1}{\dim(\mathcal{W}_\lambda)\dim(\mathcal{V}_\lambda)}\bbra{I_{\mathcal{V}_\lambda}}S_{\lambda,\lambda\rightarrow\lambda,\lambda}\kett{I_{\mathcal{V}_\lambda}}
    \end{matrix}
    \right\}_{\lambda\vdash (N,d),1\leq i \leq d_{\lambda}},
    \end{equation*}
    and obtain the measurement result \((\lambda_\TA,i_\TA)\).

    \item Alice accepts if \(\lambda_\TA=\lambda_\TB\) and \(i_\TA=i_\TB\).
\end{enumerate}
\end{proof}

Then, we show that $\hat{\mathcal{T}}$ is a \((c,s/2+c/2)\)-tester for $\mathcal{P}$ with sample complexity \(N\). 

\begin{lemma} \label{lemma:441955}
    Let $\hat{\mathcal{T}}$ be the tester defined by \cref{eq-43159}. Then \(\hat{\mathcal{T}}\) is a \((c,s/2+c/2)\)-tester for \(\mathcal{P}\) with sample complexity \(N\), i.e.,
    \begin{enumerate}
        \item For every $\rho \in \mathcal{P}^{\textup{yes}}$, $\tr\sbra{\hat{\mathcal{T}} \rho^{\otimes N}} \geq c$.
        \item For every $\rho \in \mathcal{P}^{\textup{no}}$, $\tr\sbra{\hat{\mathcal{T}} \rho^{\otimes N}} \leq s/2+c/2$.
    \end{enumerate}
\end{lemma}

\begin{proof}
    Note that $L_\lambda$ is actually an LOCC approximation of \(\frac{1}{\dim(\mathcal{V_{\lambda}})}\kettbbra{I_{\mathcal{V}_\lambda}}{I_{\mathcal{V}_\lambda}}\) satisfying
    \begin{equation*}
        \Abs*{L_\lambda - \frac{1}{\dim(\mathcal{V_{\lambda}})}\kettbbra{I_{\mathcal{V}_\lambda}}{I_{\mathcal{V}_\lambda}}} \leq \frac{1}{\dim(\mathcal{V}_\lambda)+1}.
    \end{equation*}
    Furthermore, when \(\dim(\mathcal{V}_\lambda)=1\), 
    \begin{equation*}
        L_\lambda = \frac{1}{\dim(\mathcal{V_{\lambda}})}\kettbbra{I_{\mathcal{V}_\lambda}}{I_{\mathcal{V}_\lambda}}.
    \end{equation*}
Note that \(\hat{\mathcal{T}}\sqsupseteq\overline{\mathcal{T}}\) since $L_\lambda \sqsupseteq \frac{1}{\dim(\mathcal{V_{\lambda}})}\kettbbra{I_{\mathcal{V}_\lambda}}{I_{\mathcal{V}_\lambda}}$. Thus \(\hat{\mathcal{T}}\) will accept \(\ket{\psi}_{\textup{AB}}\in\mathcal{P}^{\textup{yes}}\) with probability \(\geq c\). If \(\rho_{\textup{AB}}\in\mathcal{P}^{\textup{no}}\), then
\begin{equation*}
\begin{split}
\tr\left[\hat{\mathcal{T}}\rho_{\textup{AB}}^{\otimes N}\right] &\leq \tr\left[\overline{\mathcal{T}}\rho_{\textup{AB}}^{\otimes N}\right]+\max_{\substack{\lambda\vdash (N,d)\\\dim(\mathcal{V}_\lambda)>1}}\frac{1}{\dim(\mathcal{V}_\lambda)+1} \\
& \leq  s + \frac{1}{N}\leq s/2+c/2,
\end{split}
\end{equation*}
where the second inequality is because the smallest dimension of \(\mathcal{V}_\lambda\) other than \(1\) is \(N-1\),\footnote{The dimension of \(\mathcal{V}_\lambda\) equals the number of standard Young tableaux with shape \(\lambda\), and can be calculated by the hook length formula~\cite{fulton2013representation}. Based on it, one can prove that for any Young diagram \(\lambda\) with \(N\) boxes, if \(\lambda\) is not a one-column or one-row Young diagram, then there exists more than \(N-1\) standard Young tableaux with shape \(\lambda\).} and the third inequality is because \(N=\max\rbra{N_0,2(c-s)^{-1}}\geq 2(c-s)^{-1}\) (see \cref{eq-423150}).
\end{proof}

Next, we provide \cref{lemma-4181503}, by which we can amplify the $(c,s/2+c/2)$-tester \(\hat{\mathcal{T}}\) in \cref{lemma:441955} to a \((c,s)\)-tester with an overhead of $100\ln(1/s+10)(c-s)^{-3}$.

\begin{lemma}\label{lemma-4181503}
If \(0<s<c\leq 1\), then a \((c,s/2+c/2)\)-tester can be amplified to a \((c,s)\)-tester by repeating the tester for \(100\ln(1/s+10)(c-s)^{-3}\) times.
\end{lemma}
\begin{proof}
The strategy is simple. We repeat the \((c,s/2+c/2)\)-tester for \(n\coloneqq 100\ln(1/s+10)(c-s)^{-3}\) times, and obtain the results \(\{X_i\}_{i=1}^n\), where each \(X_i\in\{0,1\}\) is an output of the \((c,s/2+c/2)\)-tester with \(0\) corresponding to ``reject'' and \(1\) corresponding to ``accept''. Then, we compute the average value \(X_{\textup{avg}}=\sum_{i=1}^n X_i/n\), and accept iff \(X_{\textup{avg}}\geq t\), where \( t \coloneqq 3c/4+s/4\) is a threshold value. We define \(\Delta\coloneqq c-t=t-(s/2+c/2)=(c-s)/4\) and note that \(n>6\ln(1/s+10)\Delta^{-3}\).

If the tested state is from \(\mathcal{P}^{\textup{no}}\), then the \((c,s/2+c/2)\)-tester accepts it with probability \(\Pr\sbra{X_i=1}\leq s/2+c/2 \). By Hoeffding's inequality (\cref{lemma:hoef}), the final acceptance probability can be bounded as:
\begin{equation*}
\begin{split}
\Pr\sbra*{X_{\textup{avg}}\geq t}&\leq \exp\left(-2n \left(t-\frac{s+c}{2}\right)^2 \right)\\ 
&=\exp\left(-2n\Delta^2 \right)<\exp\left(-\ln(1/s+10)\right)<s
\end{split}
\end{equation*}

If the tested state is from \(\mathcal{P}^{\textup{yes}}\), then the \((c,s/2+c/2)\)-tester accepts it with probability \(\Pr\sbra{X_i=1}\geq c \). First, if \(1-c\geq \exp(-2/\Delta)\), then by Hoeffding's inequality (\cref{lemma:hoef}), we have
\begin{equation*}
\begin{split}
\Pr\sbra*{X_{\textup{avg}}< t} & \leq \exp\left(-2n \left(c-t\right)^2 \right)\\
& = \exp\left(-2n\Delta^2\right)
<\exp\left(-2/\Delta\right)\leq 1-c.
\end{split}
\end{equation*}
Otherwise, suppose \(1-c<\exp(-2/\Delta)\). Note that \(\Delta=(c-s)/4< 1/4\), and \(c> 1-\exp(-8)>2/3\), and thus \(t>3c/4>1/2\). By the Chernoff-Hoeffding inequality (\cref{lemma:cher-hoef}), we have
\begin{align}
\Pr\sbra*{X_{\textup{avg}}< t}
\leq &\exp\Bigl( -n D_{\textup{KL}}(t \Vert c)\Bigr) \nonumber \\ 
=&\exp\left(-n \left[t\ln\left(\frac{t}{c}\right) + (1-t)\ln\left(\frac{1-t}{1-c}\right)\right]\right)\nonumber \\
\leq &\exp\left(n \left[t\ln\left(\frac{1}{t}\right) + (1-t)\ln\left(\frac{1}{1-t}\right)\right]\right)\cdot (1-c)^{(1-t)n}\nonumber \\
\leq &\exp\left(n\left[2(1-t)\ln\left(\frac{1}{1-t}\right)\right]\right)\cdot (1-c)^{(1-t)n}\nonumber \\
=&\left(\frac{1-c}{(1-t)^2}\right)^{(1-t)n},\nonumber 
\end{align}
where the third inequality is because \(t\ln(1/t)\leq (1-t)\ln(1/(1-t))\) when \(t>1/2\). Then, note that \(1-c< \exp(-2/\Delta)<\Delta^3<(1-t)^3\), thus we have
\begin{equation*}
\left(\frac{1-c}{(1-t)^2}\right)^{(1-t)n}\leq \Big(\sqrt[3]{1-c}\Big)^{(1-t)n}\leq \left(1-c\right)^{\frac{\Delta}{3}n}<1-c.
\end{equation*}
Therefore, we have proved that \(\Pr\sbra*{X_{\textup{avg}}<t}\leq 1-c\), which means \(\Pr\sbra*{X_{\textup{avg}}\geq t}\geq c\).
\end{proof}

\subsection{Technical Lemmas}

\begin{lemma}[{\cite[Theorem 1]{hayashi2006study}}]\label{lemma-42239}
Suppose \(\mathcal{H}\) is a \(d\)-dimensional Hilbert space. Then, 
\begin{equation}\label{eq-43134}
\begin{split}
&\frac{1}{d}\kettbbra{I_\mathcal{H}}{I_\mathcal{H}}+\frac{1}{d+1}(I-\frac{1}{d}\kettbbra{I_{\mathcal{H}}}{I_{\mathcal{H}}}) 
=\int_{U\in\mathbb{U}_d} \begin{matrix}U\\U^*\end{matrix}\left(\sum_{i=1}^d \begin{matrix}\ketbra{i}{i}\\\ketbra{i}{i}\end{matrix}\right)\begin{matrix}U^\dag\\U^{*\dag}\end{matrix}.
\end{split}
\end{equation}
Thus, the measurement \(\frac{1}{d}\kettbbra{I_\mathcal{H}}{I_\mathcal{H}}+\frac{1}{d+1}(I-\frac{1}{d}\kettbbra{I_{\mathcal{H}}}{I_{\mathcal{H}}})\) can be implemented by one-way LOCC.
\end{lemma}
\begin{proof}
Note that \((U\mapsto U\otimes U^*,\mathcal{H}\otimes\mathcal{H})\) is a representation of \(\mathbb{U}_d\), which contains two irreducible subspaces
\begin{equation*}
\spanspace(\kett{I_\mathcal{H}}),\,\, \{\kett{X}:\tr(X)=0, X\textup{ is a linear map on }\mathcal{H}\}.
\end{equation*}
Therefore, 
\begin{equation}\label{eq-43132}
\begin{split}
&\int_{U\in\mathbb{U}_d} \begin{matrix}U\\U^*\end{matrix}\left(\sum_{i=1}^d \begin{matrix}\ketbra{i}{i}\\\ketbra{i}{i}\end{matrix}\right)\begin{matrix}U^\dag\\U^{*\dag}\end{matrix} 
=\frac{\alpha}{d}\kettbbra{I_{\mathcal{H}}}{I_{\mathcal{H}}}+\beta\rbra*{I-\frac{1}{d}\kettbbra{I_{\mathcal{H}}}{I_{\mathcal{H}}}},
\end{split}
\end{equation}
for some \(\alpha,\beta\geq 0\), where \(I-\frac{1}{d}\kettbbra{I_{\mathcal{H}}}{I_{\mathcal{H}}}\) can be viewed as the projector onto the subspace \(\{\kett{X}\,|\,\tr(X)=0, X\colon\mathcal{H}\rightarrow\mathcal{H}\textup{ is a linear map}\}\). Note that,
\begin{equation*}
 \frac{1}{d} \bbra{I_{\mathcal{H}}}  \sum_{i=1}^d \begin{matrix}\ketbra{i}{i}\\\ketbra{i}{i}\end{matrix} \kett{I_\mathcal{H}}=1,
\end{equation*}
which implies \(\alpha=1\). This further implies \(\beta=\frac{1}{d+1}\), since the RHS and LHS of \cref{eq-43132} have the same trace.

Then, we can see that the RHS of \cref{eq-43134} can be implemented in the following way: Bob (the second system) first samples a Haar random unitary, and performs the measurement \(\{U^*\ket{i}\}_i\), and then send both the unitary \(U\) and the measurement result \(i_\TB\) to Alice (the first system), and then Alice performs the measurement \(\{U\ket{i}\}_i\) to obtain a result \(i_\TA\), and accepts if \(i_\TA=i_\TB\).
\end{proof}

We also need the following probability inequalities for amplifying the completeness and soundness of quantum testers. 

\begin{lemma} [Chernoff-Hoeffding, {\cite[Theorem 1]{hoeffding1994probability}}] \label{lemma:cher-hoef}
    Suppose that $X_1, X_2, \dots, X_n$ are independent and identical random variables and $0 \leq X_i \leq 1$ for $1 \leq i \leq n$. 
    Let $p = \mathbb{E}\sbra{X_1}$. 
    Then,
    \begin{equation*}
    \begin{split}
        \Pr\sbra*{\frac{1}{n} \sum_{i=1}^n X_i \leq t} \leq&\, \rbra*{\rbra*{\frac{p}{t}}^t \rbra*{\frac{1-p}{1-t}}^{1-t}}^n 
        =\, \exp\rbra*{-n \mathrm{D}_{\textup{KL}}\rbra{t\Vert p}},
    \end{split}
    \end{equation*}
    where
    \begin{equation*}
        \mathrm{D}_{\textup{KL}}\rbra{x\Vert y} = x \ln\rbra*{\frac{x}{y}} + \rbra{1-x}\ln\rbra*{\frac{1-x}{1-y}}
    \end{equation*}
    is the Kullback–Leibler divergence. 
\end{lemma}

\begin{lemma} [Hoeffding, {\cite[Theorem 2]{hoeffding1994probability}}] \label{lemma:hoef}
    Suppose that $X_1, X_2, \dots, X_n$ are independent and identical random variables and $0 \leq X_i \leq 1$ for $1 \leq i \leq n$. 
    Let $p = \mathbb{E}\sbra{X_1}$. 
    Then, 
    \begin{align}
        \Pr\sbra*{\frac{1}{n} \sum_{i=1}^n X_i - p \leq t} & \leq \exp\rbra*{-2nt^2}, \nonumber \\
        \Pr\sbra*{\frac{1}{n} \sum_{i=1}^n X_i - p \geq t} & \leq \exp\rbra*{-2nt^2}. \nonumber
    \end{align}
\end{lemma}

\section*{Acknowledgment}

The authors thank Yupan Liu for pointing out the related works \cite{soleimanifar2022testing,aaronson2022quantum,JW24}, John Wright and Mehdi Soleimanifar for communication regarding the related work \cite{soleimanifar2022testing}, Benjamin Lovitz and Angus Lowe for communication regarding the related work \cite{LL24}, Zhenhuan Liu for pointing out the related work~\cite{liu2024separation} and raising question~\ref{sec-310104} in \cref{sec-4201540} about single-copy testers, Atsuya Hasegawa for independently asking the same question, and Wang Fang for valuable discussion about matrix product states.

Part of the work of Kean Chen was done when the author was with the Institute of Software, Chinese Academy of Sciences, China.
Part of the work of Qisheng Wang was done when the author was with the Graduate School of Mathematics, Nagoya University, Nagoya, Japan. 
Part of the work of Zhicheng Zhang was done when the author was with the Centre for Quantum Software and Information, University of Technology Sydney, Australia.

\addcontentsline{toc}{section}{References}

\bibliographystyle{alphaurl}
\bibliography{main}

\newcommand{\etalchar}[1]{$^{#1}$}
\begin{thebibliography}{PGVWC07}

\bibitem[AAG22]{anshu2022area}
Anurag Anshu, Itai Arad, and David Gosset.
\newblock An area law for 2d frustration-free spin systems.
\newblock In {\em Proceedings of the 54th Annual ACM SIGACT Symposium on Theory of Computing}, pages 12--18, 2022.
\newblock \href {https://doi.org/10.1145/3519935.3519962} {\path{doi:10.1145/3519935.3519962}}.

\bibitem[ABF{\etalchar{+}}24]{aaronson2022quantum}
Scott Aaronson, Adam Bouland, Bill Fefferman, Soumik Ghosh, Umesh Vazirani, Chenyi Zhang, and Zixin Zhou.
\newblock Quantum pseudoentanglement.
\newblock In {\em Proceedings of the 15th Innovations in Theoretical Computer Science Conference}, pages 2:1--2:21, 2024.
\newblock \href {https://doi.org/10.4230/LIPIcs.ITCS.2024.2} {\path{doi:10.4230/LIPIcs.ITCS.2024.2}}.

\bibitem[ACQ22]{aharonov2022quantum}
Dorit Aharonov, Jordan Cotler, and Xiao-Liang Qi.
\newblock Quantum algorithmic measurement.
\newblock {\em Nature communications}, 13(1):887, 2022.
\newblock \href {https://doi.org/10.1038/s41467-021-27922-0} {\path{doi:10.1038/s41467-021-27922-0}}.

\bibitem[AHL{\etalchar{+}}14]{AHL+14}
Dorit Aharonov, Aram~W. Harrow, Zeph Landau, Daniel Nagaj, Mario Szegedy, and Umesh Vazirani.
\newblock Local tests of global entanglement and a counterexample to the generalized area law.
\newblock In {\em Proceedings of the 55th IEEE Annual Symposium on Foundations of Computer Science}, pages 246--255, 2014.
\newblock \href {https://doi.org/10.1109/FOCS.2014.34} {\path{doi:10.1109/FOCS.2014.34}}.

\bibitem[ALL22]{anshu2022distributed}
Anurag Anshu, Zeph Landau, and Yunchao Liu.
\newblock Distributed quantum inner product estimation.
\newblock In {\em Proceedings of the 54th Annual ACM SIGACT Symposium on Theory of Computing}, pages 44--51, 2022.
\newblock \href {https://doi.org/10.1145/3519935.3519974} {\path{doi:10.1145/3519935.3519974}}.

\bibitem[ALV12]{ALV12}
Itai Arad, Zeph Landau, and Umesh Vazirani.
\newblock Improved one-dimensional area law for frustration-free systems.
\newblock {\em Physical Review B}, 85(19):195145, 2012.
\newblock \href {https://doi.org/10.1103/physrevb.85.195145} {\path{doi:10.1103/physrevb.85.195145}}.

\bibitem[BBC{\etalchar{+}}93]{BBC+93}
Charles~H. Bennett, Gilles Brassard, Claude Cr{\'{e}}peau, Richard Jozsa, Asher Peres, and William~K. Wootters.
\newblock Teleporting an unknown quantum state via dual classical and {Einstein-Podolsky-Rosen} channels.
\newblock {\em Physical Review Letters}, 70(13):1895, 1993.
\newblock \href {https://doi.org/10.1103/PhysRevLett.70.1895} {\path{doi:10.1103/PhysRevLett.70.1895}}.

\bibitem[BBPS96]{BBPS96}
Charles~H. Bennett, Herbert~J. Bernstein, Sandu Popescu, and Benjamin Schumacher.
\newblock Concentrating partial entanglement by local operations.
\newblock {\em Physical Review A}, 53(4):2046, 1996.
\newblock \href {https://doi.org/10.1103/PhysRevA.53.2046} {\path{doi:10.1103/PhysRevA.53.2046}}.

\bibitem[BCC{\etalchar{+}}15]{BCC+15}
Dominic~W. Berry, Andrew~M. Childs, Richard Cleve, Robin Kothari, and Rolando~D. Somma.
\newblock Simulating {Hamiltonian} dynamics with a truncated {Taylor} series.
\newblock {\em Physical Review Letters}, 114(9):090502, 2015.
\newblock \href {https://doi.org/10.1103/PhysRevLett.114.090502} {\path{doi:10.1103/PhysRevLett.114.090502}}.

\bibitem[BCO26]{bluhm2024hamiltonian}
Andreas Bluhm, Matthias~C. Caro, and Aadil Oufkir.
\newblock Hamiltonian property testing.
\newblock {\em Quantum}, 10:1979, 2026.
\newblock \href {https://doi.org/10.22331/q-2026-01-21-1979} {\path{doi:10.22331/q-2026-01-21-1979}}.

\bibitem[BHMT02]{BHMT02}
Gilles Brassard, Peter H{\o}yer, Michele Mosca, and Alain Tapp.
\newblock Quantum amplitude amplification and estimation.
\newblock In Samuel~J. Lomonaco, Jr. and Howard~E. Brandt, editors, {\em Quantum Computation and Information}, volume 305 of {\em Contemporary Mathematics}, pages 53--74. American Mathematical Society, 2002.
\newblock \href {https://doi.org/10.1090/conm/305/05215} {\path{doi:10.1090/conm/305/05215}}.

\bibitem[BOW19]{BOW19}
Costin B{\u{a}}descu, Ryan O'Donnell, and John Wright.
\newblock Quantum state certification.
\newblock In {\em Proceedings of the 51st Annual ACM SIGACT Symposium on Theory of Computing}, pages 503--514, 2019.
\newblock \href {https://doi.org/10.1145/3313276.3316344} {\path{doi:10.1145/3313276.3316344}}.

\bibitem[Bra11]{Bra11}
Sergey Bravyi.
\newblock Efficient algorithm for a quantum analogue of 2-{SAT}.
\newblock In Kazem Mahdavi, Deborah Koslover, and Leonard~L. Brown, III, editors, {\em Cross Disciplinary Advances in Quantum Computing}, volume 536 of {\em Contemporary Mathematics}, pages 33--48. American Mathematical Society, 2011.
\newblock \href {https://doi.org/10.1090/conm/536/10552} {\path{doi:10.1090/conm/536/10552}}.

\bibitem[BW92]{BJ92}
Charles~H. Bennett and Stephen~J. Wiesner.
\newblock Communication via one- and two-particle operators on {Einstein-Podolsky-Rosen} states.
\newblock {\em Physical Review Letters}, 69(20):2881, 1992.
\newblock \href {https://doi.org/10.1103/PhysRevLett.69.2881} {\path{doi:10.1103/PhysRevLett.69.2881}}.

\bibitem[CCHL22]{chen2022exponential}
Sitan Chen, Jordan Cotler, Hsin-Yuan Huang, and Jerry Li.
\newblock Exponential separations between learning with and without quantum memory.
\newblock In {\em Proceedings of the 62nd IEEE Annual Symposium on Foundations of Computer Science}, pages 574--585, 2022.
\newblock \href {https://doi.org/10.1109/FOCS52979.2021.00063} {\path{doi:10.1109/FOCS52979.2021.00063}}.

\bibitem[CCHS24]{CCHS24}
Nai-Hui Chia, Kai-Min Chung, Tzu-Hsiang Huang, and Jhih-Wei Shih.
\newblock Complexity theory for quantum promise problems.
\newblock ArXiv e-prints, 2024.
\newblock \href {https://arxiv.org/abs/2411.03716} {\path{arXiv:2411.03716}}.

\bibitem[CGJ19]{CGJ19}
Shantanav Chakraborty, Andr\'{a}s Gily\'{e}n, and Stacey Jeffery.
\newblock The power of block-encoded matrix powers: improved regression techniques via faster {Hamiltonian} simulation.
\newblock In {\em Proceedings of the 46th International Colloquium on Automata, Languages, and Programming}, volume 132 of {\em Leibniz International Proceedings in Informatics (LIPIcs)}, pages 33:1--33:14, 2019.
\newblock \href {https://doi.org/10.4230/LIPIcs.ICALP.2019.33} {\path{doi:10.4230/LIPIcs.ICALP.2019.33}}.

\bibitem[CHW07]{CHW07}
Andrew~M. Childs, Aram~W. Harrow, and Pawe{\l} Wocjan.
\newblock Weak {Fourier-Schur} sampling, the hidden subgroup problem, and the quantum collision problem.
\newblock In {\em Proceedings of the 24th Annual Symposium on Theoretical Aspects of Computer Science}, pages 598--609, 2007.
\newblock \href {https://doi.org/10.1007/978-3-540-70918-3_51} {\path{doi:10.1007/978-3-540-70918-3_51}}.

\bibitem[CL08]{CL08}
Pasquale Calabrese and Alexandre Lefevre.
\newblock Entanglement spectrum in one-dimensional systems.
\newblock {\em Physical Review A}, 78(3):032329, 2008.
\newblock \href {https://doi.org/10.1103/PhysRevA.78.032329} {\path{doi:10.1103/PhysRevA.78.032329}}.

\bibitem[CLM{\etalchar{+}}14]{chitambar2014everything}
Eric Chitambar, Debbie Leung, Laura Man{\v{c}}inska, Maris Ozols, and Andreas Winter.
\newblock Everything you always wanted to know about {LOCC} (but were afraid to ask).
\newblock {\em Communications in Mathematical Physics}, 328:303--326, 2014.
\newblock \href {https://doi.org/10.1007/s00220-014-1953-9} {\path{doi:10.1007/s00220-014-1953-9}}.

\bibitem[CW12]{CW12}
Andrew~M. Childs and Nathan Wiebe.
\newblock Hamiltonian simulation using linear combinations of unitary operations.
\newblock {\em Quantum Information and Computation}, 12(11--12):901--924, 2012.
\newblock \href {https://doi.org/10.26421/QIC12.11-12-1} {\path{doi:10.26421/QIC12.11-12-1}}.

\bibitem[CWLY23]{chen2023unitarity}
Kean Chen, Qisheng Wang, Peixun Long, and Mingsheng Ying.
\newblock Unitarity estimation for quantum channels.
\newblock {\em IEEE Transactions on Information Theory}, 69(8):5116--5134, 2023.
\newblock \href {https://doi.org/10.1109/TIT.2023.3263645} {\path{doi:10.1109/TIT.2023.3263645}}.

\bibitem[CWZ25]{CWZ25}
Kean Chen, Qisheng Wang, and Zhicheng Zhang.
\newblock A list of complexity bounds for property testing by quantum sample-to-query lifting.
\newblock ArXiv e-prints, 2025.
\newblock \href {https://arxiv.org/abs/2512.01971} {\path{arXiv:2512.01971}}.

\bibitem[CYZ25]{chen2025quantum}
Kean Chen, Nengkun Yu, and Zhicheng Zhang.
\newblock Quantum channel tomography and estimation by local test.
\newblock ArXiv e-prints, 2025.
\newblock \href {https://arxiv.org/abs/2512.13614} {\path{arXiv:2512.13614}}.

\bibitem[ECP10]{ECP10}
J.~Eisert, M.~Cramer, and M.~B. Plenio.
\newblock Colloquium: Area laws for the entanglement entropy.
\newblock {\em Reviews of Modern Physics}, 82(1):277–306, 2010.
\newblock \href {https://doi.org/10.1103/revmodphys.82.277} {\path{doi:10.1103/revmodphys.82.277}}.

\bibitem[EGH{\etalchar{+}}11]{etingof2011introduction}
Pavel Etingof, Oleg Golberg, Sebastian Hensel, Tiankai Liu, Alex Schwendner, Dmitry Vaintrob, and Elena Yudovina.
\newblock {\em Introduction to Representation Theory}, volume~59 of {\em Student Mathematical Library}.
\newblock American Mathematical Society, 2011.
\newblock \href {https://doi.org/10.1090/stml/059} {\path{doi:10.1090/stml/059}}.

\bibitem[EPR35]{EPR35}
A.~Einstein, B.~Podolsky, and N.~Rosen.
\newblock Can quantum-mechanical description of physical reality be considered complete?
\newblock {\em Physical Review}, 47(10):777, 1935.
\newblock \href {https://doi.org/10.1103/PhysRev.47.777} {\path{doi:10.1103/PhysRev.47.777}}.

\bibitem[EY36]{eckart1936approximation}
Carl Eckart and Gale Young.
\newblock The approximation of one matrix by another of lower rank.
\newblock {\em Psychometrika}, 1(3):211--218, 1936.
\newblock \href {https://doi.org/10.1007/BF02288367} {\path{doi:10.1007/BF02288367}}.

\bibitem[FFGO23]{fawzi2023quantum}
Omar Fawzi, Nicolas Flammarion, Aur{\'e}lien Garivier, and Aadil Oufkir.
\newblock Quantum channel certification with incoherent measurements.
\newblock In {\em Proceedings of the 36th Conference on Learning Theory}, pages 1822--1884, 2023.
\newblock URL: \url{https://proceedings.mlr.press/v195/fawzi23a.html}.

\bibitem[FH13]{fulton2013representation}
William Fulton and Joe Harris.
\newblock {\em Representation Theory: A First Course}, volume 129 of {\em Graduate Texts in Mathematics}.
\newblock Springer, 2013.
\newblock \href {https://doi.org/10.1007/978-1-4612-0979-9} {\path{doi:10.1007/978-1-4612-0979-9}}.

\bibitem[FvdG99]{FvdG99}
Christopher~A. Fuchs and Jeroen van~de Graaf.
\newblock Cryptographic distinguishability measures for quantum-mechanical states.
\newblock {\em IEEE Transactions on Information Theory}, 45(4):1216--1227, 1999.
\newblock \href {https://doi.org/10.1109/18.761271} {\path{doi:10.1109/18.761271}}.

\bibitem[GML25]{girardi2025random}
Filippo Girardi, Francesco~Anna Mele, and Ludovico Lami.
\newblock Random purification channel made simple.
\newblock ArXiv e-prints, 2025.
\newblock \href {https://arxiv.org/abs/2511.23451} {\path{arXiv:2511.23451}}.

\bibitem[GMZ{\etalchar{+}}25]{GMZFL25}
Filippo Girardi, Francesco~Anna Mele, Haimeng Zhao, Marco Fanizza, and Ludovico Lami.
\newblock Random stinespring superchannel: converting channel queries into dilation isometry queries.
\newblock ArXiv e-prints, 2025.
\newblock \href {https://arxiv.org/abs/2512.20599} {\path{arXiv:2512.20599}}.

\bibitem[GN16]{GN16}
David Gosset and Daniel Nagaj.
\newblock Quantum 3-{SAT} is $\mathsf{QMA}_1$-complete.
\newblock {\em SIAM Journal on Computing}, 45(3):1080--1128, 2016.
\newblock \href {https://doi.org/10.1137/140957056} {\path{doi:10.1137/140957056}}.

\bibitem[Gro96]{Gro96}
Lov~K. Grover.
\newblock A fast quantum mechanical algorithm for database search.
\newblock In {\em Proceedings of the 28th Annual ACM Symposium on Theory of Computing}, pages 212--219, 1996.
\newblock \href {https://doi.org/10.1145/237814.237866} {\path{doi:10.1145/237814.237866}}.

\bibitem[GSLW19]{GSLW19}
Andr\'{a}s Gily\'{e}n, Yuan Su, Guang~Hao Low, and Nathan Wiebe.
\newblock Quantum singular value transformation and beyond: exponential improvements for quantum matrix arithmetics.
\newblock In {\em Proceedings of the 51st Annual ACM SIGACT Symposium on Theory of Computing}, pages 193--204, 2019.
\newblock \href {https://doi.org/10.1145/3313276.3316366} {\path{doi:10.1145/3313276.3316366}}.

\bibitem[Har13]{harrow2013church}
Aram~W. Harrow.
\newblock The church of the symmetric subspace.
\newblock ArXiv e-prints, 2013.
\newblock \href {https://arxiv.org/abs/1308.6595} {\path{arXiv:1308.6595}}.

\bibitem[Has07]{hastings2007area}
M.~B. Hastings.
\newblock An area law for one-dimensional quantum systems.
\newblock {\em Journal of Statistical Mechanics: Theory and Experiment}, 2007(08):P08024, 2007.
\newblock \href {https://doi.org/10.1088/1742-5468/2007/08/P08024} {\path{doi:10.1088/1742-5468/2007/08/P08024}}.

\bibitem[Hel67]{Hel67}
Carl~W. Helstrom.
\newblock Detection theory and quantum mechanics.
\newblock {\em Information and Control}, 10(3):254--291, 1967.
\newblock \href {https://doi.org/10.1016/S0019-9958(67)90302-6} {\path{doi:10.1016/S0019-9958(67)90302-6}}.

\bibitem[HHHH09]{HHHH09}
Ryszard Horodecki, Pawe{\l} Horodecki, Micha{\l} Horodecki, and Karol Horodecki.
\newblock Quantum entanglement.
\newblock {\em Reviews of Modern Physics}, 81(2):865, 2009.
\newblock \href {https://doi.org/10.1103/RevModPhys.81.865} {\path{doi:10.1103/RevModPhys.81.865}}.

\bibitem[HM10]{HM10}
Aram~W. Harrow and Ashley Montanaro.
\newblock An efficient test for product states with applications to quantum {Merlin--Arthur} games.
\newblock In {\em Proceedings of the 51st IEEE Annual Symposium on Foundations of Computer Science}, pages 633--642, 2010.
\newblock \href {https://doi.org/10.1109/FOCS.2010.66} {\path{doi:10.1109/FOCS.2010.66}}.

\bibitem[HMT06]{hayashi2006study}
Masahito Hayashi, Keiji Matsumoto, and Yoshiyuki Tsuda.
\newblock A study of {LOCC}-detection of a maximally entangled state using hypothesis testing.
\newblock {\em Journal of Physics A: Mathematical and General}, 39(46):14427, 2006.
\newblock \href {https://doi.org/10.1088/0305-4470/39/46/013} {\path{doi:10.1088/0305-4470/39/46/013}}.

\bibitem[Hoe63]{hoeffding1994probability}
Wassily Hoeffding.
\newblock Probability inequalities for sums of bounded random variables.
\newblock {\em Journal of the American Statistical Association}, 58(301):13--30, 1963.
\newblock \href {https://doi.org/10.1080/01621459.1963.10500830} {\path{doi:10.1080/01621459.1963.10500830}}.

\bibitem[Hol73]{Hol73}
Alexander~S. Holevo.
\newblock Statistical decision theory for quantum systems.
\newblock {\em Journal of Multivariate Analysis}, 3(4):337--394, 1973.
\newblock \href {https://doi.org/10.1016/0047-259X(73)90028-6} {\path{doi:10.1016/0047-259X(73)90028-6}}.

\bibitem[How87]{howe1987gl}
Roger Howe.
\newblock $(\mathit{GL}_n, \mathit{GL}_m)$-duality and symmetric plethysm.
\newblock In {\em Proceedings of the Indian Academy of Sciences-Mathematical Sciences}, volume~97, pages 85--109, 1987.
\newblock \href {https://doi.org/10.1007/BF02837817} {\path{doi:10.1007/BF02837817}}.

\bibitem[JW24]{JW24}
Fernando~G. Jeronimo and Pei Wu.
\newblock Dimension independent disentanglers from unentanglement and applications.
\newblock In {\em Proceedings of the 39th Computational Complexity Conference}, pages 26:1--26:28, 2024.
\newblock \href {https://doi.org/10.4230/LIPIcs.CCC.2024.26} {\path{doi:10.4230/LIPIcs.CCC.2024.26}}.

\bibitem[Kir08]{kirillov2008introduction}
Alexander~A. Kirillov.
\newblock {\em An Introduction to Lie Groups and Lie Algebras}, volume 113 of {\em Cambridge Studies in Advanced Mathematics}.
\newblock Cambridge University Press, 2008.
\newblock \href {https://doi.org/10.1017/CBO9780511755156} {\path{doi:10.1017/CBO9780511755156}}.

\bibitem[KP06]{KP06}
Alexei Kitaev and John Preskill.
\newblock Topological entanglement entropy.
\newblock {\em Physical Review Letters}, 96(11):110404, 2006.
\newblock \href {https://doi.org/10.1103/PhysRevLett.96.110404} {\path{doi:10.1103/PhysRevLett.96.110404}}.

\bibitem[LGDC24]{LGDC24}
Zhenhuan Liu, Weiyuan Gong, Zhenyu Du, and Zhenyu Cai.
\newblock Exponential separations between quantum learning with and without purification.
\newblock ArXiv e-prints, 2024.
\newblock \href {https://arxiv.org/abs/2410.17718} {\path{arXiv:2410.17718}}.

\bibitem[LH08]{LH08}
Hui Li and F.~D.~M. Haldane.
\newblock Entanglement spectrum as a generalization of entanglement entropy: identification of topological order in {non-Abelian} fractional quantum {Hall} effect states.
\newblock {\em Physical Review Letters}, 101(1):010504, 2008.
\newblock \href {https://doi.org/10.1103/PhysRevLett.101.010504} {\path{doi:10.1103/PhysRevLett.101.010504}}.

\bibitem[LL26]{LL24}
Benjamin Lovitz and Angus Lowe.
\newblock Nearly tight bounds for testing tree tensor network states.
\newblock {\em IEEE Transactions on Information Theory}, 72(5):3074--3097, 2026.
\newblock \href {https://doi.org/10.1109/TIT.2026.3670320} {\path{doi:10.1109/TIT.2026.3670320}}.

\bibitem[LW06]{LW06}
Michael Levin and Xiao-Gang Wen.
\newblock Detecting topological order in a ground state wave function.
\newblock {\em Physical Review Letters}, 96(11):110405, 2006.
\newblock \href {https://doi.org/10.1103/PhysRevLett.96.110405} {\path{doi:10.1103/PhysRevLett.96.110405}}.

\bibitem[LW25]{liu2024separation}
Zhenhuan Liu and Fuchuan Wei.
\newblock Separation between entanglement criteria and entanglement detection protocols.
\newblock {\em Physical Review Research}, 7(3):033121, 2025.
\newblock \href {https://doi.org/10.1103/hm8j-wgqm} {\path{doi:10.1103/hm8j-wgqm}}.

\bibitem[Mat10]{matsumoto2010test}
Keiji Matsumoto.
\newblock Test of purity by {LOCC}.
\newblock ArXiv e-prints, 2010.
\newblock \href {https://arxiv.org/abs/1009.3121} {\path{arXiv:1009.3121}}.

\bibitem[MB25]{mele2025optimallearningquantumchannels}
Antonio~Anna Mele and Lennart Bittel.
\newblock Optimal learning of quantum channels in diamond distance.
\newblock ArXiv e-prints, 2025.
\newblock \href {https://arxiv.org/abs/2512.10214} {\path{arXiv:2512.10214}}.

\bibitem[MdW16]{MdW16}
Ashley Montanaro and Ronald de~Wolf.
\newblock A survey of quantum property testing.
\newblock In {\em Theory of Computing Library}, number~7 in Graduate Surveys, pages 1--81. University of Chicago, 2016.
\newblock \href {https://doi.org/10.4086/toc.gs.2016.007} {\path{doi:10.4086/toc.gs.2016.007}}.

\bibitem[MGC{\etalchar{+}}25]{mele2025random}
Francesco~Anna Mele, Filippo Girardi, Senrui Chen, Marco Fanizza, and Ludovico Lami.
\newblock Random purification channel for passive gaussian bosons.
\newblock ArXiv e-prints, 2025.
\newblock \href {https://arxiv.org/abs/2512.16878} {\path{arXiv:2512.16878}}.

\bibitem[MH07]{matsumoto2007universal}
Keiji Matsumoto and Masahito Hayashi.
\newblock Universal distortion-free entanglement concentration.
\newblock {\em Physical Review A}, 75(6):062338, 2007.
\newblock \href {https://doi.org/10.1103/PhysRevA.75.062338} {\path{doi:10.1103/PhysRevA.75.062338}}.

\bibitem[MKB05]{MKB05}
Florian Mintert, Marek Ku{\'s}, and Andreas Buchleitner.
\newblock Concurrence of mixed multipartite quantum states.
\newblock {\em Physical Review Letters}, 95(26):260502, 2005.
\newblock \href {https://doi.org/10.1103/PhysRevLett.95.260502} {\path{doi:10.1103/PhysRevLett.95.260502}}.

\bibitem[MNY24]{morimae2024unconditionally}
Tomoyuki Morimae, Barak Nehoran, and Takashi Yamakawa.
\newblock Unconditionally secure commitments with quantum auxiliary inputs.
\newblock In {\em Advances in Cryptology – CRYPTO 2024}, pages 59--92. Springer, 2024.
\newblock \href {https://doi.org/10.1007/978-3-031-68394-7_3} {\path{doi:10.1007/978-3-031-68394-7_3}}.

\bibitem[OW15]{OW15}
Ryan O'Donnell and John Wright.
\newblock Quantum spectrum testing.
\newblock In {\em Proceedings of the 47th Annual ACM Symposium on Theory of Computing}, pages 529--538, 2015.
\newblock \href {https://doi.org/10.1145/2746539.2746582} {\path{doi:10.1145/2746539.2746582}}.

\bibitem[PGVWC07]{PGVWC07}
David Perez-Garcia, Frank Verstraete, Michael~M. Wolf, and J.~Ignacio Cirac.
\newblock Matrix product state representations.
\newblock {\em Quantum Information and Computation}, 7(5--6):401--430, 2007.
\newblock \href {https://doi.org/10.26421/QIC7.5-6-1} {\path{doi:10.26421/QIC7.5-6-1}}.

\bibitem[PLM18]{pallister2018optimal}
Sam Pallister, Noah Linden, and Ashley Montanaro.
\newblock Optimal verification of entangled states with local measurements.
\newblock {\em Physical Review Letters}, 120(17):170502, 2018.
\newblock \href {https://doi.org/10.1103/PhysRevLett.120.170502} {\path{doi:10.1103/PhysRevLett.120.170502}}.

\bibitem[PSTW25]{pelecanos2025mixedstatetomographyreduces}
Angelos Pelecanos, Jack Spilecki, Ewin Tang, and John Wright.
\newblock Mixed state tomography reduces to pure state tomography.
\newblock ArXiv e-prints, 2025.
\newblock \href {https://arxiv.org/abs/2511.15806} {\path{arXiv:2511.15806}}.

\bibitem[PTBO10]{PTBO10}
Frank Pollmann, Ari~M. Turner, Erez Berg, and Masaki Oshikawa.
\newblock Entanglement spectrum of a topological phase in one dimension.
\newblock {\em Physical Review B}, 81(6):064439, 2010.
\newblock \href {https://doi.org/10.1103/PhysRevB.81.064439} {\path{doi:10.1103/PhysRevB.81.064439}}.

\bibitem[Ser77]{serre1977linear}
Jean-Pierre Serre.
\newblock {\em Linear Representations of Finite Groups}, volume~42 of {\em Graduate Texts in Mathematics}.
\newblock Springer, 1977.
\newblock \href {https://doi.org/10.1007/978-1-4684-9458-7} {\path{doi:10.1007/978-1-4684-9458-7}}.

\bibitem[SW22]{soleimanifar2022testing}
Mehdi Soleimanifar and John Wright.
\newblock Testing matrix product states.
\newblock In {\em Proceedings of the 2022 Annual ACM-SIAM Symposium on Discrete Algorithms}, pages 1679--1701, 2022.
\newblock \href {https://doi.org/10.1137/1.9781611977073.68} {\path{doi:10.1137/1.9781611977073.68}}.

\bibitem[SY23]{SY23}
Adrian She and Henry Yuen.
\newblock Unitary property testing lower bounds by polynomials.
\newblock In {\em Proceedings of the 14th Innovations in Theoretical Computer Science Conference}, pages 96:1--96:17, 2023.
\newblock \href {https://doi.org/10.4230/LIPIcs.ITCS.2023.96} {\path{doi:10.4230/LIPIcs.ITCS.2023.96}}.

\bibitem[Tsa88]{Tsa88}
Constantino Tsallis.
\newblock Possible generalization of {Boltzmann-Gibbs} statistics.
\newblock {\em Journal of Statistical Physics}, 52(2):479--487, 1988.
\newblock \href {https://doi.org/10.1007/BF01016429} {\path{doi:10.1007/BF01016429}}.

\bibitem[TWZ25]{TWZ25}
Ewin Tang, John Wright, and Mark Zhandry.
\newblock Conjugate queries can help.
\newblock ArXiv preprints, 2025.
\newblock \href {https://arxiv.org/abs/2510.07622} {\path{arXiv:2510.07622}}.

\bibitem[Wat18]{watrous2018theory}
John Watrous.
\newblock {\em The Theory of Quantum Information}.
\newblock Cambridge University Press, 2018.
\newblock \href {https://doi.org/10.1017/9781316848142} {\path{doi:10.1017/9781316848142}}.

\bibitem[Weg25]{Weg24}
Jordi Weggemans.
\newblock Lower bounds for unitary property testing with proofs and advice.
\newblock {\em Quantum}, 9:1717, 2025.
\newblock \href {https://doi.org/10.22331/q-2025-04-18-1717} {\path{doi:10.22331/q-2025-04-18-1717}}.

\bibitem[WH19]{wang2019optimal}
Kun Wang and Masahito Hayashi.
\newblock Optimal verification of two-qubit pure states.
\newblock {\em Physical Review A}, 100(3):032315, 2019.
\newblock \href {https://doi.org/10.1103/PhysRevA.100.032315} {\path{doi:10.1103/PhysRevA.100.032315}}.

\bibitem[Wil13]{Wil13}
Mark~M. Wilde.
\newblock {\em Quantum Information Theory}.
\newblock Cambridge University Press, 2013.
\newblock \href {https://doi.org/10.1017/CBO9781139525343} {\path{doi:10.1017/CBO9781139525343}}.

\bibitem[WW25]{walter2025random}
Michael Walter and Freek Witteveen.
\newblock A random purification channel for arbitrary symmetries with applications to fermions and bosons.
\newblock ArXiv e-prints, 2025.
\newblock \href {https://arxiv.org/abs/2512.15690} {\path{arXiv:2512.15690}}.

\bibitem[WZ25a]{WZ23}
Qisheng Wang and Zhicheng Zhang.
\newblock Quantum lower bounds by sample-to-query lifting.
\newblock {\em SIAM Journal on Computing}, 54(5):1294--1334, 2025.
\newblock \href {https://doi.org/10.1137/24M1638616} {\path{doi:10.1137/24M1638616}}.

\bibitem[WZ25b]{WZ24}
Qisheng Wang and Zhicheng Zhang.
\newblock Time-efficient quantum entropy estimator via samplizer.
\newblock {\em IEEE Transactions on Information Theory}, 71(12):9569--9599, 2025.
\newblock \href {https://doi.org/10.1109/TIT.2025.3576137} {\path{doi:10.1109/TIT.2025.3576137}}.

\bibitem[YNU{\etalchar{+}}25]{yoshida2025random}
Satoshi Yoshida, Ryotaro Niwa, Takeru Utsumi, Ryuji Takagi, and Mio Murao.
\newblock Random dilation superchannel.
\newblock ArXiv e-prints, 2025.
\newblock \href {https://arxiv.org/abs/2512.21260} {\path{arXiv:2512.21260}}.

\bibitem[YQ10]{YQ10}
Hong Yao and Xiao-Liang Qi.
\newblock Entanglement entropy and entanglement spectrum of the {Kitaev} model.
\newblock {\em Physical Review Letters}, 105(8):080501, 2010.
\newblock \href {https://doi.org/10.1103/PhysRevLett.105.080501} {\path{doi:10.1103/PhysRevLett.105.080501}}.

\bibitem[ZH19]{zhu2019efficient}
Huangjun Zhu and Masahito Hayashi.
\newblock Efficient verification of pure quantum states in the adversarial scenario.
\newblock {\em Physical Review Letters}, 123(26):260504, 2019.
\newblock \href {https://doi.org/10.1103/PhysRevLett.123.260504} {\path{doi:10.1103/PhysRevLett.123.260504}}.

\end{thebibliography}

\appendix

\section{Sample Lower Bound for Purity Testing}

In this section, we show a matching sample lower bound for testing the purity of mixed states. 
This property is defined as follows. 

\begin{itemize}
    \item Whether the mixed state $\rho$ in Hilbert space $\mathcal{H}$ is pure. Formally, we define
    \begin{equation*}
        \textsc{Purity} = \set{\rho \in \mathcal{D}\rbra{\mathcal{H}}}{\tr\sbra{\rho^2} = 1}.
    \end{equation*}
\end{itemize}

To prove a sample lower bound for purity testing, we need an upper bound for the success probability of quantum state discrimination, known as the Helstrom-Holevo bound \cite{Hel67,Hol73}.
We use the version given in \cite{Wil13}.

\begin{theorem} [Quantum state discrimination, cf.~{\cite[Section 9.1.4]{Wil13}}]
\label{thm:HH-bound}
    Suppose that $\rho_0$ and $\rho_1$ are two mixed quantum states. 
    Let $\varrho$ be a quantum state such that $\varrho = \rho_0$ or $\varrho = \rho_1$ with equal probability. 
    For any POVM $\Lambda = \cbra{\Lambda_0, \Lambda_1}$,
    the success probability of distinguishing the two cases is bounded by
    \begin{equation*}
    p_{\mathrm{succ}} = \frac 1 2 \tr\sbra*{\Lambda_0\rho_0} + \frac 1 2 \tr\sbra*{\Lambda_1\rho_1} \leq \frac{1}{2}\rbra*{1 + d_{\tr} \rbra{\rho_0,\rho_1}}. 
    \end{equation*}
    Moreover, the equality holds for some POVM $\Lambda^*$.
\end{theorem}

Next, we prove a lower bound for purity testing through a reduction of quantum state discrimination. 

\begin{theorem} \label{thm:purity}
    Testing whether a mixed state is pure or $\varepsilon$-far (in trace distance) requires sample complexity $\mathsf{S}\rbra{\textsc{Purity}_\varepsilon} = \Omega\rbra{1/\varepsilon}$. 
\end{theorem}
\begin{proof}
    For $\varepsilon \in \rbra{0, 1}$, we consider the problem of distinguishing two mixed states $\rho_0$ and $\rho_\varepsilon$, where 
    \begin{equation*}
    \rho_{x} = \rbra{1-x} \ketbra{0}{0} + x \ketbra{1}{1}.
    \end{equation*}
    Note that $\rho_0$ is a pure state and $\rho_\varepsilon$ is $\varepsilon$-far (in trace distance) from $\rho_0$.
    
    Suppose there is a tester for purity with sample complexity $S$, then the tester can be used to distinguish between $\rho_0^{\otimes S}$ and $\rho_\varepsilon^{\otimes S}$ with success probability $p_{\textup{succ}} \geq 2/3$. 
    On the other hand, by \cref{thm:HH-bound}, we have
    \begin{equation} \label{eq:p_succ}
        p_{\textup{succ}} \leq \frac{1}{2}\rbra*{1+ d_{\tr} \rbra{\rho_0^{\otimes S},\rho_\varepsilon^{\otimes S}}}.
    \end{equation}
    By the Fuchs–van de Graaf inequalities (\cref{lemma:fi-td}), we have 
    \begin{equation} \label{eq:FvdG}
        d_{\tr} \rbra{\rho_0^{\otimes S},\rho_\varepsilon^{\otimes S}} \leq \sqrt{1 - \mathrm{F}\rbra*{\rho_0^{\otimes S}, \rho_\varepsilon^{\otimes S}}^2},
    \end{equation}
    where $\mathrm{F}\rbra{\rho, \sigma} = \tr\sbra{\sqrt{\sqrt{\sigma}\rho\sqrt{\sigma}}}$ is the fidelity between two mixed states.
    A simple calculation shows that
    \begin{equation} \label{eq:fidelity}
        \mathrm{F}\rbra*{\rho_0^{\otimes S}, \rho_\varepsilon^{\otimes S}}^2 = \mathrm{F}\rbra*{\rho_0, \rho_\varepsilon}^{2S} = \rbra{1 - \varepsilon}^S.
    \end{equation}
    By Equations \eqref{eq:p_succ}, \eqref{eq:FvdG} and \eqref{eq:fidelity}, we have
    \begin{equation*}
        p_{\textup{succ}} \leq \frac 1 2 \rbra*{1 + \sqrt{1 - \rbra*{1 - \varepsilon}^S}},
    \end{equation*}
    which, together with $p_{\textup{succ}} \geq 2/3$, gives $S = \Omega\rbra{1/\varepsilon}$. 
\end{proof}

The lower bound in \cref{thm:purity} matches the upper bound in \cite[Section 4.2]{MdW16}, which means that $\mathsf{S}\rbra{\textsc{Purity}_\varepsilon} = \Theta\rbra{1/\varepsilon}$.

\end{document}